\documentclass[pre, twocolumn, superscriptaddress, amsmath, amssymb, longbibliography, floatfix]{revtex4-2}
\usepackage[utf8]{inputenc}
\usepackage[T1]{fontenc}
\usepackage{mathtools}
\usepackage{hyperref}
\usepackage{booktabs}
\usepackage{bm}
\usepackage{xcolor}
\usepackage{tikz}
\usetikzlibrary{arrows.meta,positioning,shapes.geometric}

\newcommand{\W}{\mathcal{W}}
\newcommand{\D}{\mathcal{D}}
\newcommand{\Lg}{\mathcal{L}}
\newcommand{\C}{\mathcal{C}}            
\newcommand{\Cd}{\mathcal{C}_\delta}    
\newcommand{\J}{\mathcal{J}}            
\newcommand{\Mr}{M}                     
\newcommand{\Lop}{\mathcal{L}}           
\newcommand{\F}{\mathcal{F}}            
\newcommand{\Da}{\mathrm{Da}}
\newcommand{\Kn}{\mathrm{Kn}}

\newcommand{\pder}[2]{\frac{\partial #1}{\partial #2}}
\newcommand{\eq}{\mathrm{eq}}

\providecommand{\A}{\mathcal{A}}

\newtheorem{prop}{Proposition}
\newtheorem{rem}{Remark}

\begin{document}

\title{Richards' equation as a hydrodynamic limit:\\
Chapman--Enskog reduction of the continuum kinetic equation\\
for unsaturated soil water}

\author{Riccardo Rigon}
\email{riccardo.rigon@unitn.it}
\affiliation{Centro Agricoltura Alimenti Ambiente (C3A), University of Trento, 38098 San Michele all'Adige (TN), Italy}

\date{September 3, 2026}

\begin{abstract}
Richards' equation for unsaturated water flow is derived from the kinetic
theory of the pore-filling distribution $g(r,\mathbf{x},t)\in[0,1]$
developed in a companion paper~\cite{Rigon2026PRE}.  The derivation
separates two limits that the macroscopic theory ordinarily conflates.
The \emph{spatial} limit $\varepsilon=\bar L/\Lambda\to0$ contracts the
representative elementary volume (REV) to a point and is purely kinematic:
it produces the closed \emph{continuum kinetic equation} (CKE)
$\partial_t g+\nabla\!\cdot\!\mathbf{F}=\C[g]$, with $\mathbf{F}$ an
entangled, pre-closure pore-resolved flux and $\C[g]$ the occupancy-gated
intra-REV redistribution operator; this limit is treated
in~\cite{Rigon2026PRE}.  The dynamics is hidden in the \emph{temporal}
limit, which is the subject of this paper: a Chapman--Enskog (CE)
reduction of the CKE, controlled by the Damk\"ohler number
$\Da=\tau_{\rm redis}/\tau_{\rm forcing}$, the ratio of the pore-scale
redistribution time to the macroscopic forcing time.  The small-$\Da$
limit is the structural analogue of the passage from the Boltzmann
equation to Navier--Stokes.  Following Cercignani's treatment of the
linearized collision operator, we show that the linearized redistribution
operator $\J=\C'|_{g_\eq}$ is self-adjoint and negative semidefinite in
the mass inner product $\langle u,v\rangle_f=\int u\,v\,f\,dr$, with a
one-dimensional kernel that fixes a single invariant, water, and hence a
single macroscopic equation; the CE hierarchy inverts the same operator
$\J$ at every order, only the source changing.  Four results follow:
(1)~the equilibrium step $g^{(0)}=H(r^*\!-r)$ defines the retention curve;
(2)~the linearized water budget plus the inter-REV source set the
first-order equation; (3)~its solvability \emph{is} mass conservation,
Richards' equation; (4)~the response function $\bm{\chi}$ yields the
macroscopic flux and identifies the hydraulic conductivity
$K=\phi\,\langle\kappa\,|\,\J^{-1}|\,S\rangle_f$ as a first-order transport
coefficient, the structural counterpart of viscosity in the kinetic theory
of gases.  The mean-field reduction recovers the standard integral formula;
the serial-path correction produces Mualem's heterogeneity penalty
$\exp(-4\sigma^2)$.  A bimodal pore-size distribution opens a gap in the
relaxation spectrum; projecting the CKE onto the resulting bands and
applying CE within each, while the cross-band relaxation survives at
$O(1)$, derives the dual- and multiple-permeability models (Gerke--van
Genuchten, Weiler's IN3M, mobile--immobile) from first principles, with an
exchange coefficient fixed by the inter-band connectivity.  When $\Da$ is not small the one-mode expansion breaks
down; retaining the slow spectral modes explicitly yields a hierarchy
of finite multi-band models that interpolates between Richards'
equation and the full CKE, which is recovered as the endpoint of the
hierarchy.
\end{abstract}

\maketitle

\section{Introduction}
\label{sec:intro}
Richards' equation~\cite{Richards1931,TubiniRigon2022},
\begin{equation}
\pder{\theta}{t}=\nabla\!\cdot\!\big[K(\psi)\,\nabla\psi\big]
  +\pder{K(\psi)}{z},
\label{eq:richards_intro}
\end{equation}
has been the standard model for
unsaturated water flow for nearly a century.  It evolves the water
content $\theta(\mathbf{x},t)$ [-] through a matric potential $\psi$ [m] and a
hydraulic conductivity $K(\psi)$ [m/s] (with $z$ the vertical coordinate,
positive upward), both treated as material properties
fixed by empirical constitutive
relations~\cite{vanGenuchten1980,Brooks1964,Mualem1976}.  Yet the
equation is \emph{postulated} from macroscopic mass conservation and the
Darcy--Buckingham flux law~\cite{Darcy1856,Buckingham1907}, not derived
from a microscopic theory, and its domain of validity is left
unspecified.

Because the derivation below moves between scales, we fix the vocabulary
at the outset.  The \emph{microscale} is that of individual pores and
grains (radius $r$, typically $\mu$m to mm): here water either occupies
a pore or does not, and motion follows the Stokes/Hagen--Poiseuille laws
of viscous flow.  The \emph{mesoscale}, or representative elementary
volume (REV)~\cite{Bear1972}, is a volume large enough, in practice several hundreds of pore
diameters (side
$\bar L\sim$\,mm--cm) to contain a statistically meaningful population of
pores, yet small compared with the macroscopic flow domain.  Over the
REV the pore-by-pore description is replaced by a \emph{distribution}:
instead of asking which individual pore is wet, one asks what
\emph{fraction} of pores of each radius is wet.  The \emph{continuum}
(or macroscale, side $\Lambda\sim$\,cm--m) is that of the field
problem, where $\theta$, $\psi$ and $K$ are smooth fields.  Classical
soil physics works directly at the macroscale; the kinetic theory works
at the mesoscale and reaches the macroscale by a controlled limit
$\bar L/\Lambda\to0$.  The two small parameters that organize the whole
construction, the spatial ratio $\varepsilon=\bar L/\Lambda$ and the
temporal (Damk\"ohler) ratio $\Da$,are defined precisely in
Secs.~\ref{sec:continuum} and~\ref{sec:Da}.

Richards' equation can fail. Under intense rainfall water bypasses fine pores and
follows preferential pathways~\cite{Beven1982,Germann2018}; laboratory
and field conductivities differ by factors of
$10$--$100$~\cite{Diamantopoulos2012}; capillary hysteresis requires ad
hoc closure~\cite{Hilfer2006}; dynamic non-equilibrium
effects~\cite{Hassanizadeh2002} violate the assumption that the
capillary-pressure--saturation relation is unique and instantaneous.
All these failures share one origin: Richards' equation assumes the
internal pore-scale configuration is permanently at local equilibrium,
so that $\theta$ (or $\psi$) is a sufficient state variable.

In a companion paper~\cite{Rigon2026PRE} we develop a kinetic theory
whose state variable is the \textbf{filling distribution}
$g(r,\mathbf{x},t)\in[0,1]$, the fraction of pores of radius $r$ filled
with water at $(\mathbf{x},t)$.  That paper carries the construction up
to the closed \emph{continuum kinetic equation} (CKE)
\begin{equation}
\partial_t g(r,\mathbf{x},t)+\nabla\!\cdot\!\mathbf{F}(r,\mathbf{x},t)
  =\C[g],
\label{eq:contKE}
\end{equation}
Here $g(r,\mathbf{x},t)\in[0,1]$ is the (dimensionless) filling
distribution introduced above; $\C[g]$ is the \emph{redistribution
operator} [s$^{-1}$], which moves water between pore classes inside a REV
and is defined in Sec.~\ref{sec:continuum}; and $\mathbf{F}(r,\mathbf{x},t)$
[m\,s$^{-1}$] is the pore-resolved water flux of class $r$, a velocity
(volume of water of class $r$ crossing unit area per unit time, per unit
of the measure $f\,dr$).  At this stage $\mathbf{F}$ is not yet Darcy's
law: it is the \emph{entangled}, pre-closure flux
\begin{equation}
\mathbf{F}(r,\mathbf{x},t)=-\int_0^\infty\Gamma(r,r';g)\,
  \nabla g(r',\mathbf{x},t)\,dr',
\label{eq:F_intro}
\end{equation}
in which the flux of class $r$ is driven by the gradients of \emph{all}
connected classes $r'$ through the kernel $\Gamma(r,r';g)$, so that
transport coefficient and driving gradient are not yet separated
(the kernel is constructed in Sec.~\ref{sec:gradflux},
Eq.~\eqref{eq:Gamma}).  Equation~\eqref{eq:contKE} is closed by boundary
and initial data: an initial filling $g(r,\mathbf{x},0)$, and on the
domain boundary $\partial\Omega$ either a prescribed flux
$\mathbf{q}\!\cdot\!\hat n=-(\text{rainfall}-\text{evaporation})$
(Neumann, the usual atmospheric forcing) or a prescribed potential
(Dirichlet, e.g.\ a water table), with the incoming surface flux
apportioned among classes in proportion to their abundance
$f(r)$~\cite{Rigon2026PRE}.  These conditions are inherited unchanged by
Richards' equation at the end of the reduction.

The companion paper~\cite{Rigon2026PRE} develops Eq. \ref{eq:contKE}
the soil-water analogue of the Boltzmann equation, and states, without
the supporting algebra, that Richards' equation is its quasi-static
limit.  The present paper supplies that algebra and answers the
question the companion raises.

The route from~\eqref{eq:contKE} to Richards' equation passes through
two \emph{independent} limits that the  literature routinely
merges.

\textit{(i) The spatial limit} $\varepsilon=\bar L/\Lambda\to0$
contracts the REV to a point.  It is purely kinematic: it does not
assume anything about the processes timing. 
In the companion paper this limit is taken at the level of
the result; its gradient-expansion algebra is supplied here
(Sec.~\ref{sec:continuum}, App.~\ref{app:gradient}).

\textit{(ii) The temporal limit} $\Da\to0$ is the CE
reduction of Eq.~\eqref{eq:contKE}.  It assumes the intra-REV (local)
redistribution is fast compared to the macroscopic forcing, slaves
$g$ to a one-parameter equilibrium manifold, and separates
$\mathbf{F}$ into a scalar conductivity times a potential gradient.
This is the soil-water counterpart of the Boltzmann to
Navier--Stokes passage~\cite{ChapmanCowling1970,Hilbert1912,Enskog1917,
Cercignani1988}.

The CE reduction rests entirely on the spectral properties of a
linearized redistribution operator $\J=\C'|_{g_\eq}$ where  $\C$ is the
redistribution operator of Eq.~\eqref{eq:contKE}, the (nonlinear)
integral operator that moves water between pore classes within a REV and
relaxes the filling distribution toward equilibrium (its explicit
gain--loss form is Eqs.~\eqref{eq:master}--\eqref{eq:L} below); $\J$ it has
dimensions of an inverse time, a rate of change of the dimensionless
$g$.  The prime denotes the \emph{functional} (Fr\'echet) derivative of
$\C$ with respect to $g$, evaluated at the local-equilibrium
distribution $g_\eq$: writing $g=g_\eq+\delta g$, the leading response
of $\C$ to a small departure $\delta g$ is 
\begin{equation} 
\J[\delta g]=\C'|_{g_\eq}
[\delta g],
\end{equation}
 so $\J$ is the linear operator governing how fast small
deviations from equilibrium relax.  
We establish
these by direct analogy with the linearized Boltzmann collision
operator $L$~\cite[Ch.~IV]{Cercignani1988}: $\J$ is self-adjoint and
non-positive in the mass inner product, its kernel is spanned by the
single redistribution invariant (water mass), and the CE hierarchy
reduces, at every order, to inverting the \emph{same} operator against a
source built from lower orders~\cite[Eqs.~(IV.2.4),~(V.3.10)]{Cercignani1988}.
Where the gas has five invariants and hence five macroscopic equations,
viscous pore flow has one, water mass, and hence the one macroscopic
equation obtained, as the Fredholm solvability condition, in this
paper.

Figure~\ref{fig:flowchart} charts the derivation, with the two limits
kept separate.
Section~\ref{sec:continuum} recalls the kinetic equation and carries out
the spatial limit, deriving the entangled flux and
Eq.~\eqref{eq:contKE}.  Section~\ref{sec:scale} establishes the scale
separation and the Hilbert-versus-CE distinction.
Section~\ref{sec:operator} develops the linearized operator $\J$ and its
spectrum.  Sections~\ref{sec:hierarchy}--\ref{sec:K} carry out the CE
hierarchy: equilibrium, first order, solvability${}={}$mass
conservation, and the response function that defines $K$.
Section~\ref{sec:forces} extends $\Phi$, the driving potential, to gravity, adsorption and
osmosis.  Section~\ref{sec:Richards} assembles Richards' equation.
Section~\ref{sec:classical} records the classical-limit recovery
(retention curves, $K_{\rm sat}$, $K(\theta)$, percolation exponents,
field capacity).  Section~\ref{sec:multiperm} derives the dual- and
multiple-permeability models as a multi-band reduction.
Section~\ref{sec:breakdown} characterizes the breakdown,
Section~\ref{sec:geometry} the geometric picture, and
Section~\ref{sec:synthesis} synthesizes.  Appendices collect the
gradient-expansion algebra (App.~\ref{app:gradient}), the spectral
proofs (App.~\ref{app:spectral}), the mean-field/serial-path
resummation (App.~\ref{app:chi}), the variational principle for $K$
the band-projection algebra
(App.~\ref{app:bands}), the extensions of $\Phi$ (gravitational, osmotic and adsorptive terms;
companion paper~\cite{Rigon2026PRE}) and the $\varepsilon$--$\Da$ relation
(Sec.~\ref{sec:scale}).

The Supplemental Material~\cite{SM_hilbert} reports two parameter-free
numerical tests of the redistribution spectrum on which the reduction
rests: a direct diagonalization of the pore-class operator $\J$,
confirming the radius labeling of its eigenmodes and the spectral sum
rule, and a three-dimensional pore-network experiment showing that the
spectral gap and the Stokes permeability vanish together at the
percolation threshold, where the expansion fails.
\begin{figure*}[t]
\centering
\begin{tikzpicture}[
  cebox/.style={draw, rounded corners, minimum width=6.6cm,
    text width=6.2cm, minimum height=0.95cm, align=center, font=\small, fill=blue!8},
  spbox/.style={draw, rounded corners, minimum width=6.6cm,
    text width=6.2cm, minimum height=0.95cm, align=center, font=\small, fill=cyan!8},
  resbox/.style={draw, rounded corners, minimum width=6.6cm,
    text width=6.2cm, minimum height=0.95cm, align=center, font=\small, fill=green!8},
  supbox/.style={draw, rounded corners, minimum width=4.4cm,
    minimum height=0.8cm, align=center, font=\small, fill=orange!8},
  arr/.style={-{Stealth[length=2.5mm]}, thick},
  darr/.style={-{Stealth[length=2.5mm]}, thick, dashed, gray},
  node distance=0.5cm]
\node[cebox] (KE) {Master equation~\eqref{eq:master}:\;
  $\partial_t g=\C[g]$ at one REV};
\node[spbox, below=of KE] (CL) {\textbf{Spatial limit} $\varepsilon\to0$
  (kinematic):\\
  gradient expansion $\Rightarrow$
  $\partial_t g+\nabla\!\cdot\!\mathbf{F}=\C[g]$ \;[\S\,\ref{sec:continuum}]};
\node[cebox, below=of CL] (S1) {\textbf{CE 0:} equilibrium
  $g^{(0)}=H(r^*\!-r)$, $\psi(\theta)$, $\mathcal{S}_\eq$
  \;[\S\,\ref{sec:step1}]};
\node[cebox, below=of S1] (S2) {\textbf{CE 1:} budget $\J$ + source $S$\\
  $\J[g^{(1)}]=\partial_t g^{(0)}+\nabla\!\cdot\!\mathbf{F}[g^{(0)}]$
  \;[\S\,\ref{sec:step2}]};
\node[cebox, below=of S2] (S3) {\textbf{Solvability} = mass conservation\\
  $\partial_t\theta+\nabla\!\cdot\!\mathbf{q}=-S_{\rm e}-S_{\rm r}$
  \;[\S\,\ref{sec:solvability}]};
\node[cebox, below=of S3] (S4) {\textbf{Response} $\bm\chi$ $\to$
  recognize $\mathbf{q}$ $\to$ $K$\\
  $K=\phi\langle\kappa|\J^{-1}|S\rangle_f$ \;[\S\,\ref{sec:K}]};
\node[resbox, below=of S4] (RE) {\textbf{Richards' equation}\\
  $\partial_t\theta=\nabla\!\cdot\![K(\nabla\psi+\hat z)]-S_{\rm e}-S_{\rm r}$};
\node[supbox, right=1.5cm of CL] (GE) {Gradient algebra\\
  \;[App.~\ref{app:gradient}]};
\node[supbox, right=1.5cm of S2] (SP) {Spectral analysis\\
  self-adjoint $\J$ \;[App.~\ref{app:spectral}]};
\node[supbox, right=1.5cm of RE] (MB) {Bands $\Rightarrow$
  dual-perm.\\ \;[\S\,\ref{sec:multiperm}, App.~\ref{app:bands}]};
\draw[arr] (KE)--(CL) node[midway,right,font=\scriptsize]{$\varepsilon\to0$};
\draw[arr] (CL)--(S1) node[midway,right,font=\scriptsize]{$\Da\ll1$, $O(\Da^{-1})$};
\draw[arr] (S1)--(S2) node[midway,right,font=\scriptsize]{$O(\Da^0)$};
\draw[arr] (S2)--(S3) node[midway,right,font=\scriptsize]{$\int(\cdot)f\,dr=0$};
\draw[arr] (S3)--(S4) node[midway,right,font=\scriptsize]{invert $\J$};
\draw[arr] (S4)--(RE) node[midway,right,font=\scriptsize]{substitute};
\draw[darr] (GE)--(CL);
\draw[darr] (SP)--(S2);
\draw[darr] (MB)--(RE);
\end{tikzpicture}
\caption{The derivation, with the two limits kept separate.  The
\emph{spatial} limit (cyan) is kinematic and produces the continuum
kinetic equation; the \emph{Chapman--Enskog} steps (blue) reduce it to
Richards' equation (green).  Solid arrows: logical chain; dashed:
supporting analysis.  Finite-$\Da$ band structure branches to
dual-/multiple-permeability.}
\label{fig:flowchart}
\end{figure*}
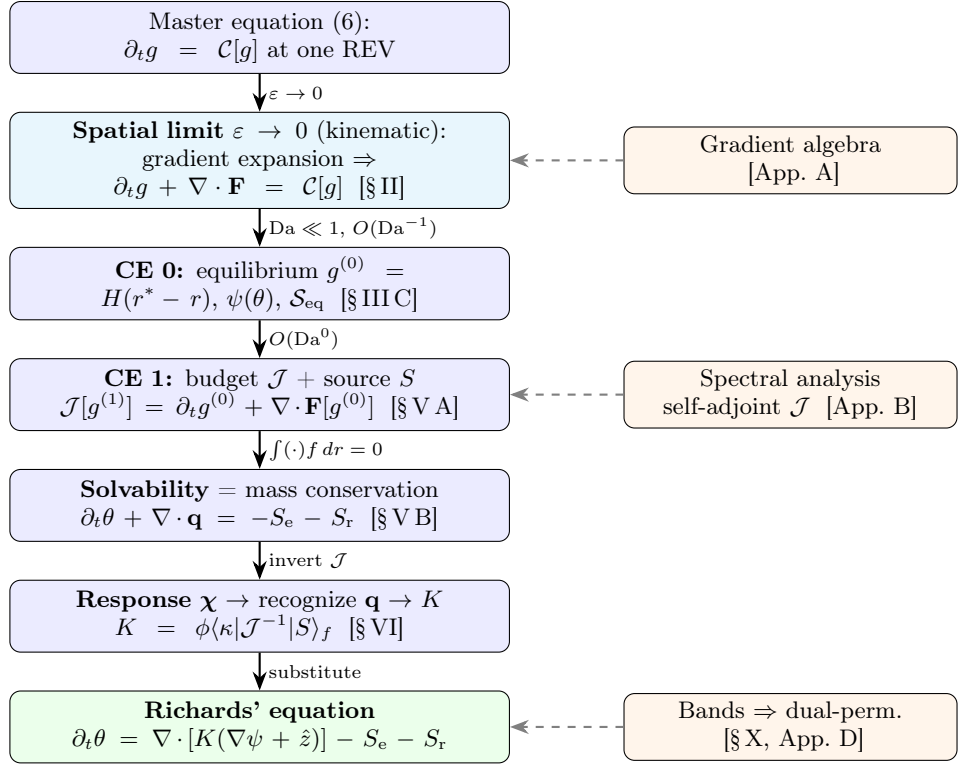

\section{From the kinetic equation to the continuum balance:
the spatial limit}
\label{sec:continuum}

Equation~\eqref{eq:F_intro} was derived thoroughly in the companion
paper~\cite{Rigon2026PRE}; here we summarize the passage from the
mesoscale to the continuum for the reader's sake, as a preliminary to
the expansion that follows.

\subsection{The kinetic equation at a single REV}

At each REV~\cite{Bear1972} the state is the filling distribution
$g(r,\mathbf{x},t)\in[0,1]$, and the water content is its zeroth moment,
\begin{equation}
\theta(\mathbf{x},t)=\phi\int_0^\infty g(r,\mathbf{x},t)\,f(r)\,dr,
\label{eq:theta}
\end{equation}
with $f(r)$ the pore-size distribution and $\phi$ the porosity.
Within one REV the occupancy obeys the master equation of the companion
paper~\cite{Rigon2026PRE},
\begin{equation}
\pder{g}{t}=\underbrace{\C[g]}_{\text{intra-REV}}-\mathcal{E}-\mathcal{T},
\qquad \C[g]:=\mathcal{G}_{\rm int}-\Lg_{\rm int},
\label{eq:master}
\end{equation}
with the gain and loss operators
\begin{widetext}
\begin{align}
\mathcal{G}_{\rm int}(r)&=\kappa(r,g)\,[1-g(r)]
  \!\int_0^\infty\! C(r,r')\,\rho_w\mathrm{g}\,\Phi(r,r')\,g(r')\,f(r')\,dr',
\label{eq:G}\\
\Lg_{\rm int}(r)&=\kappa(r,g)\,g(r)
  \!\int_0^\infty\! C(r,r')\,\rho_w\mathrm{g}\,\Phi(r',r)\,[1-g(r')]\,f(r')\,dr',
\label{eq:L}
\end{align}
\end{widetext}
where $C(r,r')\ge0$ is the connectivity, $\Phi(r,r')$ the antisymmetric
inter-pore driving potential, and $\kappa(r,g)$ the
state-dependent Hagen--Poiseuille/Lucas--Washburn
rate~\cite{Hagen1839,Poiseuille1840,Lucas1918,Washburn1921}.

The pair weight in Eqs.~\eqref{eq:G},~\eqref{eq:L}, $g(r')[1-g(r)]$ is
the exclusion rule (donor full, receiver empty).  The driving potential deserves emphasis, as it is the single channel
through which all physics other than connectivity enters.  It is the
(dimensionless) difference in water chemical potential between two pore
classes,
\begin{equation}
\Phi(r,r')=\frac{\mu_w(r')-\mu_w(r)}{\rho_w\mathrm{g}},
\label{eq:Phi_def}
\end{equation}
measured as an equivalent head (length), with $\rho_w$ the water
density, $\mathrm{g}$ gravity, and $\mu_w$ the chemical potential per
unit volume.  It is antisymmetric, $\Phi(r',r)=-\Phi(r,r')$, so water
moves from high to low potential.  For pure capillarity $\mu_w(r)$ is
the Young--Laplace term $-2\gamma\cos\theta_c/r$, so smaller pores
(more negative potential) pull water from larger ones.  Crucially, the
\emph{form} of $\mu_w$ can be enriched, by adding gravitational,
adsorptive or osmotic contributions (Sec.~\ref{sec:forces}), without
altering any of the structure that follows: every such force enters
only through $\Phi$, and the operator algebra below is indifferent to
which terms $\mu_w$ contains: adding gravitational, osmotic or adsorptive
contributions changes only the effective equilibrium distribution
$g_\eq$, leaving the operator structure of the reduction intact.

In the
quasi-static regime relevant below the contributing classes sit at
$g\!\approx\!1$, where 
\begin{equation}
\kappa(r,g)\to\kappa_{\rm HP}(r)=
r^2/(8\mu\bar L^2\bar\tau^2)\quad[\mathrm{Pa^{-1}s^{-1}}]
\end{equation}
~\cite{Hagen1839,Poiseuille1840}; we write
$\kappa(r)$ for this limiting value.  

The bilinear pair~\eqref{eq:G}--\eqref{eq:L} with signed $\Phi$ is the
phenomenological form; the operator we analyze is its equilibrium-consistent
completion, the
\emph{Onsager gradient flow} of the Gibbs free energy $\F[g]$~\cite{Rigon2026PRE}, which $\C[g]$ is taken to be from here on:  
\begin{equation}
\C[g](r)=\int_0^\infty \Mr(r,r';g)\,\big[\mu_w(r')-\mu_w(r)\big]\,
   f(r')\,dr',
\label{eq:onsager}
\end{equation}
carried by the symmetric, non-negative \emph{mobility kernel}
Explicitly, in terms of the symmetric harmonic-mean pair conductance
$\kappa_s$ [$\mathrm{Pa^{-1}s^{-1}}$], the connectivity $C$ [--] and the occupancy gates, functions of $g$ that record the status of filling of pores
\begin{widetext}
\begin{equation}
\Mr(r,r';g)=\kappa_s(r,r')\,C(r,r')\,
   \sqrt{g(r)\big[1-g(r)\big]\,g(r')\big[1-g(r')\big]}\,,
\label{eq:Mgeneral}
\end{equation}
\end{widetext}
with $\Mr$ in [Pa$^{-1}$s$^{-1}$]; its equilibrium value
$\Mr(r,r';g_\eq)$ is Eq.~\eqref{eq:Mexplicit}.
The \emph{geometric} mean of the two occupancy factors, rather than
their product, is what Arrhenius (heat-bath) exclusion rates deliver on
linearization; the double product would count each gate twice, once as
exclusion, once inside the intrinsic conductance, and, as
Sec.~\ref{sec:gap} shows, would hand the slow end of the relaxation spectrum
to the edges of the active band instead of to the frontier.
$\Mr(r,r';g)$,the central operator of the companion
paper~\cite{Rigon2026PRE} (the principal symbols and their units are
collected in Table~\ref{tab:symbols}).  
Two things
are worth keeping in view.  First, every force enters only through
$\mu_w$ (equivalently $\Phi$); the mobility $\Mr$ carries the
constitutive content, and the linearized operator of
Sec.~\ref{sec:operator} is precisely $\Mr$ evaluated at the equilibrium
step.  Second, \eqref{eq:onsager} is the \emph{conservative form},a
current divergence in pore-class space, which is what lets the whole
reduction be restarted from the continuum equation alone
(App.~\ref{app:schur}).
The sinks $\mathcal{E}$ (evaporation) and $\mathcal{T}$ (root uptake)
act within the REV.

\subsection{Intra- and inter-REV operators}

Making the continuum coordinate explicit, $g=g(r,\mathbf{x},t)$, the
redistribution operator $\C$ of~\eqref{eq:master} moves water between
\emph{pore classes} within the REV at $\mathbf{x}$ and conserves the
local water content, $\int\C[g]\,f\,dr=0$.  Alongside it the companion
paper defines the \emph{inter-REV} (boundary) operator
\begin{equation}
\Cd:=\mathcal{G}_\partial-\Lg_\partial,
\label{eq:Cd_def}
\end{equation}
built from the same bilinear gain/loss but with donor and receiver
populations evaluated at neighbouring REVs separated by the spacing
$\ell\sim\bar L$.  For a neighbour at $\mathbf{x}_+=\mathbf{x}+\ell\hat n$,
\begin{align}
\mathcal{G}_\partial(r)&=\kappa(r)\,[1-g(r,\mathbf{x})]\nonumber\\
  &\quad\times\!\int_0^\infty\! C_\partial(r,r')\,\rho_w\mathrm{g}\,\Phi(r,r')\,
  g(r',\mathbf{x}_+)\,f(r')\,dr',
\label{eq:Gbnd}\\
\Lg_\partial(r)&=\kappa(r)\,g(r,\mathbf{x})\nonumber\\
  &\quad\times\!\int_0^\infty\! C_\partial(r,r')\,\rho_w\mathrm{g}\,\Phi(r',r)\,
  [1-g(r',\mathbf{x}_+)]\,f(r')\,dr'.
\label{eq:Lbnd}
\end{align}
The boundary connectivity $C_\partial$ need not equal the internal $C$.
Only the spatial offset $\mathbf{x}_+\!\neq\!\mathbf{x}$ distinguishes
$\Cd$ from $\C$: it is this offset that makes $\Cd$ \emph{transport}
water between REVs rather than redistribute it within one.

\subsection{The gradient expansion and the entangled flux}
\label{sec:gradflux}

Contracting the REV to a point is the spatial homogenisation
$\varepsilon=\ell/L_{\rm mac}\to0$.  Expanding the neighbour's
distribution 
\begin{equation}
g(r',\mathbf{x}_+)=g(r',\mathbf{x})+\ell\,\hat n\!\cdot\!
\nabla g(r',\mathbf{x})+O(\ell^2)
\end{equation}
 in~\eqref{eq:Gbnd}--\eqref{eq:Lbnd}
and using the antisymmetry $\Phi(r',r)=-\Phi(r,r')$, the
\emph{zeroth}-order terms reproduce the internal redistribution already
contained in $\C$ and cancel from $\Cd$; the \emph{first}-order terms
give a directional current linear in $\nabla g$ (full algebra in
App.~\ref{app:gradient}):
\begin{widetext}
\begin{equation}
\Cd=\ell\,\hat n\!\cdot\!\kappa(r)\,[1-2g(r)]
  \!\int_0^\infty\! C_\partial(r,r')\,\Phi(r,r')\,
  \nabla g(r',\mathbf{x})\,f(r')\,dr'+O(\ell^2).
\label{eq:Cd}
\end{equation}
\end{widetext}
Reading the kernel off~\eqref{eq:Cd}, the inter-REV current carried by
class $r$ is
\begin{equation}
\mathbf{F}(r,\mathbf{x},t):=-\int_0^\infty
  \Gamma(r,r';g)\,\nabla g(r',\mathbf{x},t)\,dr',
\label{eq:Fentangled}
\end{equation}
with the \emph{entangled} kernel
\begin{equation}
\Gamma(r,r';g):=\kappa(r)\,[1-2g(r)]\,C_\partial(r,r')\,\rho_w\mathrm{g}\,\Phi(r,r')\,f(r'),
\label{eq:Gamma}
\end{equation}
The occupancy weight $[1-2g]$ is the linearization in $g'$ of the bilinear
gain--loss pair~\eqref{eq:G}--\eqref{eq:L} with signed $\Phi$; for the
rectified rates, or for the Onsager form~\eqref{eq:onsager}, the weight is
the gate derivative of the active direction and differs away from the
frontier.  The reduction below uses only the localization of $\Gamma$ on the
frontier and the magnitude of its driving, which are the same for all three
forms.
so that $\Cd=-\nabla\!\cdot\!\mathbf{F}+O(\ell^2)$, the divergence being
made precise by the limit~\eqref{eq:Flimit} below.  Two features carry
over from the boundary expansion and are essential below: (i)~the flux
of class $r$ is a sum over \emph{all} connected classes $r'$, weighted
by $C_\partial\Phi f'$,a nonlocality in $r$ that no spatial limit can
remove; and (ii)~the kernel $\kappa(1-2g)C_\partial\Phi$ \emph{entangles}
the transport coefficient with the driving force $\nabla g$.  

\subsection{The closed continuum kinetic equation}

Inter-REV transport conserves water and can only cross $\partial V$, the boundaries of the REV volumes, so
the boundary exchange contracts to a divergence,
\begin{equation}
\frac{1}{|V|}\oint_{\partial V}\mathbf{F}\!\cdot\!\hat n\,dA
  \xrightarrow{\ell\to0}\nabla\!\cdot\!\mathbf{F}.
  \label{eq:Flimit}
\end{equation}
Dividing the class-$r$ budget by $|V|$ and letting $\ell\to0$ gives the
local continuum balance, valid \emph{at any timescale separation}.  It
is Eq.~\eqref{eq:contKE} of the Introduction, quoted there as the
starting point, derived here, now with the interior sinks displayed
explicitly:
\begin{equation}
\boxed{\;
\partial_t g(r,\mathbf{x},t)+\nabla\!\cdot\!\mathbf{F}(r,\mathbf{x},t)
  =\C[g]-\mathcal{E}(r)-\mathcal{T}(r)\;}
\label{eq:contKE2}
\end{equation}
The continuum limit adds only the transport divergence
$\nabla\!\cdot\!\mathbf{F}$ to the intra-REV dynamics; the
redistribution operator $\C$, the $H$-theorem $d\F/dt\le0$
(companion paper~\cite{Rigon2026PRE}), and the
well-posedness of the master equation carry over unchanged, now
evaluated at each $\mathbf{x}$.  The REV has vanished as a geometric
object and survives only through the local PSD $f(r,\mathbf{x})$ carried
by the operators.  Equation~\eqref{eq:contKE2} is the
\textbf{starting point} of the CE reduction; the bulk
water budget
\begin{equation}
\partial_t\theta+\nabla\!\cdot\!\mathbf{q}=0,
\qquad \mathbf{q}=\int\mathbf{F}\,f\,dr,
\label{eq:bulk}
\end{equation}
is its $f$-weighted zeroth moment (using $\int\C[g]f\,dr=0$) and is
\emph{not} closed: $\mathbf{q}$ is undetermined by $\theta$ alone.  The
missing closure $\mathbf{q}[g]$ is the entire non-equilibrium content; it
is what the CE expansion supplies.

\begin{table*}[t]
\caption{Principal symbols, their meaning, the relation that ties each to
the others, and SI units.  The pore-class current $j$, the mobility
$\Mr$, and the spatial flux $\mathbf{F}$ are the moving parts of the
double conservation law
$\partial_t g+\nabla\!\cdot\!\mathbf{F}+\mathrm{div}_r j
=-\mathcal{E}-\mathcal{T}$: $\Mr$ is the constitutive coefficient, $j$ the
current it drives through pore-size space, $\mathbf{F}$ the analogous
current through physical space.}
\label{tab:symbols}
\renewcommand{\arraystretch}{1.22}
\begin{ruledtabular}
\begin{tabular}{l p{0.34\textwidth} p{0.31\textwidth} l}
symbol & meaning & relation (hint) & units \\
\hline
$g(r,\mathbf{x},t)$ & filled fraction of pores of radius $r$ &
   state variable; $\theta=\phi\!\int g f\,dr$ & -- \\
$f(r)$ & pore-size density (PSD) & measure on classes, $\int f\,dr=1$ &
   $\mathrm{m^{-1}}$ \\
$\theta,\ \phi$ & water content; porosity & $\theta=\phi\!\int g f\,dr$ &, \\
$g_\eq(r;\theta)$ & equilibrium step $H(r^*-r)$ &
   minimizes $\mathcal{F}$; $\J:=\C'|_{g_\eq}$ & -- \\
$\mu_w(r)$ & chemical potential (suction) &
   lattice gas $\mu_0+\psi_T\ln\frac{g}{1-g}$ & Pa \\
$\Phi(r,r')$ & driving head & $\Phi=[\mu_w(r'){-}\mu_w(r)]/\rho_w\mathrm{g}$
   & m \\
$\psi_T$ & configurational temperature &
   stiffness $\partial_g\mu_w=\psi_T/[g(1{-}g)]$ & Pa \\
$\kappa_s(r,r')$ & symmetric pair rate & building block of $\Mr,\Gamma$ &
   $\mathrm{Pa^{-1}s^{-1}}$ \\
$\omega(r,r')$ & equilibrium exchange conductance & $\omega=\psi_T\Mr(g_\eq)=\psi_T\kappa_sC\sqrt{mm'}$ & $\mathrm{s^{-1}}$ \\
Arrhenius (heat-bath) rate & pair transfer rate $w_{r'\to r}\propto\kappa_s e^{-[\mu_w(r)-\mu_w(r')]/2\psi_T}$ & detailed-balanced by construction; linearizes to the $\sqrt{mm'}$ gate & -- \\
Ohmic (edge, conductance) & linear current--force relation with a state-independent conductance & the linearized edge $K_{rr'}[\delta\mu_w'-\delta\mu_w]$, Eq.~\eqref{eq:Krr} & -- \\
$C,\ C_\partial$ & connectivity (within / between REV) &
   gates which pairs exchange & -- \\
$\C[g](r)$ & redistribution operator &
   $\C=\int j\,f'dr'=-\mathrm{div}_r j$ & $\mathrm{s^{-1}}$ \\
$j[g](r,r')$ & pore-class pair-current &
   $j=\Mr\,[\mu_w(r'){-}\mu_w(r)]$; antisym.\ & $\mathrm{s^{-1}}$ \\
$\Mr(r,r';g)$ & mobility kernel (Onsager) &
   carrier of $j$; $K_{rr'}=\Mr(g_\eq)f'$ & $\mathrm{Pa^{-1}s^{-1}}$ \\
$\mathbf{F}(r,\mathbf{x},t)$ & entangled spatial flux &
   $\mathbf{F}=-\!\int\Gamma\nabla g\,dr'$; $\mathsf{T}=\nabla\!\cdot\!
   \mathbf{F}$ & $\mathrm{m\,s^{-1}}$ \\
$\Gamma(r,r';g)$ & transport kernel (spatial mobility) &
   spatial analogue of $\Mr$ & $\mathrm{m\,s^{-1}}$ \\
$\J=\C'|_{g_\eq}$ & linearized operator (graph-Laplacian) &
   $\J=\mathrm{div}_r(\Mr\nabla_r\mu_w)$; self-adjoint & $\mathrm{s^{-1}}$ \\
$\mathcal{A}(r),\ \tau{=}m/\mathcal{A}$ & exchange conductance / relaxation time &
   $\J h=-(\mathcal{A}/m)h+\int\mathcal{B}h f'$ & $\mathrm{s^{-1}},\ \mathrm{s}$ \\
$\lambda_n,\ \varphi_n$ & relaxation spectrum &
   $\J\varphi_n=-\lambda_n\varphi_n$, $\lambda_0{=}0$ &
   $\mathrm{s^{-1}},$ , \\
$\mathrm{Da}=\tau_\eq/\tau_{\rm mac}$ & Damk\"ohler number &
   small parameter of the reduction & -- \\
$\mathbb{L}=\Da^{-1}\J-\mathsf{T}$ & two-scale generator &
   Schur complement $\to$ Richards & $\mathrm{s^{-1}}$ \\
$K(\theta)$ & hydraulic conductivity &
   $K=\phi\langle\kappa|\J^{-1}|S_0C_w\rangle$ & $\mathrm{m\,s^{-1}}$ \\
\hline
\multicolumn{4}{l}{\emph{additional symbols used in the appendices}}\\
\hline
$\kappa(r)$ & single-pore (Hagen--Poiseuille) rate &
   $\kappa=r^2/(8\mu\bar L^2\bar\tau^2)$ & $\mathrm{Pa^{-1}s^{-1}}$ \\
$\mathcal{B}(r,r')$ & off-diagonal redistribution kernel &
   $\J h=-\mathcal{A}h+\int\mathcal{B}h f'dr'$ & $\mathrm{s^{-1}}$ \\
$K_{rr'}$ & edge conductance (graph-Laplacian weight) &
   $K_{rr'}=\Mr(g_\eq)f(r')$ & $\mathrm{Pa^{-1}s^{-1}m^{-1}}$ \\
$S,\ S_0$ & first-order source; its gradient factor &
   $S=\nabla\!\cdot\!\mathbf{F}[g^{(0)}]$, $S=S_0\nabla\theta\!\cdot\!\hat n$
   & $\mathrm{s^{-1}}$ \\
$\bm\chi(r)$ & response function &
   $\J\bm\chi=S_0C_w$; $K=\phi\langle\kappa,\bm\chi\rangle_f$ &
   $\mathrm{m^{-1}s}$ \\
$C_w=d\theta/d\psi$ & specific moisture capacity &
   slope of the retention curve & $\mathrm{m^{-1}}$ \\
$\mathsf{P},\ \mathsf{Q}$ & null projector; its complement &
   $\mathsf{P}u=(\int uf\,dr)\cdot1$, $\mathsf{Q}=\mathcal{I}-\mathsf{P}$
   & -- \\
$W=(\ker\J)^\perp$ & zero-mean fluctuation space &
   $W=\{u:\int uf\,dr=0\}$ & -- \\
$\mathsf{T}=\nabla\!\cdot\!\mathbf{F}$ & transport (streaming) operator &
   fast/slow split $\mathbb{L}=\Da^{-1}\J-\mathsf{T}$ & $\mathrm{s^{-1}}$ \\
$a_w(r),\ \Xi(r),\ \bar\tau$ & accessibility; geometry factor; tortuosity &
   enter $K_{\rm mf}$~\eqref{eq:K_mf} & -- \\
$\Gamma_w,\ \alpha_w$ & inter-band exchange; its coefficient &
   $\Gamma_w=\alpha_w(\psi_t-\psi_b)$ &
   $\mathrm{s^{-1}},\ \mathrm{m^{-1}s^{-1}}$ \\
$\varrho_\tau=\tau_b/\tau_t$ & band separation ratio &
   two-band validity, $\varrho_\tau\ll1$ & -- \\
\end{tabular}
\end{ruledtabular}
\end{table*}

\section{Scale separation, the Chapman--Enskog expansion,
and the equilibrium manifold}
\label{sec:scale}

This section does three things: it identifies the small parameter that
makes the reduction possible, sets up the CE expansion in it,
and develops the expansion's leading order, the equilibrium soil, in
full.  The two higher orders are previewed here and carried out, once the
operator $\J$ is in hand, in Secs.~\ref{sec:operator}--\ref{sec:K}.

What follows is a \emph{formal} CE reduction, in the
tradition of kinetic theory: the expansion is organized in powers of the
Damk\"ohler number and closed order by order through Fredholm
solvability, exactly as in the Boltzmann-to-Navier--Stokes passage.  The
operator statements below, self-adjointness, the spectral gap, bounded
invertibility on $(\ker\J)^\perp$,are established at the level of rigor
customary for such reductions, through explicit construction and physical
argument, rather than as theorems carrying full existence, convergence,
and error estimates, which lie beyond the present scope.  

Before expanding, we recall the one result of the companion
paper~\cite{Rigon2026PRE} that the reduction takes as given.  Setting
the redistribution operator to zero, $\C[g]=0$, and minimizing the Gibbs
free energy $\F[g]$ at fixed water content selects a unique local
equilibrium: the \emph{packing step}
\begin{equation}
g_\eq(r;\theta)=H\!\big(r^*(\theta)-r\big),
\label{eq:geq_recap}
\end{equation}
in which all pores narrower than a frontier radius $r^*$ are full and
all wider ones empty (the narrow pores hold water most tightly, so they
fill first).  The frontier $r^*$ is fixed by the water content through
$\theta=\phi\int_0^{r^*}f\,dr$, and $H$ is the Heaviside step function.  This
step distribution is the attractor of the fast redistribution dynamics
and the manifold onto which the CE expansion collapses the state; that
such a manifold exists, is attracting, and can be constructed at
\emph{finite} scale ratios is the content of the invariant-manifold
theory of dissipative multiscale
systems~\cite{Roberts2015ima,BunderRoberts2021,Roberts2025tma}, to which
we return in Sec.~\ref{sec:geometry}.
The packing step is the unique minimizer, and the attractor of the fast
dynamics, only on a fully connected pore space over which the free energy
is convex in the filling at fixed $\theta$, so that water can always
reach the narrowest available pores.  Two physical departures bound the
domain of the reduction. 
(i)~\emph{Disconnected or trapped clusters}: below percolation, or where
a phase is isolated, some pores cannot exchange, and the accessible
equilibrium is a \emph{constrained} step on the spanning cluster only,
described through the accessibility function $a_w(r)$.  (ii)~\emph{Ink-bottle
and metastable states}: a wide pore shielded behind a narrow throat can
stay filled (or empty) against the global step, a metastable
configuration separated from $g_\eq$ by a finite barrier.  The reduction
assumes these are absent or slow compared with $\tau_{\rm redis}$, so the
fast dynamics relaxes to~\eqref{eq:geq_recap} on the connected pore space
before the macroscopic gradient acts.  The first of these limits the
constitutive inputs to a single branch; the second is precisely the
percolation regime in which the spectral gap closes
(Sec.~\ref{sec:gap}) and the expansion ceases to hold.

\subsection{The single control parameter}
\label{sec:Da}

We non-dimensionalize~\eqref{eq:contKE2} using the macroscopic length
$\Lambda$ (the scale over which $\theta$ varies appreciably) for space,
the forcing time $\tau_{\rm forcing}$ for time, and the equilibrium step
for $g$; this exposes one dimensionless group.
The redistribution operator carries the in-REV relaxation rate
$1/\tau_{\rm redis}$ (the slowest nonzero rate of $\C$ linearized about
$g_\eq$); the streaming term carries the inter-REV forcing rate
$1/\tau_{\rm forcing}$.  Their ratio is the Damk\"ohler number
\begin{equation}
\Da(r,\mathbf{x})=\frac{\tau_{\rm redis}(r)}{\tau_{\rm forcing}(\mathbf{x})}\ll1,
\label{eq:Da}
\end{equation}
and on the slow (forcing) time scale~\eqref{eq:contKE2} reads
\begin{equation}
\partial_t g+\nabla\!\cdot\!\mathbf{F}=\frac{1}{\Da}\,\C[g],
\label{eq:singular}
\end{equation}
the canonical fast--slow form: redistribution is fast, streaming slow.
Both times are set by pore-scale quantities already in hand:
$\tau_{\rm redis}(r)=1/[\kappa_{\rm eff}(r)\,\Phi(\bar r)]$ and
$\tau_{\rm forcing}=\bar L\phi/|\mathbf{q}|$, where $\mathbf{q}$ is the
macroscopic (Darcy) flux, the $f$-weighted moment of $\mathbf{F}$
introduced in Eq.~\eqref{eq:F_intro} and set on $\partial\Omega$ by the
atmospheric/water-table boundary conditions stated there.  Thus
$\tau_{\rm forcing}$ is simply the time for the boundary flux to fill a
REV.

A single small parameter controls the reduction.  Viscous (Stokes) pore
flow has no momentum invariant: the velocity is function of the pressure
gradient, the only conserved quantity is water mass, and the spatial
scale-separation parameter $\varepsilon=\bar L/\Lambda$ (the ratio of the
REV length $\bar L\approx L_{\rm rev}$ to the macroscopic field length
$\Lambda$) is \emph{not} independent of $\Da$.  A pore-scale balance
shows $\varepsilon\sim\Da$ (companion paper~\cite{Rigon2026PRE}, App.~A): a sharp front imposed at $\Da\ll1$ self-heals by redistribution
within a few $\tau_{\rm redis}\ll\tau_{\rm forcing}$.  The single
condition $\Da\ll1$ therefore guarantees \emph{both} local equilibrium
and smooth macroscopic gradients, so the whole expansion is organized by
$\Da$ alone.

Because the inter-REV exchange, evaporation and uptake are all $O(\Da)$,
the water content evolves slowly, $\partial_t\theta\sim O(\Da)$, and it
is natural to introduce the slow time
\begin{equation}
T:=\Da\,t,
\label{eq:slowtime}
\end{equation}
on which one unit of $T$ spans many units of $t$.  The slow clock is what
makes the single power of $\Da$ cancel uniformly from the first-order
balance (Sec.~\ref{sec:step2}), leaving a $\Da$-free macroscopic law.

\subsection{The Chapman--Enskog expansion}
\label{sec:hilbert_ce}
The CE method extracts a closed macroscopic equation
from a kinetic equation with fast relaxation~\cite{ChapmanCowling1970}: it
expands not the solution but the \emph{constitutive functional}, leaving
$\theta(\mathbf{x},t)$ unexpanded and expanding only the dependence of $g$
upon it,
\begin{align}
g&=g^{(0)}\!\big(r;\theta(\mathbf{x},t)\big)
  +\Da\,g^{(1)}\!\big(r;\theta,\nabla\theta\big)+O(\Da^2),
\label{eq:ansatz}\\
0&=\phi\int_0^\infty g^{(n)}(r)\,f(r)\,dr\qquad(n\ge1),
\label{eq:constraint}
\end{align}
the constraint placing the whole water content in $g^{(0)}$.  (This is
what distinguishes CE from the older \emph{Hilbert} expansion~\cite{Hilbert1912}.)

Substituting into the singularly perturbed CKE~\eqref{eq:singular} and
expanding the redistribution operator about $g^{(0)}$,
\begin{equation}
\C[g]=\C[g^{(0)}]+\Da\,\J\,g^{(1)}+O(\Da^2),\qquad
\;\J:=\left.\frac{\delta\C}{\delta g}\right|_{g_\eq}\;
\label{eq:Jdef_intro}
\end{equation}
introduces the \emph{linearized redistribution operator} $\J$ and sorts
the reduction into orders.  We develop the leading order in full in the
next subsection, then preview the two that the operator sections carry
out.

\subsection{Zeroth order: the equilibrium soil}
\label{sec:step1}
At leading order the fast term must vanish on its own, $\C[g^{(0)}]=0$.
The bilinear exclusion structure forces $g^{(0)}(r)\in\{0,1\}$ for each
class, and minimization of the Gibbs free energy at fixed $\theta$ selects
the packing step
\begin{equation}
g^{(0)}(r;\theta)=g_\eq(r;\theta)=H(r^*(\theta)-r),
\label{eq:g0}
\end{equation}
with all pores narrower than the frontier $r^*$ full, all wider empty, with
$r^*$ fixed by inverting the cumulative pore-size distribution,
\begin{equation}
\theta=\phi\int_0^{r^*}f\,dr
\quad\Longrightarrow\quad
\frac{dr^*}{d\theta}=\frac1{\phi f(r^*)}.
\label{eq:frontier}
\end{equation}
This zeroth order  is the entire equilibrium
constitutive theory of the soil.

\textbf{Retention curve.}  Because the largest filled pore has the
well-defined radius $r^*(\theta)$, the Young--Laplace law assigns it a
single capillary pressure, which is the matric potential:
\begin{equation}
\psi(\theta)=-\frac{2\gamma\cos\theta_c}{\rho_w\mathrm{g}\,r^*(\theta)},
\label{eq:psi}
\end{equation}
single-valued because $r^*$ is monotone in $\theta$
through~\eqref{eq:frontier}
($\gamma$ is surface tension, $\theta_c$ is the contact angle).  Van
Genuchten~\cite{vanGenuchten1980} and Brooks--Corey~\cite{Brooks1964} are
coordinate transformations on this one-dimensional submanifold
(Sec.~\ref{sec:classical}).  Differentiating gives the specific moisture
capacity $C_w=d\theta/d\psi$ and the frontier-shift tangent
\begin{equation}
\pder{g_\eq}{\theta}=\frac{\delta(r-r^*)}{\phi f(r^*)},
\label{eq:tangent}
\end{equation}
the one direction on the manifold that moves water (the frontier
derivative $dr^*/d\theta=1/[\phi f(r^*)]$ is obtained by differentiating
$\theta=\phi\int_0^{r^*}\!f\,dr$; see App.~\ref{app:gradient}).

\textbf{The equilibrium manifold.}  Let
$\mathcal{S}=\{g:(0,\infty)\to[0,1]\}$ be the space of admissible occupancy
profiles.  The equilibria form a one-dimensional curve in it,
$\mathcal{S}_\eq=\{g_\eq(\cdot;\theta):\theta\in[0,\phi]\}$, parametrized
either by $\theta$ or $r^*$ or $psi$ with
tangent~\eqref{eq:tangent}.  Directions orthogonal to it reshape $g(r)$ at
fixed $\theta$: they are the degrees of freedom invisible to Richards' equation.
Classical soil physics, retention, conductivity, diffusivity, Richards'
flow, lives entirely on $\mathcal{S}_\eq$.

\textbf{The step is a reference.}  However, the sharp step presumes $\mu_w$ depends
on radius alone.  Two refinements widen it, both from the companion
paper~\cite{Rigon2026PRE}: a finite configurational temperature $\psi_T$
softens it to the Fermi-type form
\begin{equation}
g_\eq^{\rm FD}(r;\theta)=\frac{1}{1+\exp[(\mu_{\rm cap}(r)-\lambda)/\psi_T]},
\label{eq:geq_FD}
\end{equation}
relaxing to $H(r^*-r)$ as $\psi_T\to0$; and averaging
over an elevation spread $\ell$ smears it deterministically, of width
$\sim(dr^*/d\mu_{\rm cap})\rho_w\mathrm{g}\,\ell$, making the retention
curve weakly sample-height dependent.  The consequence that matters for
the reduction is that the physical width
$\max(\psi_T,\rho_w\mathrm{g}\,\ell)$ turns the frontier into a
\emph{resolved interval} $g_\eq(1-g_\eq)>0$,the active band on which the
redistribution eigenmodes of Sec.~\ref{sec:eigenmodes} live, fixing the
resolution of the relaxation spectrum and the smoothness of $K(\theta)$.
The sharp step~\eqref{eq:g0} remains the correct reference about which the
corrections $g^{(n)}$ are defined.

\subsection{The two higher orders, previewed}
\noindent\emph{Order $\Da^{0}$: the linear correction.}
The first departure from equilibrium solves
\begin{equation}
\J\,g^{(1)}=\partial_T g^{(0)}+\nabla\!\cdot\!\mathbf{F}[g^{(0)}]=:S ,
\label{eq:ord_0}
\end{equation}
the \emph{same} linear equation that recurs at every higher order with a
source rebuilt from below (Sec.~\ref{sec:step2}).  Everything now hinges on
$\J$.


Because the source is known, solving Eq.~\eqref{eq:ord_0} rests entirely
on the properties of the operator
$\J$, so Sec.~\ref{sec:operator} constructs it and establishes the three
properties the orders require, self-adjointness, a spectral gap, and
bounded invertibility on $(\ker\J)^\perp$.  Section~\ref{sec:hierarchy}
then carries out the first order and its solvability, and
Sec.~\ref{sec:K} solves~\eqref{eq:ord_0} for the conductivity.  The same
reduction has an equivalent one-step operator form, the Schur complement
of App.~\ref{app:schur}.

\section{The linearized redistribution operator}
\label{sec:operator}

Sec.~\ref{sec:hilbert_ce} rests entirely on the linearized
redistribution operator $\J=\C'|_{g_\eq}$ of Eq.~\eqref{eq:Jdef_intro},
\begin{equation}
\J:=\C'\big|_{g_\eq},\qquad g=g_\eq+\delta g .
\label{eq:Jdef}
\end{equation}
Therefore, we construct it here, directly from the conservative (continuum) form of
$\C$, and establish the three properties the orders demand.  Being a linear operator, it inherits a
structure, a local relaxation rate minus a redistribution integral, that
makes its spectral properties, and hence the inversion at the heart of the
reduction, elementary.

We build $\J$ here at the continuum level, reading the closed
redistribution operator~\eqref{eq:onsager} as the pore-class divergence
of the mobility current and linearizing that single current at the step.

\subsection{Construction from the conservative current}

We obtain $\J$ directly from the \emph{conservative form} of $\C$.  The classical kinetic-theory
derivation, differentiating the gain and loss separately at a single
REV, à la Cercignani, reaches the same operator by a longer route and is given step by step in App.~\ref{app:cercignani}.

Because $\C$ conserves water, $\int\C[g]f\,dr=0$~\eqref{eq:bulk}, it is the
net flux of an antisymmetric pair-current,
\begin{widetext}
\begin{equation}
\C[g](r)=\int_0^\infty j[g](r,r')\,f(r')\,dr',\qquad
  j[g](r',r)=-\,j[g](r,r'),
\label{eq:jcurrent}
\end{equation}
\end{widetext}
antisymmetry of $j$ being exactly what makes $\int\C f\,dr$ vanish: at the
continuum level $\C$ is already the divergence of $j$ in pore-class space,
the gain and loss subsumed into one conservative current.  
That current is
read straight off the Onsager form~\eqref{eq:onsager} as the mobility times the thermodynamic force,
\begin{equation}
j[g](r,r')=\underbrace{\Mr(r,r';g)}_{\substack{\text{mobility}\\
   \text{symmetric},\ \ge0}}\;
   \underbrace{\big[\mu_w(r')-\mu_w(r)\big]}_{\substack{\text{force}\\
   =\,\rho_w\mathrm{g}\,\Phi(r,r')}}
\label{eq:jsplit}
\end{equation}

\vspace{3ex}
\noindent
The \emph{force} is the chemical-potential difference, equivalently the
head $\Phi$; it is antisymmetric and vanishes when the two classes share a
potential.  The \emph{mobility} $\Mr$ is the symmetric, non-negative
envelope that carries the current; its equilibrium value sets the
conductance of Eq.~\eqref{eq:Krr} through $K_{rr'}=\Mr(r,r';g_\eq)\,f(r')$,
\begin{equation}
\Mr(r,r';g_\eq)=\kappa_s(r,r')\,C(r,r')\,\sqrt{m(r)\,m(r')}\,,
\label{eq:Mexplicit}
\end{equation}
where
\begin{equation}
m(r):=g_\eq(r)\,[1-g_\eq(r)]
\label{eq:mdef}
\end{equation}
is the \emph{active-band weight}. The occupancy gates carried by $\Mr$ originate in the
mesoscale gain and loss; the explicit exclusion-current form that exhibits
them, and its reduction to~\eqref{eq:jsplit}, is given step by step in
App.~\ref{app:cercignani}.  

At equilibrium the current vanishes pairwise, for two reasons that
divide the labour: on the sharp step every gate pair is closed (for $r'<r^*$
the receiver is full, for $r'>r^*$ the donor is empty), and on the resolved
band of Sec.~\ref{sec:gap} the \emph{total} potential, capillary plus
entropic, is uniform, $\mu_w[g_\eq]=\lambda$, so the force vanishes while the
mobility does not, strong detailed balance (consequence~(A) at
Eq.~\eqref{eq:g0}).
Linearizing the single current about this zero, $g=g_\eq+h$, the
vanishing of the force removes the derivative of the mobility
prefactor, leaving what is called an \emph{Ohmic} edge, a conductance that no longer depends on the state it carries (Table~\ref{tab:symbols}):
the linearized current is the mobility times the \emph{perturbed}
potential difference,
\begin{widetext}
\begin{equation}
\delta j[h](r,r')=\Mr(r,r';g_\eq)\,\big[\delta\mu_w(r')-\delta\mu_w(r)\big],
\qquad
 \delta\mu_w=\frac{\psi_T}{g_\eq(1{-}g_\eq)}\,h .
\label{eq:linj}
\end{equation}
\end{widetext}
The entropic stiffness $\partial_g\mu_w|_{g_\eq}=\psi_T/[g_\eq(1{-}g_\eq)]$
turns the disturbance $h=\delta g$ into a potential perturbation
$\delta\mu_w$, and in doing so converts the mobility $\Mr$
$[\mathrm{Pa^{-1}\,s^{-1}}]$ into a relaxation rate $[\mathrm{s^{-1}}]$.
It is absorbed into the conductances~\eqref{eq:Krr}, together with the
occupancy gates of $\Mr$; the inner-product weight of the Laplacian
core $\Lop$ is the mass measure $f$, while $\J$ itself is self-adjoint in
$f/m$ (Prop.~\ref{prop:selfadj}).
Substituting~\eqref{eq:linj} into~\eqref{eq:jcurrent} and separating the
term in $h(r)$ (the \emph{diagonal}) from the term in $h(r')$ (the
\emph{off-diagonal}) gives
\begin{equation}
\J[h](r)=-\frac{\mathcal{A}(r)}{m(r)}\,h(r)
  +\int_0^\infty \mathcal{B}(r,r')\,h(r')\,f(r')\,dr',
\label{eq:J}
\end{equation}
where the two coefficients are built from
one symmetric, non-negative object, the equilibrium one-way exchange
flux
\begin{equation}
\omega(r,r'):=\psi_T\,\Mr(r,r';g_\eq)=\omega(r',r)\ \ge 0 ,
\label{eq:omega}
\end{equation}
namely
\begin{equation}
\mathcal{A}(r):=\int_0^\infty\!\omega(r,r')\,f(r')\,dr'\ \ge0,
\qquad
\mathcal{B}(r,r'):=\frac{\omega(r,r')}{m(r')} .
\label{eq:A}
\end{equation}
The factor $1/m$ is the entropic stiffness of~\eqref{eq:linj} and must
not be dropped: the Laplacian acts on the \emph{potential} perturbation
$\delta\mu_w=\psi_T h/m$, not on the occupancy perturbation $h$.  Its
presence is what distinguishes the relaxation \emph{rate} from the
exchange \emph{conductance}: $\mathcal{A}(r)$ is a row sum of
conductances, whereas

\begin{equation}
\text{relaxation rate of class }r
=\frac{\mathcal{A}(r)}{m(r)} ,
\qquad
\tau(r):=\frac{m(r)}{\mathcal{A}(r)} .
\label{eq:tau}
\end{equation}
Suppressing the redistribution term, a perturbation confined to class
$r$ obeys $\partial_t h=-[\mathcal{A}(r)/m(r)]\,h$ and decays as
$h(r,t)=h(r,0)e^{-t/\tau(r)}$; the slowest such mode,
$\tau_{\rm relax}=\max_r\tau(r)$, sets the spectral gap below.  Because
$m\to0$ on full and empty classes while $\mathcal{A}$ vanishes there
too, the rate is finite across the active band and is smallest at the
frontier itself: \emph{the classes that relax slowest are those at
$r^*$}, a fact that returns in Sec.~\ref{sec:multiperm}.  The coupling
$\mathcal{B}$ links distinct classes.  Equation~\eqref{eq:J} is a \emph{weighted graph-Laplacian} (a diffusion-like, mass-preserving operator) on
the space of pore classes.  Picture each pore class $r$ as a node of a
network, and let two classes be joined by an edge whenever water can pass
between them; the strength (conductance) of that edge is
\begin{equation}
K_{rr'}=\Mr(r,r';g_\eq)\,f(r')
\label{eq:Krr}
\end{equation}
with $\Mr(r,r';g_\eq)$ the equilibrium mobility of
Eq.~\eqref{eq:Mexplicit} ($\kappa_s$ the symmetric harmonic-mean pair
conductance),
the occupancy factors $g_\eq(1{-}g_\eq)$
vanishing for full or empty classes so that only classes \emph{near the
frontier} are active, and the other symbols as already defined. 
The occupancy factors are functions
of the class $r$, evaluated at the local water content $\theta$, which
enters parametrically through the frontier $r^*(\theta)$.  In terms of
these conductances the operator acts as a diffusion on the network of
pore classes.  From here on we use the discrete notation
$h_r:=h(r)$, $f_r:=f(r)\,\Delta r$ and
$\sum_{r'}:=\int\!\cdot\;dr'$, that is, we regard the class axis as
sampled on a grid fine enough that the integral operator and its matrix
representation may be used interchangeably; this is also how $\J$ is
realized numerically (Supplemental Material~\cite{SM_hilbert}, S1), and
no result below depends on the grid.  In this notation
\begin{equation}
\J[h](r)=\sum_{r'}K_{rr'}\,
   \big[\delta\mu_w(r')-\delta\mu_w(r)\big]
=:\Lop\big[\delta\mu_w\big](r),
\label{eq:graphlap}
\end{equation}
with $\delta\mu_w=\psi_T h/m$, i.e.\ $\J$ is the composition of a \emph{bona fide} weighted graph
Laplacian $\Lop$ with the pointwise map $h\mapsto\delta\mu_w$:
\begin{equation}
\J=\Lop\circ\frac{\psi_T}{m} .
\label{eq:Jcomp}
\end{equation}
Each class relaxes toward the conductance-weighted average of the
\emph{potentials} of the classes it is connected to.  This is the point
at which the occupancy and the potential must be kept apart: writing
$\Lop[h]$ in place of $\Lop[\delta\mu_w]$ would move the entropic
stiffness into the wrong slot and change the kernel of the operator
(Prop.~\ref{prop:selfadj}).  

These conductances are the equilibrium value of the mobility kernel
of~\eqref{eq:onsager}: the graph
Laplacian is the redistribution \emph{mobility} of the companion paper,
linearized about the step.  

\subsection{Self-adjointness, sign, and the single invariant}

The reduction needs three properties of $\J$: (i) self-adjointness in a
suitable weighted inner product; (ii) a definite sign, and (iii) a one-dimensional
kernel.  All three follow directly from the graph-Laplacian
form~\eqref{eq:graphlap}, with the weight set by the equilibrium step and
a single invariant (water mass).

\begin{prop}[Self-adjointness, kernel and cokernel]
\label{prop:selfadj}
Let $\Lop$ be the graph Laplacian of the conductances
$K_{rr'}=\Mr(r,r';g_\eq)f(r')$ and $\J=\Lop\circ(\psi_T/m)$
as in~\eqref{eq:Jcomp}.  Then
\begin{enumerate}
\item[(i)] $\Lop$ is self-adjoint and negative semidefinite in the mass inner
product $\langle u,v\rangle_f=\int u\,v\,f\,dr$, with
$\ker\Lop=\mathrm{span}\{1\}$ on a connected network;
\item[(ii)] consequently $\J$ is self-adjoint and negative semidefinite in the
weighted product $\langle u,v\rangle_{f/m}=\int u\,v\,(f/m)\,dr$;
\item[(iii)] $\ker\J=\mathrm{span}\{m\}=\mathrm{span}\{\partial_\theta g_\eq\}$, the tangent to the equilibrium family, while the \emph{cokernel}
is the mass mode, $\int\J[h]\,f\,dr=0$ for every $h$, i.e.\
$\ker\J^{\dagger}=\mathrm{span}\{1\}$.
\end{enumerate}
\end{prop}

\noindent\textit{Proof.}
Recall from~\eqref{eq:graphlap} that $\J$ is a graph-Laplacian 
built from
nonnegative, symmetric edge conductances $K_{rr'}=K_{r'r}$ on a network
of nodes (here the pore classes).  Such operators have a standard
quadratic form: for any function $u$ on the nodes,
\begin{equation}
\langle u,\Lop u\rangle_f=
-\tfrac12\sum_{rr'}K_{rr'}f_r(u_r-u_{r'})^2\le0,
\end{equation}
 a sum of squared
differences across edges that is manifestly nonpositive and vanishes
exactly when $u$ takes the same value on every node of a connected
component.  Hence $\Lop$ is negative semidefinite and, on a single
connected network, $\ker\Lop=\mathrm{span}\{1\}$, which is~(i).
Part~(ii) is immediate: with $\J=\Lop\circ(\psi_T/m)$,
$\langle u,\J v\rangle_{f/m}
=\psi_T\langle u/m,\Lop[v/m]\rangle_f$, symmetric in $u\leftrightarrow v$
and non-positive at $u=v$ by~(i).  Part~(iii): $\J[h]=0$ iff
$h/m\in\ker\Lop$, i.e.\ $h\propto m$; and differentiating the
Fermi--Dirac equilibrium at fixed shape gives
$\partial_\theta g_\eq\propto m$, so the kernel \emph{is} the tangent
to the equilibrium family.  Its orthogonal complement in
$\langle\cdot,\cdot\rangle_{f/m}$ is the zero-mass subspace, because
$\langle h,m\rangle_{f/m}=\int h\,f\,dr$: $\J$ is invertible exactly on the
fluctuations that carry no water, which is what the first-order problem
requires.  

Self-adjointness in $\langle\cdot,\cdot\rangle_f$ is the detailed-balance
symmetry of the edge weights,
\begin{equation}
K_{rr'}\,f_r=K_{r'r}\,f_{r'} ,
\label{eq:detbal}
\end{equation}
which holds identically because the edge weight factorizes through a
symmetric kernel,
\begin{equation}
K_{rr'}=\Mr(r,r';g_\eq)\,f(r'),
\qquad
\Mr(r,r')=\Mr(r',r)
\quad\text{\eqref{eq:Krr}},
\end{equation}
so that both sides of~\eqref{eq:detbal} equal
$\Mr(r,r')f(r)f(r')$; hence
$K_{rr'}f_r=\Mr(r,r')f(r')f(r)$ is symmetric in the exchange  $r$ to $r'$.
The same symmetry makes $\int\J[h]\,f\,dr
=\int\Lop[\delta\mu_w]\,f\,dr=0$ for every $h$, which is~(iii) for the
cokernel: water-mass conservation is the statement that the constants
annihilate $\J$ \emph{from the left}.\hfill$\square$

It is worth stating plainly what parts (i)--(iii) replace.  Kernel and
cokernel are \emph{distinct} here: the kernel is the tangent to the
equilibrium manifold, $\mathrm{span}\{\partial_\theta g_\eq\}$; the
cokernel is the collision invariant, $\mathrm{span}\{1\}$.  In the
Boltzmann case the two coincide, because the local Maxwellian is
parametrized by the invariants themselves, and the literature accordingly
conflates them; here they do not, because $g_\eq$ is a filling step and
not a Maxwellian.  Nothing downstream is lost by the distinction: the
Fredholm alternative tests orthogonality to the \emph{cokernel}, so the
solvability condition remains $\langle1,S\rangle_f=0$, i.e.\ exactly
water-mass conservation (Sec.~\ref{sec:solvability}).  One invariant,
one macroscopic equation.

\subsection{Spectral gap and invertibility}
\label{sec:gap}

The CE expansion requires that $\J$ be invertible on the space of
non-equilibrium perturbations, with a bounded inverse, otherwise the
first-order correction $g^{(1)}$ could not be solved for.  In network
terms this is the requirement that the slowest relaxation mode still
relax at a finite rate: there must be a \emph{gap} between the zero
eigenvalue (the conserved mass mode, which never relaxes) and the next
one.  For a graph-Laplacian this gap is the \emph{algebraic
connectivity}, the standard measure of how well connected the network
is~\cite{Fiedler1973}; it stays positive as long as the pore network
percolates, and closes as the network fragments at the percolation
threshold $\theta_c$.  

\begin{prop}[Spectral gap]
\label{prop:gap}
On a connected network and for $\theta$ bounded away from the
percolation threshold $\theta_c$, the largest nonzero eigenvalue of $\J$
is $-\lambda_1(\theta)<0$ with $\lambda_1\simeq1/\tau_{\rm relax}$,
$\tau_{\rm relax}=\max_r\tau(r)$.  Hence $\|\J^{-1}\|\lesssim\tau_{\rm relax}$
on the complement of $\ker\J$.  As $\theta\to\theta_c^+$ the conductances
across the pores population vanish through $a_w\to0$, $\lambda_1\to0$,
and $\tau_{\rm relax}\to\infty$.
\end{prop}

\noindent
There is a second, independent way for the gap to close, and it must be
stated because it constrains the reduction more tightly than percolation
does.  The conductances $K_{rr'}$ carry the occupancy factors
$m=g_\eq(1-g_\eq)$, so they are supported on the \emph{frontier band} of
partially filled classes.  Write $w:=\psi_T/|\mu_w(r^*)|$ for the
dimensionless width of that band.  As $w\to0$ the band closes, the
conductances collapse with it, and numerically (as in Supplemental Material~\cite{SM_hilbert}) we find
\begin{equation}
\lambda_1 \;\propto\; w^{\,\gamma},
\qquad \gamma\simeq1.8\text{--}1.9\quad(w\lesssim0.2) ,
\label{eq:gapwidth}
\end{equation}
over a decade in $w$ below $w\simeq0.2$; beyond $w\simeq0.2$--$0.3$ the
frontier ceases to control the slow mode, which migrates to the fine pores,
delocalizes, and saturates in rate, so the statements of this section hold
for a resolved \emph{and narrow} band.
The scaling, and the robustness of the gap itself, requires the
geometric-mean gates of Eq.~\eqref{eq:Mexplicit}: with the double product
$m\,m'$ the slowest modes migrate to the edges of the active band,
$\lambda_1$ is set by wherever the band is truncated rather than by the
physics at $r^*$, and the fitted exponent even changes sign.  The gate
content of $\Mr$ is therefore not a matter of convention.
The gap is therefore \emph{quadratic in the band width}, and the sharp
step $w\to0$ has no gap at all: $\J^{-1}$ is unbounded there and the CE
expansion is not asymptotic.  This is not a numerical inconvenience but a
statement about the medium.  The width has at least three physical
contributions, configurational fluctuations ($\psi_T$), gravitational
smearing across the REV ($\rho_w\mathrm{g}\,\ell$), and, probably
dominant, within-class geometric dispersion (throat versus body radius,
contact-angle spread, the finite width of a ``class''), and it is
measurable from micro-CT.  Accordingly we use the sharp
step~\eqref{eq:g0} for the \emph{equilibrium} statements (the retention
curve, the frontier tangent), where it is exact, and a resolved band of
finite $w$ for every statement about $\J$, its spectrum, and $K$.  The
condition $\Da\ll1$ is thereby promoted from an assumption to a
requirement on a measurable property of the soil.

\begin{rem}[the athermal limit is singular]
\label{rem:athermal}
One may instead take the transfer rates as purely viscous and rectified
to the downhill direction, $j\propto\kappa_s C\,[(\Phi)_+\,g'(1-g)
-(\Phi)_-\,g(1-g')]$, with no entropic term.  That operator has the sharp
step as an \emph{exact} equilibrium, $\C[g_\eq]\equiv0$; its linearization
is triangular in the distance from $r^*$, so its spectrum is exactly
$\{-\mathcal{A}_{\rm ath}(r)\}$ with
\begin{equation}
\mathcal{A}_{\rm ath}(r)=\!\!\int\limits_{\min(r,r^*)}^{\max(r,r^*)}\!\!
  \kappa_s(r,r')\,C(r,r')\,\rho_w\mathrm{g}\,|\Phi(r,r')|\,f(r')\,dr' ,
\label{eq:Aathermal}
\end{equation}
the radius labeling $\lambda_n=\mathcal{A}(r_n)$ becoming an identity
rather than an approximation, and $\mathcal{A}_{\rm ath}(r^*)=0$ so that
again there is no gap at the sharp step.  The two routes do not commute:
the entropic rate carries the factor $1/m$, which diverges on the wings,
whereas the athermal rate does not, so
$\lim_{\psi_T\to0}\mathcal{A}/m\neq\mathcal{A}_{\rm ath}$.  Which is
physical depends on whether the configurational entropy of the packing is
retained at the pore scale; we keep the entropic form, for which the
inner product, the spectral theorem and the Green--Kubo bracket are
available, and note that both routes agree on what matters here, the
frontier relaxes slowest, and the gap is set by the band width.
\end{rem}
\noindent\textit{Proof.}  $\J$ is a connected weighted graph-Laplacian,
spectrum $0=\lambda_0>-\lambda_1\ge\cdots$; the gap (algebraic
connectivity) is set by the slowest local rate
$\min_r[\mathcal{A}(r)/m(r)]=1/\tau_{\rm relax}$ up to an $O(1)$
connectivity constant; numerically the two agree to within $6\%$
across the band widths of Eq.~\eqref{eq:gapwidth} (Supplemental Material,
Tables~S5 and S6).\hfill$\square$

The gap is what makes the expansion asymptotic: $\J$ maps
$(\ker\J)^\perp$ boundedly to itself, so the first-order equation is
solvable there, and successive orders are bounded by powers of $\Da$.
The expansion fails exactly where the gap closes
($\theta\to\theta_c$), the  reason the reduction cannot be pushed through the
percolation threshold.

The lower bound on the connected network is controlled; the
\emph{closure} of the gap at $\theta_c$ is, by contrast, a physically
motivated result rather than a proven scaling law. 
 It is supported numerically: in the pore-network simulations of
the companion Supplemental Material the algebraic connectivity
$\lambda_1$ collapses by roughly three orders of magnitude as $\theta$ is
lowered toward $\theta_c$, and does so in lockstep with the Stokes
permeability of the same network (the two track one another in
log--log), so the vanishing of the gap and the loss of macroscopic
conductivity occur together at one threshold.  A rigorous derivation of
the critical exponent governing $\lambda_1(\theta\!\to\!\theta_c)$ is
left open.

\subsection{The relaxation spectrum: what the eigenmodes are}
\label{sec:eigenmodes}

Self-adjointness gives $\J$ a real spectrum and an
$\langle\cdot,\cdot\rangle_f$-orthonormal eigenbasis
$\{\varphi_n(r)\}_{n\ge0}$,
\begin{equation}
\J\,\varphi_n=-\lambda_n\,\varphi_n,\qquad
0=\lambda_0<\lambda_1\le\lambda_2\le\cdots,
\label{eq:eigen}
\end{equation}
the \emph{redistribution normal modes}: each $\varphi_n(r)$ is an occupancy
pattern that relaxes as $e^{-\lambda_n t}$ without change of shape, with
rate $\lambda_n=1/\tau_n$.  Being a property of $\J$ alone, the same
eigenbasis serves at every order of the hierarchy; the corrections
$g^{(n)}$ are its coefficients, not the modes themselves.  Three features
of this spectrum, and only these, carry the reduction.

\emph{One invariant.}  The kernel is one-dimensional, $\varphi_0=1$ with
$\lambda_0=0$; the entrywise statement is the row-sum identity
$\mathcal{A}(r)=\int_0^\infty K(r,r')\,dr'$ [with
$K(r,r')=\Mr(r,r';g_\eq)f(r')$, Eq.~\eqref{eq:Krr}], water conservation as
a property of $\J$.  This single invariant is why solvability yields
exactly one macroscopic equation.  Its macroscopic image is the
frontier-shift tangent $\partial_\theta g_\eq=\delta(r-r^*)/[\phi f(r^*)]$
,the one direction Richards' equation follows.

The modes are supported on the \emph{active band} of
partially filled classes straddling $r^*$; for the sharp
step~\eqref{eq:g0} the band collapses and the spectrum degenerates, and it
is the equilibrium smearing of Sec.~\ref{sec:step1} (widths $\psi_T$ and
$\rho_w\mathrm{g}\,\ell$) that gives it finite extent and a gap
$\lambda_1>0$.  The gap makes $\J^{-1}$ bounded on $(\ker\J)^\perp$, hence
the expansion asymptotic (Prop.~\ref{prop:gap}).

The conductivity inherits this spectrum directly: expanding $\J^{-1}$ on
the eigenbasis turns~\eqref{eq:K_exact} into a sum over relaxation modes,
weighted inversely by their rates [Eq.~\eqref{eq:K_spectral}], so $K$ is
regularized by the gap and, as $\theta\to\theta_c$, comes to be dominated
by the slowest mode $\varphi_1$,conductivity is set by how
fast the frontier can re-sort water among pores.  As $\theta\to\theta_c$
the gap closes, $\lambda_1\to0$, and that sum diverges, a statement about
the \emph{reduction}, not the flow: the slowest mode can no longer restore
equilibrium within a forcing time, $\Da\ll1$ fails locally, and the closed
conductivity ceases to exist.  (This is distinct from the high-$\Da$
preferential-flow regime of Sec.~\ref{sec:breakdown}, reached by fast
forcing rather than by loss of connectivity.)  At the threshold itself the
gap eigenvector is no longer the small-$r$ frontier mode but the
\emph{Fiedler vector} of the pore network~\cite{Fiedler1973}, cut by
topology rather than by radius.

A fuller characterization of the modes, their localization by pore radius,
the spectral sum rule $\psi_T\mathcal{A}(r)=f(r)\sum_{n\ge1}\lambda_n
|\varphi_n(r)|^2$ (with $\varphi_n$ orthonormal in $\langle\cdot,\cdot\rangle_{f/m}$), and the joint collapse of the gap and the permeability
at $\theta_c$,is verified numerically in the Supplemental
Material~\cite{SM_hilbert}.

\subsection{Same operator, changing source}

The structural payoff is that the whole CE hierarchy has the form
\begin{equation}
\J[g^{(n)}]=S_{n-1}[g^{(0)},\dots,g^{(n-1)}]\qquad(n\ge1),
\label{eq:hierarchy}
\end{equation}
the \emph{same} self-adjoint operator at every order, only the source
$S_{n-1}$ changing, assembled entirely from orders up to $n-1$ (the
subscript records this dependence).  This is what makes the derivation
mechanical, and it rests on a fact worth stating precisely.
\begin{prop}[Lower-order source]
\label{prop:auto}
At each order $\Da^n$ the highest correction $g^{(n)}$ enters the
expansion of $\C[g]$ exactly once, through the linear term
$\J[g^{(n)}]$; all other terms depend only on the lower orders
$g^{(0)},\dots,g^{(n-1)}$ already determined.
\end{prop}

\noindent\textit{Proof.}  Expand $\C$ in a Fr\'echet (Taylor) series about
$g^{(0)}=g_\eq$, with $\delta g=\sum_{m\ge1}\Da^{\,m}g^{(m)}$:
\begin{equation}
\begin{split}
\C[g]&=\underbrace{\C[g_\eq]}_{=0}+\underbrace{\C'[g_\eq]}_{=\J}\delta g
   +\tfrac12\C''[g_\eq](\delta g,\delta g)+\cdots\\
  &=\J[\delta g]+\tfrac12\C''(\delta g,\delta g)+\cdots
\end{split}
\label{eq:Ctaylor}
\end{equation}
The linear term gives $\J[\delta g]=\sum_n\Da^{\,n}\J[g^{(n)}]$.  The
quadratic term begins at $\Da^2$ and at order $\Da^n$ contains only
products $g^{(j)}g^{(k)}$ with $j,k\le n-1$; the cubic term begins at
$\Da^3$; and so on.  Hence $g^{(n)}$ appears at order $\Da^n$ only in
$\J[g^{(n)}]$, never trapped inside a product.  Matching order $\Da^n$
in~\eqref{eq:singular} (whose $\Da^{-1}$ shifts the redistribution down
one order) gives $\J[g^{(n)}]=S_{n-1}[g^{(0)},\dots,g^{(n-1)}]$, which is
\eqref{eq:hierarchy}.\hfill$\square$

\medskip
\noindent
One therefore inverts the single fixed operator $\J$ repeatedly, the
source on the right assembled from what is already known from the previous iteration.  Solvability of
each member is the orthogonality of $S_{n-1}$ to $\ker\J$, and that
orthogonality is the macroscopic balance.  We now carry out $n=0,1$.

Equivalently, the whole reduction can be written in one step as the
\emph{Schur complement}  given in App.~\ref{app:schur}.

\section{The Chapman--Enskog hierarchy}
\label{sec:hierarchy}

The zeroth order and the equilibrium manifold were developed in
Sec.~\ref{sec:hilbert_ce}; with the operator $\J$ now in hand
(Sec.~\ref{sec:operator}), we carry out the first order and its
solvability, which together produce the macroscopic equation.

\subsection{First order: the linearized budget and the source}
\label{sec:step2}

At order $\Da^0$ the streaming term enters and the linearized
redistribution responds.  It is convenient to fix here a term used
repeatedly below: a pore class is said to be \emph{slaved} to
equilibrium when its occupancy is not an independent degree of freedom
but a function of the macroscopic state alone,
$g(r)=g_\eq(r;\theta)$, redistribution restores equilibrium faster than
the forcing can disturb it, so the class carries no memory of its own.
The whole of the CE reduction is the statement that, when the gap is
wide, every class is slaved except through the first-order correction
$g^{(1)}$.  Inserting the
ansatz~\eqref{eq:ansatz} into~\eqref{eq:singular} and collecting
$O(\Da^0)$ terms,
\begin{equation}
\boxed{\;
\J[g^{(1)}](r)=\partial_t g^{(0)}+\nabla\!\cdot\!\mathbf{F}[g^{(0)}]
  +\mathcal{E}+\mathcal{T}\;}
\label{eq:first_order}
\end{equation}
This is~\eqref{eq:hierarchy} with $n=1$.  The left side is the linearized
budget of Sec.~\ref{sec:operator}; the right side is the source, built
from two pieces evaluated on the equilibrium step.

\emph{The inter-REV source} is the divergence of the entangled
flux~\eqref{eq:Fentangled} evaluated at $g^{(0)}$.  The neighbour
expansion
\begin{equation}
g^{(0)}(r';\theta(\mathbf{x}+\delta\mathbf{x}))\simeq H(r^*\!-r')
+\delta(r'-r^*)[\phi f(r^*)]^{-1}\nabla\theta\!\cdot\!\delta\mathbf{x}
\end{equation}
collapses the $r'$-integral onto the frontier $r'=r^*$
(App.~\ref{app:gradient}); with the gate cancellation that makes the
relaxation rate~\eqref{eq:A} non-negative,
$\big|[1-2g^{(0)}(r)]\,\Phi(r,r^*)\big|=|\Phi(r,r^*)|$ (only the magnitude
enters the reduction; the physical sign of $\mathbf q$, down the head gradient,
is fixed at the identification of $K$), this gives
\begin{equation}
S(r;\nabla\theta)=\kappa(r)\,\frac{C_\partial(r,r^*)}{f(r^*)}\,
  \big|\Phi(r,r^*)\big|\;\nabla\theta\!\cdot\!\hat n .
\label{eq:source}
\end{equation}
Only pores connected to the frontier feel the gradient: $r^*$ is the
gateway through which it enters each REV.

Since $g^{(0)}=H(r^*(\theta)-r)$,
\begin{equation}
\partial_t g^{(0)}=\frac{\delta(r-r^*)}{\phi f(r^*)}\,\partial_t\theta
=\Da\,\frac{\delta(r-r^*)}{\phi f(r^*)}\,\partial_T\theta
\label{eq:dtg0}
\end{equation}
on the slow clock~\eqref{eq:slowtime}, one order below the $O(1)$ source;
by the slaving principle the correction's own derivative is smaller still,
$\Da\,\partial_t g^{(1)}=O(\Da^2)$ (App.~\ref{app:spectral}).  So at this
order the time term is carried entirely by \eqref{eq:dtg0}.  The
constraint $\int g^{(1)}f\,dr=0$~\eqref{eq:constraint} places $g^{(1)}$ in
$(\ker\J)^\perp$, where $\J$ inverts (Prop.~\ref{prop:gap}).

Every term now carries a single power of $\Da$,the left side through
\eqref{eq:dtg0}, the internal operator because its $O(\Da^0)$ part
vanished at zeroth order ($\mathcal{G}_{\rm int}-\Lg_{\rm int}
=\Da\,\J[g^{(1)}]+O(\Da^2)$), and the inter-REV exchange by
construction, so dividing through,
\begin{equation}
\frac{\delta(r-r^*)}{\phi f(r^*)}\,\pder{\theta}{T}
  +S(r,\nabla\theta)=\J[g^{(1)}](r),
\label{eq:combine}
\end{equation}
a Richards equation in $g^{(1)}$, made macroscopic by the solvability
integration below.

\subsection{Solvability is mass conservation}
\label{sec:solvability}

As previewed in Sec.~\ref{sec:hilbert_ce},
Eq.~\eqref{eq:first_order} has a solution only if its right side is
orthogonal to $\ker\J=\mathrm{span}\{1\}$ in the $f$-measure, the
Fredholm alternative~\cite{Fredholm1903} for the self-adjoint $\J$.
This involves the $f$-weighted moment of~\eqref{eq:first_order}, and:
\begin{itemize}
\item $\int\J[g^{(1)}]\,f\,dr=0$ (null-space property), annihilating the
  redistribution term;
\item the $\delta$-function in $\partial_t g^{(0)}$ integrates to
  $(1/\phi)\,\partial_t\theta$;
\item the source integrates to
  $\int S\,f\,dr=(1/\phi)\,\nabla\!\cdot\!\mathbf{q}$, defining the
  divergence of the macroscopic flux $\mathbf{q}$ (see App.~\ref{app:fluxid} for its identification of
  $\mathbf{q}$);
\item the sinks, $\mathcal{E}$ and $\mathcal{T}$  integrate to $S_{\rm evap}+S_{\rm root}$.
\end{itemize}
Therefore, from \eqref{eq:first_order} we obtain:
\begin{equation}
\boxed{\;
\partial_t\theta+\nabla\!\cdot\!\mathbf{q}=-S_{\rm evap}-S_{\rm root}\;}
\label{eq:continuity}
\end{equation}
The mathematical requirement for $g^{(1)}$ to exist \emph{is}
the macroscopic continuity equation.  It
\emph{must} hold for the expansion to be self-consistent.  The single
invariant $\theta$ yields this single macroscopic equation, the soil
analogue of the inviscid (mass-conservation) balance that closes the
leading order.

\section{The response function and the hydraulic conductivity}
\label{sec:K}

\subsection{The response function}
Having determined the properties of $\J$, we now solve the first-order equation~\eqref{eq:first_order} for the
correction $g^{(1)}$.  Two simplifications apply at this order.  First,
the time derivative is removed by the slow-time ordering
already in force.  On the slow clock $T=\Da\,t$~\eqref{eq:slowtime} the
sole conserved field evolves as $\partial_t\theta=\Da\,\partial_T\theta$,
so the only time-dependent part of $g^{(0)}=g_\eq(r;\theta)$ is
\[
\partial_t g^{(0)}
  =\frac{\partial g_\eq}{\partial\theta}\,\partial_t\theta
  =\Da\,\frac{\partial g_\eq}{\partial\theta}\,\partial_T\theta
  =O(\Da),
\]
smaller by one power of $\Da$ than the $O(1)$ inter-REV
source~\eqref{eq:source}, hence subdominant.  The time derivative term is
therefore absent from the leading first-order balance
(Sec.~\ref{sec:solvability}; the full order-by-order ladder, and the
order at which the derivative returns, is in App.~\ref{app:dtorder}).

Second, the surviving source is proportional to the macroscopic
gradient.  With these two simplifications Eq.~\eqref{eq:first_order}
reduces to
\begin{equation}
\J[g^{(1)}]=S ,
\label{eq:first_reduced}
\end{equation}
and, since $S$ is linear in $\nabla\theta=C_w\nabla\psi$ with
$C_w=d\theta/d\psi$ the specific moisture capacity, the correction is
linear in the macroscopic gradient too,
\begin{equation}
g^{(1)}=-\bm{\chi}(r)\!\cdot\!\nabla\psi ,
\label{eq:g1chi}
\end{equation}
where the \emph{response function} $\bm{\chi}$ obeys
\begin{equation}
[\mathcal{A}-\mathcal{B}]\,\bm\chi=S_0\,C_w,
\qquad\text{equivalently}\qquad
\J[\bm{\chi}]=-\,S_0\,C_w .
\label{eq:chi}
\end{equation}
It is worth writing the inversion out.  By the split~\eqref{eq:J} the
linearized operator is a local loss rate minus a redistribution kernel,
\begin{equation}
\J=-\mathcal{A}+\mathcal{B}=-[\mathcal{A}-\mathcal{B}],
\end{equation}
so that $-\J=\mathcal{A}-\mathcal{B}$ is the \emph{relaxation operator},
positive semidefinite by the quadratic form of Sec.~\ref{sec:operator}
and positive definite on
$W:=(\ker\J)^\perp$, where the spectral gap makes it boundedly
invertible (Prop.~\ref{prop:gap}).  Substituting
$g^{(1)}=-\bm\chi\!\cdot\!\nabla\psi$ into the first-order
balance~\eqref{eq:first_reduced}, with
$S=S_0C_w\,\nabla\psi\!\cdot\!\hat n$,
and using the linearity of $\J$,
\begin{equation}
-\J[\bm\chi]\!\cdot\!\nabla\psi=S_0C_w\,\nabla\psi\!\cdot\!\hat n
\qquad\Longrightarrow\qquad
[\mathcal{A}-\mathcal{B}]\,\bm\chi=S_0C_w ,
\end{equation}
which is~\eqref{eq:chi}; inverting,
\begin{equation}
\bm\chi=[\mathcal{A}-\mathcal{B}]^{-1}[S_0C_w]
\;=\;(-\J)^{-1}[S_0C_w]\big|_W .
\label{eq:chisolved}
\end{equation}
Two properties of this inversion should be separated.  The
\emph{convention} is fixed by consistency: with $\mathcal{A}-\mathcal{B}$
positive definite on $W$, Eq.~\eqref{eq:chisolved} and the spectral
sum~\eqref{eq:K_spectral} denote the same object, whereas inverting $\J$
itself would flip the sign of every term of that sum.  The
\emph{positivity} of $K$ is a separate statement and does not follow from
the positivity of the source alone: it requires the conductance $\kappa$
and the response $\bm\chi$ to be positively correlated in the mass
measure.  That correlation is manifest in the mean-field limit, where
$\bm\chi_{\rm mf}=S_0C_w/\mathcal{A}\ge0$ class by class and $K_{\rm mf}$
of Eq.~\eqref{eq:K_mf} is a positive integral, and it is the physical
content of the fluctuation--dissipation structure discussed in
Sec.~\ref{sec:K}: a class that conducts well is also a class the
gradient drives.  Off the mean field the sign is not automatic, and the
spectral form~\eqref{eq:K_spectral} is the statement to check.  This is the soil-water analogue of the
gas response functions $A(c;T),B(c;T)$ that solve
$L^{-1}$ applied to the stress/heat polynomial sources
\cite[Eqs.~(V.3.18)]{Cercignani1988}.

\subsection{Conductivity as a first-order transport coefficient}

With the response function in hand we can now close the calculations left
open above: the macroscopic flux, the conductivity, and, through it, the
entangled kernel of Sec.~\ref{sec:gradflux}.
Each pore class carries water at rate $\kappa(r)g^{(1)}(r)$.  Integrating,
\begin{equation}
\phi\!\int_0^\infty\!\kappa(r)g^{(1)}(r)f(r)\,dr
  =-\Big[\phi\!\int_0^\infty\!\kappa(r)\bm{\chi}(r)f(r)\,dr\Big]\!
  \cdot\!\nabla\psi.
\end{equation}
The macroscopic flux is this $\kappa$-weighted moment of the deviation
$g^{(1)}$,the throughflow the classes actually carry, as in the
kinetic theory of gases, where the transport fluxes are moments of the
first-order correction and not of the streaming of the local equilibrium
(App.~\ref{app:fluxid}).  Reading it off,
\begin{equation}
\boxed{\;
K(\theta)=\phi\int_0^{r_{\max}}\!\!\kappa(r)\,\bm{\chi}(r)\,f(r)\,dr
  =\phi\,\langle\kappa\,|\,\J^{-1}\,|\,S_0 C_w\rangle_f\;}
\label{eq:K_exact}
\end{equation}
Two points of notation.  First, the upper limit is the largest pore
radius of the medium, $r_{\max}=\sup\{r:f(r)>0\}$; we write $\infty$
elsewhere when $f$ is taken with unbounded support, the integrand
vanishing beyond $r_{\max}$ in either case.  Second, on \emph{what} $K$
depends: for a given soil it is a \emph{function} of the water content
alone, $K=K(\theta)$, the $\theta$-dependence entering through the
frontier $r^*(\theta)$ that fixes both $\bm\chi$ and the active band;
across soils it is a \emph{functional} of the pore-size density and of
the connectivity kernel, $K=K[f,C_\partial](\theta)$, and it is in this
second sense that the classical closures of Sec.~\ref{sec:classical} are
approximations, they posit $f$-dependence and drop the $C_\partial$
dependence.  Third, a sign convention: since $\J$ is negative
semidefinite, we write $\J^{-1}$ throughout for the inverse of the
\emph{relaxation} operator $-\J=\mathcal{A}-\mathcal{B}$ restricted to
$W$, whose spectral form
$\J^{-1}|_W=\sum_{n\ge1}\lambda_n^{-1}\varphi_n\otimes\varphi_n$
(App.~\ref{app:spectral}) is positive definite; with this convention
Eq.~\eqref{eq:K_exact} and the spectral sum~\eqref{eq:K_spectral} below
are manifestly positive, as a transport coefficient must be.
$K$ is the PSD-averaged pore-class conductivity, the structural analogue
of the CE viscosity
$\mu_{\rm gas}=\langle\mathbf{c}\,|\,L^{-1}\,|\,\mathbf{c}\rangle$
\cite[Eqs.~(V.3.20)--(V.3.21)]{Cercignani1988}.  The bracket
$\langle\kappa|\J^{-1}|S_0C_w\rangle_f$ is shorthand for ``apply the
inverse operator $\J^{-1}$ to the source $S_0C_w$, then take the
$\kappa$-weighted moment in the measure $f\,dr$'',it is the
conductivity expressed as the overlap of the pore conductivity $\kappa$
with the relaxed response to the gradient.  Expanding $\J^{-1}$ in its
own eigenbasis $\{\varphi_n,\lambda_n\}$ (the natural relaxation modes of
the pore network, with rates $|\lambda_n|$) turns this overlap into a
sum over modes (Fig.~\ref{fig:spectrum}).  Explicitly, the response is the
source resolved on that basis, inversely weighted by the rates,
\begin{equation}
\bm\chi=\sum_{n\ge1}\frac{\langle\varphi_n,S_0C_w\rangle_f}{\lambda_n}\,
   \varphi_n ,
\qquad
g^{(1)}=-\bm\chi\!\cdot\!\nabla\psi ,
\label{eq:g1spectral}
\end{equation}
\begin{figure*}[t]
\centering
\includegraphics[width=\textwidth]{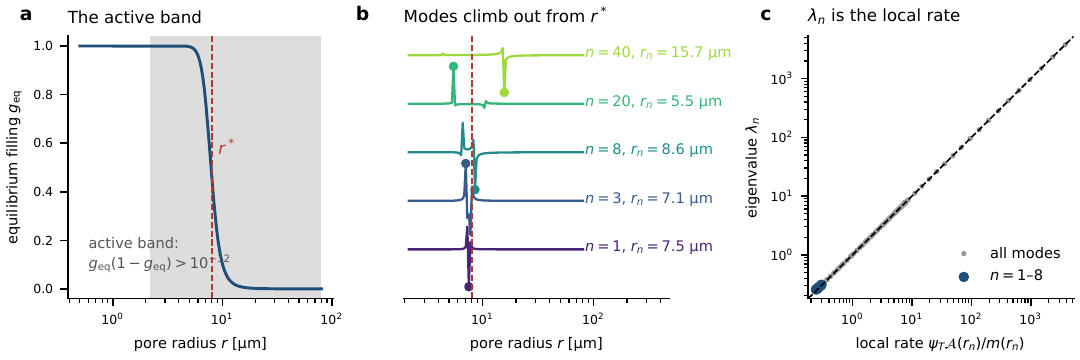}
\caption{The spectral structure behind the conductivity formula~\eqref{eq:K_exact}, computed for a representative log-normal loam (Supplemental Material, Sec.~S1). \textbf{(a)} Equilibrium filling $g_\eq(r)$ and the \emph{active band} of partially filled pore classes that carry the redistribution dynamics. \textbf{(b)} Five representative eigenmodes $\varphi_n$ of the linearized operator $\J$, offset vertically: each is localized on two or three adjacent pore classes ($\mathrm{PR}_n\simeq2$--$3$) and pinned to its own resonant radius $r_n$; the slowest sit at the frontier $r^*$ and higher modes move outward into the wings, so the mode index carries a distance from the frontier. \textbf{(c)} The resulting radius labeling $\lambda_n\simeq\psi_T\A(r_n)/m(r_n)$: each decay rate equals the local rate at that mode's radius to within $2$--$6\%$. The conductivity $K=\phi\langle\kappa|\J^{-1}|S\rangle_f$ resolves the source on this graded ladder, inversely weighted by the rates $\lambda_n$ [Eq.~\eqref{eq:g1spectral}]; the spectral gap $\lambda_1$ regularizes the sum; as $\lambda_1\to0$ at the percolation threshold the closed formula ceases to exist, while the physical $K$ vanishes through the collapse of the overlaps $\langle\kappa,\varphi_1\rangle_f$ onto the backbone (Sec.~\ref{sec:breakdown}).}
\label{fig:spectrum}
\end{figure*}

the $\lambda_0=0$ mode excluded by solvability; taking its
$\kappa$-moment gives
\begin{equation}
K=\phi\sum_{n\ge1}
  \frac{\langle\kappa,\varphi_n\rangle_f\;
        \langle\varphi_n,S_0C_w\rangle_f}{\lambda_n},
\label{eq:K_spectral}
\end{equation}
the soil-water analogue of a Green--Kubo formula~\cite{Kubo1957}: each mode contributes
in proportion to its overlap with $\kappa$ and \emph{inversely} to its
relaxation rate.  Away from threshold the two weightings compete, the
slowest modes have the largest $1/\lambda_n$ but the smallest
$\kappa$-overlap, since they live at the smallest radii, so $K$ draws on
many modes.  As $\theta\to\theta_c$, however, $\lambda_1\to0$ outruns the
overlap and the slowest mode comes to dominate: the conductivity is then
carried by the critical path through the network
bottleneck~\cite{Hunt2004}. 
When the source is aligned with the
conductance, $S_0C_w\propto\kappa$, the sum takes the manifestly positive
$|\langle\kappa,\varphi_n\rangle_f|^2/\lambda_n$ form familiar from
fluctuation--dissipation theory.  As $\theta\to\theta_c$ that slowest rate
$\lambda_1\to0$, the sum diverges, and the expansion breaks down.
At the threshold the slow modes localize on
poorly connected regions and the overlaps
$\langle\kappa,\varphi_1\rangle_f$ collapse onto the vanishing
backbone, so the physical $K$ vanishes at the topological transition
[Sec.~\ref{sec:breakdown}(c)] while the leading term of
Eq.~\eqref{eq:K_spectral} is no longer slaved and the CE
representation of $K$ ceases to exist.  The competition between the
two scalings,$1/\lambda_1$ against the collapsing overlap, is what
fixes the critical exponent of $K(\theta)$ near $\theta_c$, a
calculation we defer to future work.

\subsection{Mean-field reduction and the heterogeneity penalty}
\label{sec:meanfield}

Three reductions collapse the exact $K$ to the standard form:
\textbf{(1)~statistical homogeneity}, $C_\partial\approx C$;
\textbf{(2)~diagonal dominance}, $\mathcal{B}\approx0$, so
$\bm{\chi}_{\rm mf}=S_0 C_w/\mathcal{A}$ (each pore class responds to the
gradient independently of the others, the off-diagonal coupling through
$\mathcal{B}$ being neglected, the assumption underlying the
Burdine-type integral~\cite{Burdine1953}; 
\textbf{(3)~isotropic
orientation}, tortuosity $\bar\tau^2$.  The result is the mean-field
conductivity
\begin{equation}
K_{\rm mf}[\theta]=\frac{\phi\,\rho_w\mathrm{g}}{8\mu\,\bar\tau^2}
  \int_0^{r^*(\theta)}\!\! r^2\,\Xi(r)\,a_w(r)\,\bar C(r)\,f(r)\,dr,
\label{eq:K_mf}
\end{equation}
every factor inherited from~\eqref{eq:K_exact}: $r^2$ from $\kappa$,
$\bar\tau^{-2}$ from orientation averaging, $\bar C(r)$ from
$C_\partial\approx C$, the percolation accessibility $a_w(r)$ from the
spectral weight (so $K_{\rm mf}\to0$ below threshold), the composite
geometry factor $\Xi(r)$ from the pore-scale corrections, and the
frontier cutoff $r^*$ from the localization~\eqref{eq:source}.  This is
the formula stated without derivation as Eq.~(57) of the companion
paper.

Beyond diagonal dominance, the off-diagonal $\mathcal{B}$ correlates the
responses of different classes: large pores must push water through
small-pore bottlenecks, lowering $\bm\chi$ for large $r$ and raising it
for small $r$.  Under Mualem's random-pairing hypothesis
$C=C_0$~\cite{Mualem1976}, this correlation can be summed exactly.  Writing
$[\mathcal{A}-\mathcal{B}]^{-1}
=\sum_{k\ge0}(\mathcal{A}^{-1}\mathcal{B})^k\mathcal{A}^{-1}$,the
\emph{Neumann series}, the operator version of the geometric series
$(1-x)^{-1}=\sum_k x^k$, convergent because $\mathcal{B}$ is subordinate
to $\mathcal{A}$,each term adds one more pass through a bottleneck, and
the sum collapses to a harmonic mean (App.~\ref{app:chi}):
\begin{equation}
K_{\rm sp}=\frac{\theta\,\rho_w\mathrm{g}}{8\mu\,\bar\tau^2}
  \cdot\frac{1}{\langle r^{-2}\rangle_g},
\quad
\langle r^{-2}\rangle_g=\frac{\int_0^{r^*}r^{-2}f\,dr}{\int_0^{r^*}f\,dr}.
\label{eq:K_sp}
\end{equation}
For a log-normal PSD, $K_{\rm sp}/K_{\rm mf}=\exp(-4\sigma^2)$:
\textbf{Mualem's heterogeneity penalty}, derived here as the exact
serial-path resummation of the CE transport coefficient under mean-field
connectivity rather than introduced empirically.

Because $\J$ is self-adjoint and sign-definite, the bracket
$K=\phi\langle\kappa|\J^{-1}|S_0C_w\rangle_f$ also admits a variational
characterization (Table~\ref{tab:synthesis}, the Bhatnagar entry), and
this reframes the classical closures: Burdine- and Mualem-type formulas
are \emph{trial response functions} inserted into the variational
functional, hence one-sided bounds on $K$ with errors \emph{quadratic}
in the trial error, improvable by a Sonine-like basis expansion exactly
as gas-kinetic viscosities are computed to percent accuracy from two or
three Sonine polynomials~\cite[Ch.~7]{ChapmanCowling1970}, the
variational form of that construction being due to
Bhatnagar~\cite{Hirschfelder1954,Cercignani1988}.  The same structure
yields per-soil error
certificates for any closure (the residual of the cell problem,
evaluated on the trial function), a controlled resummation of the
heterogeneity penalty in the PSD width $\sigma$, a predicted, rather
than fitted, tortuosity exponent, and a frequency-dependent
conductivity $K(\omega)$, the response to a forcing oscillating at
angular frequency $\omega$, whose poles are the relaxation rates
$\lambda_n$ and whose zero-frequency limit is the
bracket~\eqref{eq:K_exact}.  These
developments exceed the scope of this paper and are pursued in a
dedicated companion on conductivity closures, in preparation.

\section{Beyond capillarity: gravity, adsorption, osmosis}
\label{sec:forces}

The derivation, so far,  used the capillary $\Phi$, but the full chemical
potential is
\begin{equation}
\mu_w(r)=\mu_w^0-\frac{2\gamma\cos\theta_c}{r}+\Pi(h(r))
  -\nu RT\,c(r)+\rho_w\mathrm{g}\,z,
\label{eq:mu_w}
\end{equation}
where $\mu_w^0$ [Pa] is a reference potential, $\gamma$ [N\,m$^{-1}$] the
surface tension and $\theta_c$ the contact angle (capillary term),
$\Pi(h)$ [Pa] the disjoining pressure of a film of thickness $h$
(adsorptive), $\nu$ the van~'t~Hoff factor, $R$ [J\,K$^{-1}$mol$^{-1}$]
the gas constant, $T$ [K] the temperature and $c(r)$ [mol\,m$^{-3}$] the
solute concentration (osmotic), and $\rho_w\mathrm{g}z$ [Pa] the elevation
term (gravitational).  Since
$\Phi=[\mu_w(r')-\mu_w(r)]/(\rho_w\mathrm{g})$, every force enters
\emph{through $\Phi$ and only through $\Phi$}, and the CE structure, a
self-adjoint $\J$ with a single invariant, is preserved
(companion paper~\cite{Rigon2026PRE}).  The reduction separates scales cleanly:

 \textbf{Gravity.}  Gravity is of course not a global quantity: it acts
everywhere, and $\rho_w\mathrm{g}z$ varies from point to point.  The
statement here is weaker and purely local.  At a \emph{given} location
$\mathbf{x}$, the elevation term takes the \emph{same} value for every
pore class, because all classes of one REV sit at the same $z$: it is a
common-mode contribution to $\mu_w$ across the class axis.  The
intra-REV operator $\C$ depends on $\mu_w$ only through
\emph{differences between classes at the same location}
[Eq.~\eqref{eq:Phi_def}], and a common-mode term cancels from every such
difference.  Hence gravity does not redistribute water \emph{among the
pore classes of one REV}, it cannot make a large pore drain into a
small one at fixed $z$.  What gravity does drive is transport
\emph{between} REVs, where the elevation difference no longer cancels,
and it therefore enters only the inter-REV potential.  It contributes only to the
\emph{inter}-REV potential, where the difference is taken between
neighbouring REVs at different elevations, adding the elevation gradient
$\nabla z=\hat z$ to the driving force.  Quantitatively, over one REV of
size $\bar L\sim$\,mm the gravitational head difference
$\rho_w\mathrm{g}\bar L$ is negligible against the capillary heads
$2\gamma/r$ of the pores that control redistribution
($r\lesssim\ell_c$, the capillary length); over the macroscopic scale
$\Lambda$ the accumulated head $\rho_w\mathrm{g}\Lambda$ is not, and
gravity drives flow.  The net effect is the familiar replacement of the
driving force $\nabla\psi$ by the total-head gradient
$\nabla H=\nabla\psi+\hat z$, with the transport coefficient $K$
unchanged: gravity moves water but does not alter how easily the medium
conducts it.

\textbf{Adsorption.}  The disjoining pressure $\Pi(h)$ adds an adsorptive
branch to $g^{(0)}$ (the retention dry-end tail and a finite
$\theta_r$, Sec.~\ref{sec:classical}) and replaces tube flow by film
flow $\kappa_{\rm film}\propto h^3/\mu$ in $\kappa(r)$; the formal
structure of~\eqref{eq:K_exact} persists.

\textbf{Osmosis.}  The term $-\nu RT\,c(r)$ adds an osmotic contribution;
the inter-REV driving force becomes
$\nabla\Psi=\nabla\psi+\hat z+\sigma_{\rm refl}\nabla\pi_{\rm osm}$ with
$\sigma_{\rm refl}$ the reflection coefficient, coupling a solute balance
to the water kinetic equation.

\section{Richards' equation}
\label{sec:Richards}

Substituting into the macroscopic budget~\eqref{eq:continuity} the
Darcy flux identified in App.~\ref{app:fluxid} as the
$\kappa$-weighted moment of the deviation, Eq.~\eqref{eq:app_qtr},
disentangled from the equilibrium streaming by the gradient algebra of
App.~\ref{app:gradient} and extended to gravity by
Sec.~\ref{sec:forces},
\begin{equation}
\mathbf{q}=-K(\theta)\,(\nabla\psi+\hat z),
\label{eq:darcyflux}
\end{equation}
one obtains
\begin{equation}
\boxed{\;
\partial_t\theta=\nabla\!\cdot\![K(\psi)(\nabla\psi+\hat z)]
  -S_{\rm evap}-S_{\rm root}\;}
\label{eq:Richards}
\end{equation}
Richards' equation~\cite{Richards1931}, now \textbf{derived} from the
kinetic equation as the $\Da\to0$ reduction of the continuum kinetic
equation~\eqref{eq:contKE2}.  All constitutive relations follow from
pore-network properties: $\psi(\theta)$ from $f(r)$ alone; $K(\psi)$
from $f(r),C(r,r'),a_w(r),\bar\tau$ and the fluid constants.  The
derivation requires (i)~$\Da\ll1$ (fast redistribution: local
equilibrium and smooth gradients); (ii)~statistical homogeneity
$C_\partial\approx C$; (iii)~slowly varying $a_w$ across REVs.

The step $g^{(0)}=g_\eq$ is unique at each $\theta$, but $a_w(r)$ differs
between wetting and drying, so $K^{\rm dry}\neq K^{\rm wet}$: parametric
hysteresis survives the reduction as branch selection within the
constitutive relation, the flattened remnant of the full path
dependence carried by $g(r)$.  Off the manifold ($\Da\sim1$), $K$ is a
functional of $g$: two states at the same $\theta$ can have very
different $K$ because the $r^2$ weighting amplifies large pores,
explaining the lab--field discrepancy~\cite{Diamantopoulos2012}.

\section{Recovery of the classical theory}
\label{sec:classical}

The CE limit must reproduce the established results of unsaturated-zone
hydrology.  We record the recovery compactly here, on the equilibrium
manifold $\mathcal{S}_\eq$ (i.e.\ at $\Da\ll1$); the full numerical
demonstration against pore-network simulation and the proposed
experimental tests are given in the validation supplement of the
companion paper~\cite{Rigon2026PRE} (its SM-A), to which we refer rather
than duplicating them.  Each entry below holds on $\mathcal{S}_\eq$, and
Table~\ref{tab:recovery} collects the correspondence.

\subsection{Equilibrium thermostatics}

At $\partial_t g=0$ with no macroscopic flow the kinetic equation
reduces to detailed balance and $g_\eq=H(r^*\!-r)$
[Eq.~\eqref{eq:g0}].  The retention curve is then the cumulative PSD
transformed by the capillary law~\cite{Haines1930},the inverse of the
frontier relation~\eqref{eq:psi}, written here as $\theta(\psi)$ for use
in the constitutive identifications below,
\begin{equation}
\theta(\psi)=\phi\int_0^{r^*(\psi)}f(r)\,dr,
\qquad r^*(\psi)=-\frac{2\gamma\cos\theta_c}{\rho_w\mathrm{g}\,\psi}.
\label{eq:retention}
\end{equation}

\textit{Brooks--Corey} from a power-law PSD $f\propto r^{-\beta}$:
$\theta(\psi)\propto(|\psi_e|/|\psi|)^{1-\beta}$, identifying the
Brooks--Corey index $\lambda=1-\beta$ and the air-entry
$\psi_e=-2\gamma\cos\theta_c/(\rho_w\mathrm{g}\,r_{\max})$
\cite{Brooks1964,Rawls1982}.

\textit{Kosugi/van~Genuchten} from a log-normal PSD: $\theta(\psi)$ is
the standard-normal CDF in $\ln|\psi|$~\cite{Kosugi1996}, numerically
indistinguishable from van~Genuchten~\cite{vanGenuchten1980} for
$n\gtrsim1.5$ with $\alpha\approx1/|\psi_m|$,
$n\approx1.26/\sigma+1$.  The empirical van~Genuchten curve is thus
\emph{equivalent to assuming a log-normal $f(r)$}.  The pure-capillary
Kosugi form gives $\theta_r=0$; the disjoining-pressure branch
$\Pi(h)$ in~\eqref{eq:mu_w} supplies the dry-end adsorptive tail and a
finite $\theta_r$ fixed by the Hamaker constant and $f(r)$ rather than
fitted~\cite{Tuller1999,OrTuller2000}.

\subsection{Steady transport}

At $g=1$, $a_w=1$,
Eq.~\eqref{eq:K_mf} gives the saturated conductivity,
$K_{\rm sat}=(\phi\rho_w\mathrm{g}/8\mu\bar\tau^2)\int\Xi r^2\bar C
f\,dr$; for a log-normal PSD the independent-pore (Burdine) estimate
scales as
$r_m^2\exp(2\sigma^2)$, reduced by the heterogeneity penalty
$\exp(-4\sigma^2)$~\eqref{eq:K_sp} to values within the literature
ranges across textures (sand to clay).

In the CE limit
$\mathbf{q}=-K(\psi)(\nabla\psi+\hat z)$; at saturation
$\mathbf{q}=-K_{\rm sat}\nabla h$, exactly linear because $K$ is
independent of $\nabla h$ when $g=1$.

At saturation the reduction returns the saturated conductivity directly.
Setting $g_\eq\equiv1$ removes the frontier cutoff in~\eqref{eq:K_mf}, so
the integral runs over the whole distribution and the accessibility
saturates, $a_w\to1$:
\begin{equation}
K_{\rm sat}=\frac{\phi\,\rho_w\mathrm{g}}{8\mu\,\bar\tau^2}
  \int_0^\infty r^2\,\Xi(r)\,\bar C(r)\,f(r)\,dr ,
\label{eq:Ksat}
\end{equation}
the classical $r^2$-weighted (Kozeny--Carman-like) form: $K_{\rm sat}$ is
the second moment of the PSD, corrected by tortuosity and connectivity.
Everything specific to unsaturated flow, the cutoff at $r^*$, the
accessibility $a_w$, the serial-path penalty, is a \emph{reduction} of
this quantity, so $k_r=K/K_{\rm sat}\in[0,1]$ follows with
$k_r(\theta_s)=1$ by construction.  Away from saturation,
$K(\theta)$ from~\eqref{eq:K_mf}
declines steeply through three factors: the receding frontier $r^*$
(with $r^2$ weighting), the accessibility cutoff $a_w(r;\theta)\to0$
near $\theta_c$, and the geometry factor $\Xi(r)$.  Near threshold
$K\propto(\theta-\theta_c)^\mu$, $\mu\approx2$ (3D universal exponent),
governed by the topological cutoff and hence common to the present
$r^2$ weighting and the Mualem $1/r$ weighting.

%

\begin{table}[t]
\caption{\label{tab:recovery} Classical results recovered as the CE
($\Da\to0$) limit.}
\begin{ruledtabular}
\begin{tabular}{ll}
\textbf{Classical result} & \textbf{CE recovery} \\
\hline
Young--Laplace & $\psi=-2\gamma\cos\theta_c/(\rho_w\mathrm{g}r^*)$ \\
Retention curve & CDF of $f(r)$ via $r^*(\psi)$ \\
Brooks--Corey / Kosugi & power-law / log-normal $f(r)$ \\
Darcy's law & CE limit $\mathbf{q}=-K\nabla H$ \\
$K_{\rm sat}$, $K(\theta)$ & $r^2 a_w\bar C f$ integral; $\exp(-4\sigma^2)$ \\
Richards' equation & $\Da\to0$ reduction of \eqref{eq:contKE2} \\
$K\propto(\theta-\theta_c)^\mu$ & percolation exponent $\mu\approx2$ \\
Field capacity & $\theta_{\rm FC}\approx\theta_c$ (percolation) \\
Mualem penalty & serial-path Neumann resummation \\
\end{tabular}
\end{ruledtabular}
\end{table}

\section{Dual- and multiple-permeability models as a multi-band
reduction}
\label{sec:multiperm}

The CE reduction of Secs.~\ref{sec:hierarchy}--\ref{sec:Richards}
assumed a \emph{single} fast time scale, one spectral gap
$|\lambda_1|$ cleanly separating intra-REV redistribution from
macroscopic forcing, so that every pore class could be slaved to the
same equilibrium step, every class \emph{slaved} in the sense of
Sec.~\ref{sec:step2}.

This subsection identifies precisely when that
assumption fails and shows that its failure is not a breakdown but a
\emph{richer} reduction: when the relaxation spectrum
has more than one gap, the same expansion yields not several parallel
models but a \emph{nested family}: cutting the spectrum at the first
gap gives the one-mode Richards equation; cutting at the second gives
two coupled Richards equations with an exchange term; and so on.  The
validity of each member rests on the \emph{single} gap separating its
retained modes from the slaved remainder~\cite{Roberts2025tma}; the
additional, deeper gaps serve only to close each band internally, and
each member reduces to the next by adiabatic elimination of its fastest
retained mode.  The hierarchy thereby reproduces the dual- and
multiple-permeability models of the hydrological
literature~\cite{GerkeVanGenuchten1993,Weiler2003,Simunek2003} as
controlled limits of the one kinetic equation rather than as
phenomenological constructs.

The natural way to make this precise is through the Damk\"ohler number
itself.  At a given surface flux (rainfall intensity) $I$,
a pore class of radius $r$ is kinetic ($\Da>1$, carrying its own
dynamics) or quasi-static ($\Da<1$, slaved) according to whether $r$
lies below or above the crossover radius at which $\Da=1$.  The
direction is fixed by the relaxation time~\eqref{eq:tau},
$\tau_{\rm redis}(r)=m(r)/\mathcal{A}(r)$: classes far from the frontier,
where $m\to0$ and the exchange conductances vanish, equilibrate quickly,
while classes \emph{at} the frontier have the least conductance per unit
of stored perturbation and relax slowest.  The kinetic set is therefore a
band straddling $r^*$ rather than a half-line of small pores, and it is
present at \emph{any} intensity, widening with $I$.  Setting
$\Da(r)=\tau_{\rm redis}(r)/\tau_{\rm forcing}=1$, with the single-class
relaxation time $\tau_{\rm redis}(r)=m(r)/\mathcal{A}(r)$~\eqref{eq:tau} and
the REV-filling time $\tau_{\rm forcing}=\bar L\phi/I$ at surface flux
$I$, and using $\kappa_{\rm HP}(r)=r^2/(8\mu\bar L^2\bar\tau^2)$ with the
capillary driving $\gamma\cos\theta_c/r$, solving for $r$ gives
\begin{equation}
r_\Da(I)=\frac{4\mu\bar L\,\bar\tau^2}{\gamma\cos\theta_c\,\phi}\,I,
\label{eq:rDa}
\end{equation}
with $\mu$ the viscosity, $\bar L$ the REV size, $\bar\tau$ the
tortuosity, $\gamma\cos\theta_c$ the capillary constant and $\phi$ the
porosity (companion paper~\cite{Rigon2026PRE}, SM-A \S VII.E).  The
kinetic set is the band $\{r:\Da(r)\gtrsim1\}$ around $r^*$, of half-width
set by $r_\Da(I)$, and the quasi-static set is its complement.  Two
cautions.  First, $r_\Da$ is a \emph{width} around the frontier, not a
threshold radius separating small from large pores: with the relaxation
time~\eqref{eq:tau} the slowest classes are those at $r^*$, so the
kinetic set never empties, it only narrows as $I$ falls.  Second, on
names: the classes that lag are those near the frontier, whereas the
preferential \emph{flow} that results is carried by the large classes,
which stay wet and conduct while the frontier has not equilibrated.  We
therefore keep ``kinetic'' and ``quasi-static'' for the two sets and do
not call either of them ``preferential'', which in the hydrological
literature denotes the macropore pathways themselves.  This partition is
\emph{dynamic}: it sweeps continuously with
intensity $I$, in contrast to the fixed matrix/macropore boundary
posited by phenomenological dual-porosity models.  A single-mode soil
($r_\Da$ below the whole PSD at the intensities of interest) returns the
single Richards equation of Sec.~\ref{sec:Richards}; a soil whose PSD
carries a second population within the kinetic band requires the
multi-band treatment that follows, and
the Weiler-type \cite{Weiler2003} activation of macropore flow is precisely the event
$r_\Da(I)=r_\times$, when the crossover sweeps past the spectral valley
$r_\times$ separating two pore populations.

\subsection{The relaxation spectrum of a bimodal medium}

For a bimodal PSD $f=w_b f_b+w_t f_t$ (storage \emph{bodies} at large
$r_b$, controlling \emph{throats} at small $r_t$, separated by an order
of magnitude or more), the pore--throat incidence makes the
connectivity throat-mediated (two bodies never touch: every body--body
path passes through at least one throat, so the body--body kernel is the
throat-averaged product below; for media whose body/throat aspect ratio is
near one the two bands merge and the gap of Sec.~\ref{sec:gap} closes),
$C(r,r')=\int P(r,t)\,P(t,r')\,h(t)\,dt$, where $P(r,t)$ is the
probability that a pore of radius $r$ connects through a throat of
radius $t$ and $h(t)$ the throat-size density (companion paper's Supplemental Material~A, pore--throat
subsection), which
suppresses direct body--body coupling.  Consequently the relaxation
spectrum $\{\tau(r)=m(r)/\mathcal{A}(r)\}$ develops a \emph{gap}: a fast
body mode ($\tau_b\sim 1/\kappa(r_b)$, $\kappa\propto r^2$ large) and a
slow throat mode ($\tau_t\gg\tau_b$), with few classes in between.  The
single spectral gap of $\J$ is replaced by \emph{two} well-separated
scales,
\begin{equation}
\tau_b\ll\tau_t,
\end{equation}
and 
\begin{equation}
\varrho_\tau:=\tau_b/\tau_t\ll1 ,
\label{eq:gap}
\end{equation}
the \emph{band separation ratio} (not to be confused with the macroscopic
length $\Lambda$ of Sec.~\ref{sec:scale}).
Within each band $\J$ retains its own gap and the CE reduction applies;
\emph{between} bands the relaxation is slow and cannot be slaved.

\subsection{Band projection of the kinetic equation}

The construction that follows is a two-component reduction in the sense
of the invariant-manifold literature: the retained variables need
correspond to no conservation law of their own, and a band description
and a modal one are related by a change of amplitudes and describe the
same reduced dynamics, a two-zone/two-mode equivalence established for
shear dispersion by Roberts and Strunin~\cite{RobertsStrunin2004} and
formulated generally in~\cite{Roberts2025tma}.
Split the pore-class axis at the spectral valley $r_\times$ into a
body band $\mathsf{B}=\{r>r_\times\}$ and a throat band
$\mathsf{T}=\{r<r_\times\}$, and project the continuum kinetic
equation~\eqref{eq:contKE2} onto each (App.~\ref{app:bands}).  Writing
$\theta_b=\phi\int_\mathsf{B}g f\,dr$, $\theta_t=\phi\int_\mathsf{T}g
f\,dr$ for the band water contents, the redistribution operator splits
into intra-band parts $\C_{bb},\C_{tt}$ and a cross-band part
$\C_{bt}=-\C_{tb}$.  The intra-band operators are fast (gap
$|\lambda_1^{b}|,|\lambda_1^{t}|$); the cross-band operator is slow,
$O(\Lambda)$.  Applying the CE reduction \emph{within} each band slaves
$g|_\mathsf{B}\to g_\eq^b(\theta_b)$ and
$g|_\mathsf{T}\to g_\eq^t(\theta_t)$, each with its own frontier,
retention $\psi_b(\theta_b),\psi_t(\theta_t)$, and conductivity
$K_b,K_t$ from the band-restricted version of~\eqref{eq:K_mf}.  The
cross-band term cannot be slaved: its relaxation rate is smaller
than the intra-band rates by the band separation ratio
$\varrho_\tau\ll1$ of~\eqref{eq:gap}, so on the forcing time scale it
survives as an $O(1)$ exchange rather than relaxing away.  The
band-projection algebra that makes this explicit, and that produces the
exchange coefficient below, is carried out in App.~\ref{app:bands}.  The
$f$-weighted moment over each band gives the coupled pair
\begin{align}
\partial_t\theta_b&=\nabla\!\cdot\![K_b(\psi_b)(\nabla\psi_b+\hat z)]
  -\Gamma_w-S_b,
\label{eq:dpb}\\
\partial_t\theta_t&=\nabla\!\cdot\![K_t(\psi_t)(\nabla\psi_t+\hat z)]
  +\Gamma_w-S_t,
\label{eq:dpt}
\end{align}
where subscripts $b$ (body) and $t$ (throat) label the two spectral
bands, $\theta_b,\theta_t$ are their water contents, $\psi_b,\psi_t$ their
suctions, $K_b,K_t$ their conductivities, $S_b,S_t$ their sink terms, and
$\Gamma_w$ [s$^{-1}$] is the inter-band exchange rate: two Richards
equations coupled by $\Gamma_w$.
This is precisely the dual-permeability
structure~\cite{GerkeVanGenuchten1993}.

\subsection{The exchange coefficient}

The exchange is the $f$-weighted cross-band redistribution evaluated on
the two band equilibria.  Here $\C_{bt}$ denotes the part of the
redistribution operator that connects a class of the macropore band to a
class of the matrix band; like $\C$ itself it is a \emph{functional} of
the occupancy field, and the square brackets carry its two arguments, the
occupancy of the donor band and that of the receiver band, because a
cross-band transfer is gated by both: water can move from $\mathsf B$ to
$\mathsf T$ only if $\mathsf B$ holds water and $\mathsf T$ has room for
it.  Evaluating it on the two band equilibria $g_\eq^b,g_\eq^t$ and
integrating over the donor band,
\begin{equation}
\Gamma_w=\phi\!\int_\mathsf{B}\!\!\C_{bt}[g_\eq^b,g_\eq^t]\,f\,dr
  =\alpha_w\,(\psi_t-\psi_b),
\label{eq:exchange}
\end{equation}
a first-order (Gerke--van~Genuchten) exchange whose coefficient is
fixed, not fitted, by the inter-band connectivity and the bottleneck
conductance, as shown in Appendix \ref{app:cercignani}, up to an $O(1)$
profile factor: the band-indicator (Galerkin) estimate of the exchange rate
is a Rayleigh quotient and overshoots the exact slow eigenvalue by a factor
$1.15$--$1.3$ that does not vanish as the gap widens, the analogue for the
exchange term of the serial-path correction to $K$ (Supplemental Material,
Sec.~S3).
\begin{equation}
\alpha_w=\frac{\rho_w\mathrm{g}}{\bar L^2}\,
  \iint_{\mathsf{B}\times\mathsf{T}}
  \kappa_s(r,r')\,C(r,r')\,f(r)f(r')\,dr\,dr',
\label{eq:alpha}
\end{equation}
with $\kappa_s$ the symmetric (harmonic-mean) pair conductance of the
companion paper.  Equation~\eqref{eq:exchange} linearizes the cross-band
driving potential $\Phi(r,r')$ about the two frontiers; the linearization
is controlled by the same $\Da$ that controls the intra-band CE step.
The geometric shape factor of the classical first-order exchange
(the $\beta/a^2$ in Gerke--van~Genuchten) is here the inter-band
$C(r,r')$ integral~\eqref{eq:alpha}.  Note that the $C(r,r')$
in~\eqref{eq:alpha} is not the general connectivity kernel but its
restriction to body--throat pairs $\mathsf{B}\times\mathsf{T}$; a
measured or modelled $C$ would in principle fix the band partition
itself a priori, turning the choice of bands from an input into an
output of the connectivity structure.

\subsection{IN3M, mobile--immobile, and the $N$-band hierarchy}

Three limits of~\eqref{eq:dpb}--\eqref{eq:alpha} reproduce the standard
two-domain families:

When the body band carries
the infiltration front and the throat band is the matrix, the
exchange~\eqref{eq:exchange} is the macropore--matrix transfer that
governs the initiation of macropore flow; the activation threshold is
the intensity at which the body-band frontier first exceeds the
inter-band entry radius, i.e.\ the crossover
$r_\Da(I)$~\eqref{eq:rDa}, recovering the intensity-controlled onset
that IN3M imposes by rule.

If the throat band conducts negligibly,
$K_t\to0$, Eqs.~\eqref{eq:dpb}--\eqref{eq:dpt} reduce to a single mobile
Richards equation with a first-order mass sink into an immobile
store, the mobile--immobile model~\cite{vanGenuchtenWierenga1976}, with
$\alpha_w$ the rate coefficient.

Partitioning a
multimodal PSD into $N$ bands and projecting gives $N$ coupled Richards
equations with a $\binom{N}{2}$ matrix of exchange coefficients
$\alpha_{ij}$ from~\eqref{eq:alpha},the multiple-permeability
model~\cite{Simunek2003}.  The number of domains is not a modelling
choice but the number of gaps in the relaxation spectrum; a smooth
(unimodal) spectrum has one band and returns the single Richards
equation of Sec.~\ref{sec:Richards}.

Three remarks sharpen the status of this construction.
\emph{First}, bimodality is sufficient rather than necessary.  The
retention criterion is per class, $\Da(r)\gtrsim1$
[Eq.~\eqref{eq:rDa}],equivalently, per relaxation mode,
$\Da_n:=\Da\,\lambda_1/\lambda_n\gtrsim1$ in the eigenbasis of
Sec.~\ref{sec:eigenmodes},so under sufficiently intense forcing a
unimodal PSD supports the same two-domain reduction: there is then no
spectral valley $r_\times$ to place the partition, only the dynamic
crossover $r_\Da(I)$, and the validity of the resulting model is
controlled by $\Da_n$ rather than by a gap.  Bimodality is the case in
which the partition is spectrally protected and independent of the
forcing.  \emph{Second}, the band variables $(\theta_b,\theta_t)$ and
the modal variables $(\theta,U_1)$, with $U_1$ the amplitude of the
slowest nontrivial eigenmode $\varphi_1$ of $\J$, are related by an
invertible change of amplitudes and parametrize the same reduced
dynamics: in the two-band limit $\varphi_1$ approaches the balanced
band indicator, so $U_1$ measures the inter-band mass imbalance.  The
zonal and modal descriptions of a two-component reduction are one
object in two coordinate systems, as established for shear dispersion
by Roberts and Strunin~\cite{RobertsStrunin2004}.  \emph{Third}, the
hierarchy is transitive, and quantitatively so: adiabatically
eliminating the exchange mode ($\partial_t U_1\to0$) in the two-band
system returns the single Richards equation with the \emph{full}
conductivity of Eq.~\eqref{eq:K_spectral}, because the retained mode
carries exactly the $n=1$ term of that spectral sum, eliminating it
restores the term to $K$.  This mode-by-mode sum rule is an exact
internal-consistency identity for any implementation, and is verified
against the diagonalization of the Supplemental
Material~\cite{SM_hilbert} (S3).

\section{When the expansion breaks down}
\label{sec:breakdown}

Four mechanisms defeat the reduction, one per structural assumption, the
first three controlled by $\Da$ and the fourth by the spatial limit that
precedes it~\cite[\S V.4]{Cercignani1988}:

\textbf{(a) Kinetic forcing ($\Da\sim1$).}  Redistribution is no longer
fast; $g$ departs from the step, $r^*$ and $\psi$ cease to be state
variables, sharp fronts appear, the analogue of the Burnett and
super-Burnett regimes~\cite{Burnett1936} where the CE series ceases to
be asymptotic.  Before the full kinetic description becomes
unavoidable, the finite multi-band hierarchy of
Sec.~\ref{sec:multiperm} mediates: retaining the modes with
$\Da_n\gtrsim1$ as explicit variables extends the reduced description
mode by mode, rather than order by order in $\Da$,the alternative to
the Burnett route.

\textbf{(b) Sharp wetting fronts ($\Da\sim1$).}  The gradient scale
$\Lambda\sim\bar L$ and the expansion diverges, the analogue of shock
layers.  In an inertia-free system this cannot occur independently of
$\Da\gtrsim1$: redistribution self-heals any imposed sharp profile on
$\tau_{\rm redis}\ll\tau_{\rm forcing}$ whenever $\Da\ll1$
(Sec.~\ref{sec:scale}).

\textbf{(c) Percolation threshold ($\theta\to\theta_c^+$).}  The
spectral gap closes ($\lambda_1\to0$), $\tau_{\rm redis}\to\infty$,
$\Da\to\infty$ automatically, and $K\to0$ as a topological transition.

\textbf{(d) Loss of scale separation ($\xi\gtrsim\ell$).}  If the
connectivity length approaches the REV size, the \emph{spatial} limit of
Sec.~\ref{sec:continuum} fails and Eq.~\eqref{eq:contKE2} is not
available to expand, a failure upstream of the CE limit, independent of
$\Da$.

In cases (c) and (d) the reduction has no finite
refinement, $g(r)$, not $\theta$ or $\psi$, is the irreducible state
variable and the full kinetic equation must be carried.  In case (a)
the failure is graded: the multi-band ladder of
Sec.~\ref{sec:multiperm} extends the reduced description until the
retained spectrum itself becomes dense, and the full CKE is the
endpoint of that ladder rather than its only alternative.

\section{Geometric interpretation}
\label{sec:geometry}

The kinetic equation has the fast--slow form $\partial_t g=
(1/\Da)\,\C[g]+\mathcal{T}[g]$.  The equilibrium manifold
$\mathcal{S}_\eq$ is normally hyperbolic: motion along it is slow (the
null mode), motion toward it exponentially fast (the negative spectrum
of $\J$).  By Fenichel theory~\cite{Fenichel1979,Jones1995} a nearby
slow manifold $\mathcal{S}_\eq^{\Da}$ persists under the $O(\Da)$
perturbation, and the CE series is its representation
$g=g_\eq+\Da\,g^{(1)}+O(\Da^2)$; Richards' equation is the flow on
$\mathcal{S}_\eq^{\Da}$~\cite{Gorban2005}.

The invariant-manifold theory of dissipative multiscale
systems~\cite{Roberts2015ima,BunderRoberts2021,HochsRoberts2019,%
Roberts2025tma},developed under the name of homogenisation, though
it assumes no averaging, a point its authors emphasize, upgrades this
picture to a validity statement stronger than the asymptotic one.
Because $\J$ is self-adjoint and negative semidefinite with discrete
spectrum, one-dimensional kernel, and gap $\lambda_1(\theta)>0$, the
theorems guarantee: (i)~existence of the slow manifold
$\mathcal{S}_\eq^{\Da}$ in a \emph{finite} neighbourhood of the
equilibrium family and at \emph{finite} Damk\"ohler number, no limit
$\Da\to0$ is required, $\Da$ entering only through the size of the
neighbourhood and of the errors; (ii)~its emergence: all nearby
solutions are attracted at rate $\lambda_1/\tau_{\rm redis}$, which
converts the local-equilibrium ansatz from assumption into
consequence, and prescribes the initial datum of the reduced model as
the projection of $g(\cdot,\mathbf{x},0)$ along the stable fibration;
(iii)~an error of the same order as the residual of any truncation,
evaluable locally in space-time; and (iv), in the backward form
of~\cite{HochsRoberts2019}, exactness of the truncated model for a
kinetic system correspondingly close to the CKE.  The CE series above
is the Taylor approximation of this manifold; the multi-band hierarchy
of Sec.~\ref{sec:multiperm} consists of the wider invariant manifolds
obtained by retaining the slow modes explicitly, tied to the one-mode
$K(\theta)$ by the sum rule of Eq.~\eqref{eq:K_spectral}; and the
provable order of nonlinearity of an $M$-band model is bounded by the
gap ratio $\lambda_M/\lambda_{M-1}$~\cite{Roberts2025tma}, so where
eigenvalues accumulate near zero on approach to $\theta_c$, nonlinear
multi-band models fail before the linear gap-closure criterion is
reached.

Two leading mechanisms make the reduced description incomplete at finite
$\Da$.  The first is
\emph{hysteresis}.  The wetting and drying operators of the kinetic
theory do not commute, so a closed loop in the forcing returns the pore
configuration $g(r)$ to a different state than it started in even when
$\theta$ returns to its initial value; this path-dependence is what the
companion paper describes as a non-zero curvature (holonomy) in the
space of forcing histories~\cite{Rigon2026PRE}.  In the $\Da\to0$ limit
the fast redistribution erases the configurational memory between
forcing steps, the loop closes, and only a flattened remnant
survives, the branch dependence $K^{\rm wet}\neq K^{\rm dry}$ already
identified in Sec.~\ref{sec:Richards}.  The second is the
\emph{multiple-permeability exchange} of Sec.~\ref{sec:multiperm}: an
un-slaved redistribution channel between pore bands that the single-band
slaving discards.  Both are $O(\Da)$ corrections to the flow on
$\mathcal{S}_\eq^{\Da}$; Richards' equation is recovered precisely when
both vanish, i.e.\ when the configuration relaxes faster than the
forcing varies and the relaxation spectrum has a single gap.  This is
the geometric content of ``$\Da\ll1$'': the dynamics collapses onto a
one-dimensional manifold parametrized by $\theta$, on which the only
surviving transport datum is the scalar $K(\theta)$.

\section{Synthesis}
\label{sec:synthesis}

The reduction shows Richards' equation is the unique macroscopic
equation obtainable from the kinetic theory under $\Da\ll1$: not one
model among many but the hydrodynamic limit, in the precise sense in
which Navier--Stokes is the hydrodynamic limit of Boltzmann. 
The same expansion yields four further results, each a familiar object of
soil physics recovered as a structural consequence rather than a
constitutive input.

\emph{The hydraulic conductivity is a transport coefficient.}  $K(\theta)
=\phi\langle\kappa|\J^{-1}|S_0C_w\rangle_f$ is a first-order CE
coefficient, the exact analogue of viscosity, so it is a property of the
medium's relaxation spectrum, independent of the forcing.  Its mean-field
reduction returns the classical Mualem--Burdine integral, and the
serial-path resummation returns the heterogeneity penalty $e^{-4\sigma^2}$;
both are approximations \emph{of} the transport coefficient, not
independent closures.

\emph{Dual- and multiple-permeability models are derived, not posited.}
This is the result with the widest practical reach, and the reduction puts
it on a first-principles footing.  A bimodal pore-size distribution splits
the relaxation spectrum of $\J$ into two bands separated by a gap; applying
the CE reduction \emph{within} each band, while the
slow cross-band relaxation is retained at $O(1)$, produces two coupled
Richards equations, exactly the Gerke--van~Genuchten dual-permeability
system, Weiler's IN3M as its three-band case, and mobile--immobile
transport as the limit in which one band is conductively dead
(Sec.~\ref{sec:multiperm}).  Crucially, the inter-band exchange coefficient
$\Gamma_w$ is not a fitting parameter but is \emph{computed} from the
cross-band connectivity, $\Gamma_w$ (App.~\ref{app:bands}); the first-order shape assumed by
phenomenological dual-porosity theory is here the leading term of a
controlled expansion, with its coefficient fixed by pore geometry.  What
was a modelling choice becomes a theorem with a formula.

\emph{The reduction predicts its own breakdown.}  Two distinct routes leave
the Richards regime, and the theory locates both.  Raising the forcing
drives $\Da\to1$ (sharp fronts); lowering the water content drives the
spectral gap $\lambda_1\to0$ at the percolation threshold, where $\J^{-1}$
becomes unbounded and the closed conductivity ceases to exist.  The second
is the sharper statement, a macroscopic constitutive law dissolving at a
connectivity transition its own spectrum detects, and it is verified
numerically in the Supplemental Material.  A companion volume of worked
derivations, in which every bracket and every projection of this paper is
carried out step by step, is available from the author.

\emph{Why Richards' equation is universal, and what $\theta_c$ is.}
The multi-band hierarchy of Sec.~\ref{sec:multiperm} can be read as a
sequence of eliminations: starting from the full kinetic equation, one
removes the fastest relaxation modes one group at a time, each removal
depositing that group's contribution into the conductivity through the
spectral sum~\eqref{eq:K_spectral}.  Eliminating in two steps or in one
gives the same final model, so the eliminations compose, the property
that makes this a renormalization-group flow rather than a mere
sequence of approximations.  Two consequences follow.  Away from the
percolation threshold, where a gap is present, every mode is eventually
eliminated and the end point of the flow is always the same
\emph{form}, a single conservation law with a flux linear in
$\nabla\psi$; the soil survives in that end point only through the two
functions $\psi(\theta)$ and $K(\theta)$.  This is the structural
reason a century-old empirical equation works across soils of wildly
different microstructure: what is universal is the form, not the
coefficients.  Approaching $\theta_c$, by contrast, the gap closes and
the modes to be eliminated accumulate at zero relaxation rate, with a
spectral density of the anomalous, fractal
type~\cite{AlexanderOrbach1982}; the elimination no longer terminates,
and $K(\theta)$ should acquire a non-analytic dependence on
$\theta-\theta_c$ fixed by that density together with the collapse of
the overlaps $\langle\kappa,\varphi_1\rangle_f$
in~\eqref{eq:K_spectral}.  Computing that exponent is deferred to
future work; what the present framework establishes is that
$\theta_c$ is not where Richards' equation becomes inaccurate but where
the object it describes stops existing.

\emph{Hysteresis has an operator home.}  The non-commutativity
$[\Lg_\W,\Lg_\D]\neq0$ of the wetting and drying generators, which has no
counterpart in the gas ($[\Lg_\W,\Lg_\D]=0$), is the microscopic origin of
hysteresis; at the operator level it is the curvature
$[\mathsf{T},\mathsf{P}]$, where $\mathsf{T}:=\nabla\!\cdot\!\mathbf{F}$
is the transport (streaming) operator and $\mathsf{P}$ the projector onto
local equilibrium (App.~\ref{app:schur}), the same object that carries
dynamic capillarity.

Keeping the spatial and temporal limit apart is what makes the structure of the
reduction transparent.  It separates the spatial
contraction, a kinematic coarse-graining
that involves no averaging over the pore ensemble, which
assumes nothing about time scales, from the CE
expansion, which is the only place a time-scale separation is invoked;
it makes the role of the single small parameter $\Da$ explicit; and it
isolates exactly where each classical assumption, local equilibrium,
statistical homogeneity, slowly varying accessibility, enters.

At the operator level (App.~\ref{app:schur}) the reduction is the Schur
complement of the fast block of $\mathbb{L}=\Da^{-1}\J-\mathsf{T}$; the
projector onto local equilibrium depends on position through $\theta$, and
the curvature $[\mathsf{T},\mathsf{P}]$ of that dependence is the operator
home of the non-Fickian corrections, dynamic capillarity and the
hysteretic $[\Lg_\W,\Lg_\D]\neq0$ entry of Table~\ref{tab:synthesis},
vanishing only for a spatially uniform frontier.

The analogy we draw on throughout with Boltzmann equation treatment is \emph{structural rather than
literal}.  Unlike the Boltzmann case there is no momentum conservation
and no Galilean invariance; the kinetic variable is a pore-filling
fraction, not a molecular velocity; and the redistribution operator
encodes capillary exchange across pores, not binary molecular
collisions.  What carries over is the operator architecture, a single
conserved mode, a spectral gap, a local-equilibrium manifold, and a
Fredholm solvability condition that closes the macroscopic balance, and
it is this architecture, not the microscopic content, that delivers
Richards' equation as a hydrodynamic limit.

\begin{table*}
\caption{\label{tab:synthesis} The complete Boltzmann${}\to{}$Richards
analogy, with the new entries (bands, continuum/CE split) made explicit.}
\begin{ruledtabular}
\begin{tabular}{lll}
\textbf{Classical gas dynamics} & \textbf{Soil water (this work)} &
  \textbf{Shared structure} \\
\hline
Boltzmann equation & Continuum kinetic eq.~\eqref{eq:contKE2} &
  Kinetic eq.\ with fast relaxation \\
Spatial homogenisation (none needed) & REV $\to$ point,
  $\varepsilon\to0$ & Kinematic continuum limit \\
Maxwellian $M(\mathbf{v})$ & Packing step $g_\eq=H(r^*\!-r)$ &
  Equilibrium distribution \\
$\{1,\mathbf{v},|\mathbf{v}|^2\}$: 5 invariants &
  $\{1\}$: 1 invariant ($\theta$) & Hydrodynamic projection \\
Linearized $L$, self-adjoint $\le0$ & Linearized $\J$, self-adjoint
  $\le0$ & Spectral gap, slaving \\
Euler $\to$ Navier--Stokes & Mass conservation $\to$ Richards &
  Solvability $\to$ macroscopic eq. \\
Viscosity $\mu=(\mathbf{c},L^{-1}\mathbf{c})$ &
  $K=\phi\langle\kappa,\J^{-1}S\rangle_f$ & 1st-order transport coeff. \\
Bhatnagar variational $\mu$ & Variational $K$
  & 2nd-order-accurate estimate \\
Boltzmann $H$-theorem & $d\F/dt\le0$ & Lyapunov / entropy \\
Maxwellian manifold & $\mathcal{S}_\eq$ & Slow invariant manifold \\
Shock layers ($\Kn\sim1$) & Sharp fronts ($\Da\sim1$) & CE breakdown \\
Burnett ($\Kn\sim1$) & Kinetic forcing ($\Da\sim1$) &
  Full kinetic eq.\ needed \\
-- & Percolation threshold & Spectral gap closes \\
Multi-scale relaxation & Bimodal spectrum gap $\Lambda$ &
  Dual/multiple permeability \\
$[\Lg_\W,\Lg_\D]=0$ & $[\Lg_\W,\Lg_\D]\neq0$ &
  Forcing curvature (hysteresis) \\
\end{tabular}
\end{ruledtabular}
\end{table*}

\bibliographystyle{apsrev4-2}
\bibliography{Richards_as_limit}

\begin{thebibliography}{48}%
\makeatletter
\providecommand \@ifxundefined [1]{%
 \@ifx{#1\undefined}
}%
\providecommand \@ifnum [1]{%
 \ifnum #1\expandafter \@firstoftwo
 \else \expandafter \@secondoftwo
 \fi
}%
\providecommand \@ifx [1]{%
 \ifx #1\expandafter \@firstoftwo
 \else \expandafter \@secondoftwo
 \fi
}%
\providecommand \natexlab [1]{#1}%
\providecommand \enquote  [1]{``#1''}%
\providecommand \bibnamefont  [1]{#1}%
\providecommand \bibfnamefont [1]{#1}%
\providecommand \citenamefont [1]{#1}%
\providecommand \href@noop [0]{\@secondoftwo}%
\providecommand \href [0]{\begingroup \@sanitize@url \@href}%
\providecommand \@href[1]{\@@startlink{#1}\@@href}%
\providecommand \@@href[1]{\endgroup#1\@@endlink}%
\providecommand \@sanitize@url [0]{\catcode `\\12\catcode `\$12\catcode
  `\&12\catcode `\#12\catcode `\^12\catcode `\_12\catcode `\%12\relax}%
\providecommand \@@startlink[1]{}%
\providecommand \@@endlink[0]{}%
\providecommand \url  [0]{\begingroup\@sanitize@url \@url }%
\providecommand \@url [1]{\endgroup\@href {#1}{\urlprefix }}%
\providecommand \urlprefix  [0]{URL }%
\providecommand \Eprint [0]{\href }%
\providecommand \doibase [0]{https://doi.org/}%
\providecommand \selectlanguage [0]{\@gobble}%
\providecommand \bibinfo  [0]{\@secondoftwo}%
\providecommand \bibfield  [0]{\@secondoftwo}%
\providecommand \translation [1]{[#1]}%
\providecommand \BibitemOpen [0]{}%
\providecommand \bibitemStop [0]{}%
\providecommand \bibitemNoStop [0]{.\EOS\space}%
\providecommand \EOS [0]{\spacefactor3000\relax}%
\providecommand \BibitemShut  [1]{\csname bibitem#1\endcsname}%
\let\auto@bib@innerbib\@empty
\bibitem [{\citenamefont {Rigon}(2026)}]{Rigon2026PRE}%
  \BibitemOpen
  \bibfield  {author} {\bibinfo {author} {\bibfnamefont {R.}~\bibnamefont
  {Rigon}},\ }\href@noop {} {\bibfield  {journal} {\bibinfo  {journal} {arXiv
  preprint}\ } (\bibinfo {year} {2026})},\ \bibinfo {note} {companion paper
  (PRE-1)},\ \Eprint {https://arxiv.org/abs/2607.09416} {arXiv:2607.09416
  [cond-mat.stat-mech]} \BibitemShut {NoStop}%
\bibitem [{\citenamefont {Richards}(1931)}]{Richards1931}%
  \BibitemOpen
  \bibfield  {author} {\bibinfo {author} {\bibfnamefont {L.~A.}\ \bibnamefont
  {Richards}},\ }\href@noop {} {\bibfield  {journal} {\bibinfo  {journal}
  {Physics}\ }\textbf {\bibinfo {volume} {1}},\ \bibinfo {pages} {318}
  (\bibinfo {year} {1931})}\BibitemShut {NoStop}%
\bibitem [{\citenamefont {Tubini}\ and\ \citenamefont
  {Rigon}(2022)}]{TubiniRigon2022}%
  \BibitemOpen
  \bibfield  {author} {\bibinfo {author} {\bibfnamefont {N.}~\bibnamefont
  {Tubini}}\ and\ \bibinfo {author} {\bibfnamefont {R.}~\bibnamefont {Rigon}},\
  }\href@noop {} {\bibfield  {journal} {\bibinfo  {journal} {Geoscientific
  Model Development}\ }\textbf {\bibinfo {volume} {15}},\ \bibinfo {pages} {75}
  (\bibinfo {year} {2022})}\BibitemShut {NoStop}%
\bibitem [{\citenamefont {van Genuchten}(1980)}]{vanGenuchten1980}%
  \BibitemOpen
  \bibfield  {author} {\bibinfo {author} {\bibfnamefont {M.~T.}\ \bibnamefont
  {van Genuchten}},\ }\href@noop {} {\bibfield  {journal} {\bibinfo  {journal}
  {Soil Sci. Soc. Am. J.}\ }\textbf {\bibinfo {volume} {44}},\ \bibinfo {pages}
  {892} (\bibinfo {year} {1980})}\BibitemShut {NoStop}%
\bibitem [{\citenamefont {Brooks}\ and\ \citenamefont
  {Corey}(1964)}]{Brooks1964}%
  \BibitemOpen
  \bibfield  {author} {\bibinfo {author} {\bibfnamefont {R.~H.}\ \bibnamefont
  {Brooks}}\ and\ \bibinfo {author} {\bibfnamefont {A.~T.}\ \bibnamefont
  {Corey}},\ }\href@noop {} {\emph {\bibinfo {title} {Hydraulic properties of
  porous media}}},\ \bibinfo {type} {Hydrology Paper}\ \bibinfo {number} {3}\
  (\bibinfo  {institution} {Colorado State University},\ \bibinfo {address}
  {Fort Collins},\ \bibinfo {year} {1964})\BibitemShut {NoStop}%
\bibitem [{\citenamefont {Mualem}(1976)}]{Mualem1976}%
  \BibitemOpen
  \bibfield  {author} {\bibinfo {author} {\bibfnamefont {Y.}~\bibnamefont
  {Mualem}},\ }\href@noop {} {\bibfield  {journal} {\bibinfo  {journal} {Water
  Resour. Res.}\ }\textbf {\bibinfo {volume} {12}},\ \bibinfo {pages} {513}
  (\bibinfo {year} {1976})}\BibitemShut {NoStop}%
\bibitem [{\citenamefont {Darcy}(1856)}]{Darcy1856}%
  \BibitemOpen
  \bibfield  {author} {\bibinfo {author} {\bibfnamefont {H.}~\bibnamefont
  {Darcy}},\ }\href@noop {} {\emph {\bibinfo {title} {Les Fontaines Publiques
  de la Ville de Dijon}}}\ (\bibinfo  {publisher} {Victor Dalmont},\ \bibinfo
  {address} {Paris},\ \bibinfo {year} {1856})\BibitemShut {NoStop}%
\bibitem [{\citenamefont {Buckingham}(1907)}]{Buckingham1907}%
  \BibitemOpen
  \bibfield  {author} {\bibinfo {author} {\bibfnamefont {E.}~\bibnamefont
  {Buckingham}},\ }\href@noop {} {\emph {\bibinfo {title} {Studies on the
  movement of soil moisture}}},\ \bibinfo {type} {Bulletin}\ \bibinfo {number}
  {38}\ (\bibinfo  {institution} {U.S. Department of Agriculture, Bureau of
  Soils},\ \bibinfo {address} {Washington, DC},\ \bibinfo {year}
  {1907})\BibitemShut {NoStop}%
\bibitem [{\citenamefont {Bear}(1972)}]{Bear1972}%
  \BibitemOpen
  \bibfield  {author} {\bibinfo {author} {\bibfnamefont {J.}~\bibnamefont
  {Bear}},\ }\href@noop {} {\emph {\bibinfo {title} {Dynamics of Fluids in
  Porous Media}}}\ (\bibinfo  {publisher} {Elsevier},\ \bibinfo {address} {New
  York},\ \bibinfo {year} {1972})\BibitemShut {NoStop}%
\bibitem [{\citenamefont {Beven}\ and\ \citenamefont
  {Germann}(1982)}]{Beven1982}%
  \BibitemOpen
  \bibfield  {author} {\bibinfo {author} {\bibfnamefont {K.}~\bibnamefont
  {Beven}}\ and\ \bibinfo {author} {\bibfnamefont {P.}~\bibnamefont
  {Germann}},\ }\href@noop {} {\bibfield  {journal} {\bibinfo  {journal} {Water
  Resour. Res.}\ }\textbf {\bibinfo {volume} {18}},\ \bibinfo {pages} {1311}
  (\bibinfo {year} {1982})}\BibitemShut {NoStop}%
\bibitem [{\citenamefont {Germann}(2018)}]{Germann2018}%
  \BibitemOpen
  \bibfield  {author} {\bibinfo {author} {\bibfnamefont {P.~F.}\ \bibnamefont
  {Germann}},\ }\href@noop {} {\emph {\bibinfo {title} {Preferential Flow:
  Stokes Approach to Infiltration and Drainage}}}\ (\bibinfo  {publisher}
  {Geographica Bernensia},\ \bibinfo {address} {Bern},\ \bibinfo {year}
  {2018})\BibitemShut {NoStop}%
\bibitem [{\citenamefont {Diamantopoulos}\ and\ \citenamefont
  {Durner}(2012)}]{Diamantopoulos2012}%
  \BibitemOpen
  \bibfield  {author} {\bibinfo {author} {\bibfnamefont {E.}~\bibnamefont
  {Diamantopoulos}}\ and\ \bibinfo {author} {\bibfnamefont {W.}~\bibnamefont
  {Durner}},\ }\bibfield  {journal} {\bibinfo  {journal} {Vadose Zone J.}\
  }\textbf {\bibinfo {volume} {11}},\ \href
  {https://doi.org/10.2136/vzj2011.0197} {10.2136/vzj2011.0197} (\bibinfo
  {year} {2012})\BibitemShut {NoStop}%
\bibitem [{\citenamefont {Hilfer}(2006)}]{Hilfer2006}%
  \BibitemOpen
  \bibfield  {author} {\bibinfo {author} {\bibfnamefont {R.}~\bibnamefont
  {Hilfer}},\ }\href@noop {} {\bibfield  {journal} {\bibinfo  {journal} {Phys.
  Rev. E}\ }\textbf {\bibinfo {volume} {73}},\ \bibinfo {pages} {016307}
  (\bibinfo {year} {2006})}\BibitemShut {NoStop}%
\bibitem [{\citenamefont {Hassanizadeh}\ \emph {et~al.}(2002)\citenamefont
  {Hassanizadeh}, \citenamefont {Celia},\ and\ \citenamefont
  {Dahle}}]{Hassanizadeh2002}%
  \BibitemOpen
  \bibfield  {author} {\bibinfo {author} {\bibfnamefont {S.~M.}\ \bibnamefont
  {Hassanizadeh}}, \bibinfo {author} {\bibfnamefont {M.~A.}\ \bibnamefont
  {Celia}},\ and\ \bibinfo {author} {\bibfnamefont {H.~K.}\ \bibnamefont
  {Dahle}},\ }\href@noop {} {\bibfield  {journal} {\bibinfo  {journal} {Vadose
  Zone J.}\ }\textbf {\bibinfo {volume} {1}},\ \bibinfo {pages} {38} (\bibinfo
  {year} {2002})}\BibitemShut {NoStop}%
\bibitem [{\citenamefont {Chapman}\ and\ \citenamefont
  {Cowling}(1970)}]{ChapmanCowling1970}%
  \BibitemOpen
  \bibfield  {author} {\bibinfo {author} {\bibfnamefont {S.}~\bibnamefont
  {Chapman}}\ and\ \bibinfo {author} {\bibfnamefont {T.~G.}\ \bibnamefont
  {Cowling}},\ }\href@noop {} {\emph {\bibinfo {title} {The Mathematical Theory
  of Non-Uniform Gases}}},\ \bibinfo {edition} {3rd}\ ed.\ (\bibinfo
  {publisher} {Cambridge University Press},\ \bibinfo {address} {Cambridge},\
  \bibinfo {year} {1970})\BibitemShut {NoStop}%
\bibitem [{\citenamefont {Hilbert}(1912)}]{Hilbert1912}%
  \BibitemOpen
  \bibfield  {author} {\bibinfo {author} {\bibfnamefont {D.}~\bibnamefont
  {Hilbert}},\ }\href@noop {} {\bibfield  {journal} {\bibinfo  {journal} {Math.
  Ann.}\ }\textbf {\bibinfo {volume} {72}},\ \bibinfo {pages} {562} (\bibinfo
  {year} {1912})}\BibitemShut {NoStop}%
\bibitem [{\citenamefont {Enskog}(1917)}]{Enskog1917}%
  \BibitemOpen
  \bibfield  {author} {\bibinfo {author} {\bibfnamefont {D.}~\bibnamefont
  {Enskog}},\ }\href@noop {} {\emph {\bibinfo {title} {Kinetische {Theorie} der
  {Vorg\"ange} in {m\"a\ss ig} verd{\"u}nnten {Gasen}}}}\ (\bibinfo
  {publisher} {Almqvist \& Wiksell},\ \bibinfo {address} {Uppsala},\ \bibinfo
  {year} {1917})\BibitemShut {NoStop}%
\bibitem [{\citenamefont {Cercignani}(1988)}]{Cercignani1988}%
  \BibitemOpen
  \bibfield  {author} {\bibinfo {author} {\bibfnamefont {C.}~\bibnamefont
  {Cercignani}},\ }\href@noop {} {\emph {\bibinfo {title} {The {Boltzmann}
  Equation and Its Applications}}}\ (\bibinfo  {publisher} {Springer},\
  \bibinfo {address} {New York},\ \bibinfo {year} {1988})\BibitemShut {NoStop}%
\bibitem [{SM_(2026)}]{SM_hilbert}%
  \BibitemOpen
  \href@noop {} {} (\bibinfo {year} {2026}),\ \bibinfo {note} {see Supplemental
  Material for the numerical tests on which the reduction rests: S1,
  diagonalization of the pore-class operator (radius labeling of the modes,
  spectral sum rule, gap versus band width); S2, joint closure of the spectral
  gap and of the Stokes permeability at the percolation threshold; S3, the
  two-band exchange coefficient of a bimodal medium.}\BibitemShut {Stop}%
\bibitem [{\citenamefont {Hagen}(1839)}]{Hagen1839}%
  \BibitemOpen
  \bibfield  {author} {\bibinfo {author} {\bibfnamefont {G.}~\bibnamefont
  {Hagen}},\ }\href@noop {} {\bibfield  {journal} {\bibinfo  {journal} {Ann.
  Phys.}\ }\textbf {\bibinfo {volume} {122}},\ \bibinfo {pages} {423} (\bibinfo
  {year} {1839})}\BibitemShut {NoStop}%
\bibitem [{\citenamefont {Poiseuille}(1840)}]{Poiseuille1840}%
  \BibitemOpen
  \bibfield  {author} {\bibinfo {author} {\bibfnamefont {J.~L.~M.}\
  \bibnamefont {Poiseuille}},\ }\href@noop {} {\bibfield  {journal} {\bibinfo
  {journal} {C. R. Acad. Sci.}\ }\textbf {\bibinfo {volume} {11}},\ \bibinfo
  {pages} {961} (\bibinfo {year} {1840})}\BibitemShut {NoStop}%
\bibitem [{\citenamefont {Lucas}(1918)}]{Lucas1918}%
  \BibitemOpen
  \bibfield  {author} {\bibinfo {author} {\bibfnamefont {R.}~\bibnamefont
  {Lucas}},\ }\href {https://doi.org/10.1007/BF01461107} {\bibfield  {journal}
  {\bibinfo  {journal} {Kolloid-Z.}\ }\textbf {\bibinfo {volume} {23}},\
  \bibinfo {pages} {15} (\bibinfo {year} {1918})}\BibitemShut {NoStop}%
\bibitem [{\citenamefont {Washburn}(1921)}]{Washburn1921}%
  \BibitemOpen
  \bibfield  {author} {\bibinfo {author} {\bibfnamefont {E.~W.}\ \bibnamefont
  {Washburn}},\ }\href {https://doi.org/10.1103/PhysRev.17.273} {\bibfield
  {journal} {\bibinfo  {journal} {Phys. Rev.}\ }\textbf {\bibinfo {volume}
  {17}},\ \bibinfo {pages} {273} (\bibinfo {year} {1921})}\BibitemShut
  {NoStop}%
\bibitem [{\citenamefont {Roberts}(2015)}]{Roberts2015ima}%
  \BibitemOpen
  \bibfield  {author} {\bibinfo {author} {\bibfnamefont {A.~J.}\ \bibnamefont
  {Roberts}},\ }\href {https://doi.org/10.1093/imamat/hxv004} {\bibfield
  {journal} {\bibinfo  {journal} {IMA Journal of Applied Mathematics}\ }\textbf
  {\bibinfo {volume} {80}},\ \bibinfo {pages} {1492} (\bibinfo {year}
  {2015})}\BibitemShut {NoStop}%
\bibitem [{\citenamefont {Bunder}\ and\ \citenamefont
  {Roberts}(2021)}]{BunderRoberts2021}%
  \BibitemOpen
  \bibfield  {author} {\bibinfo {author} {\bibfnamefont {J.~E.}\ \bibnamefont
  {Bunder}}\ and\ \bibinfo {author} {\bibfnamefont {A.~J.}\ \bibnamefont
  {Roberts}},\ }\href {https://doi.org/10.1007/s42452-021-04229-9} {\bibfield
  {journal} {\bibinfo  {journal} {SN Applied Sciences}\ }\textbf {\bibinfo
  {volume} {3}},\ \bibinfo {pages} {1} (\bibinfo {year} {2021})}\BibitemShut
  {NoStop}%
\bibitem [{\citenamefont {Roberts}(2025)}]{Roberts2025tma}%
  \BibitemOpen
  \bibfield  {author} {\bibinfo {author} {\bibfnamefont {A.~J.}\ \bibnamefont
  {Roberts}},\ }\href {https://doi.org/10.1093/imatrm/tnaf001} {\bibfield
  {journal} {\bibinfo  {journal} {Transactions of Mathematics and Its
  Applications}\ }\textbf {\bibinfo {volume} {9}},\ \bibinfo {pages} {tnaf001}
  (\bibinfo {year} {2025})}\BibitemShut {NoStop}%
\bibitem [{\citenamefont {Fiedler}(1973)}]{Fiedler1973}%
  \BibitemOpen
  \bibfield  {author} {\bibinfo {author} {\bibfnamefont {M.}~\bibnamefont
  {Fiedler}},\ }\href@noop {} {\bibfield  {journal} {\bibinfo  {journal}
  {Czech. Math. J.}\ }\textbf {\bibinfo {volume} {23}},\ \bibinfo {pages} {298}
  (\bibinfo {year} {1973})}\BibitemShut {NoStop}%
\bibitem [{\citenamefont {Fredholm}(1903)}]{Fredholm1903}%
  \BibitemOpen
  \bibfield  {author} {\bibinfo {author} {\bibfnamefont {I.}~\bibnamefont
  {Fredholm}},\ }\href@noop {} {\bibfield  {journal} {\bibinfo  {journal} {Acta
  Math.}\ }\textbf {\bibinfo {volume} {27}},\ \bibinfo {pages} {365} (\bibinfo
  {year} {1903})}\BibitemShut {NoStop}%
\bibitem [{\citenamefont {Kubo}(1957)}]{Kubo1957}%
  \BibitemOpen
  \bibfield  {author} {\bibinfo {author} {\bibfnamefont {R.}~\bibnamefont
  {Kubo}},\ }\href {https://doi.org/10.1143/JPSJ.12.570} {\bibfield  {journal}
  {\bibinfo  {journal} {J. Phys. Soc. Jpn.}\ }\textbf {\bibinfo {volume}
  {12}},\ \bibinfo {pages} {570} (\bibinfo {year} {1957})}\BibitemShut
  {NoStop}%
\bibitem [{\citenamefont {Hunt}(2004)}]{Hunt2004}%
  \BibitemOpen
  \bibfield  {author} {\bibinfo {author} {\bibfnamefont {A.~G.}\ \bibnamefont
  {Hunt}},\ }\href@noop {} {\bibfield  {journal} {\bibinfo  {journal} {Geophys.
  Res. Lett.}\ }\textbf {\bibinfo {volume} {31}},\ \bibinfo {pages} {L21503}
  (\bibinfo {year} {2004})}\BibitemShut {NoStop}%
\bibitem [{\citenamefont {Burdine}(1953)}]{Burdine1953}%
  \BibitemOpen
  \bibfield  {author} {\bibinfo {author} {\bibfnamefont {N.~T.}\ \bibnamefont
  {Burdine}},\ }\href@noop {} {\bibfield  {journal} {\bibinfo  {journal} {J.
  Pet. Technol.}\ }\textbf {\bibinfo {volume} {5}},\ \bibinfo {pages} {71}
  (\bibinfo {year} {1953})}\BibitemShut {NoStop}%
\bibitem [{\citenamefont {Hirschfelder}\ \emph {et~al.}(1954)\citenamefont
  {Hirschfelder}, \citenamefont {Curtiss},\ and\ \citenamefont
  {Bird}}]{Hirschfelder1954}%
  \BibitemOpen
  \bibfield  {author} {\bibinfo {author} {\bibfnamefont {J.~O.}\ \bibnamefont
  {Hirschfelder}}, \bibinfo {author} {\bibfnamefont {C.~F.}\ \bibnamefont
  {Curtiss}},\ and\ \bibinfo {author} {\bibfnamefont {R.~B.}\ \bibnamefont
  {Bird}},\ }\href@noop {} {\emph {\bibinfo {title} {Molecular Theory of Gases
  and Liquids}}}\ (\bibinfo  {publisher} {John Wiley \& Sons},\ \bibinfo
  {address} {New York},\ \bibinfo {year} {1954})\BibitemShut {NoStop}%
\bibitem [{\citenamefont {Haines}(1930)}]{Haines1930}%
  \BibitemOpen
  \bibfield  {author} {\bibinfo {author} {\bibfnamefont {W.~B.}\ \bibnamefont
  {Haines}},\ }\href@noop {} {\bibfield  {journal} {\bibinfo  {journal} {J.
  Agric. Sci.}\ }\textbf {\bibinfo {volume} {20}},\ \bibinfo {pages} {97}
  (\bibinfo {year} {1930})}\BibitemShut {NoStop}%
\bibitem [{\citenamefont {Rawls}\ \emph {et~al.}(1982)\citenamefont {Rawls},
  \citenamefont {Brakensiek},\ and\ \citenamefont {Saxton}}]{Rawls1982}%
  \BibitemOpen
  \bibfield  {author} {\bibinfo {author} {\bibfnamefont {W.~J.}\ \bibnamefont
  {Rawls}}, \bibinfo {author} {\bibfnamefont {D.~L.}\ \bibnamefont
  {Brakensiek}},\ and\ \bibinfo {author} {\bibfnamefont {K.~E.}\ \bibnamefont
  {Saxton}},\ }\href@noop {} {\bibfield  {journal} {\bibinfo  {journal} {Trans.
  ASAE}\ }\textbf {\bibinfo {volume} {25}},\ \bibinfo {pages} {1316} (\bibinfo
  {year} {1982})}\BibitemShut {NoStop}%
\bibitem [{\citenamefont {Kosugi}(1996)}]{Kosugi1996}%
  \BibitemOpen
  \bibfield  {author} {\bibinfo {author} {\bibfnamefont {K.}~\bibnamefont
  {Kosugi}},\ }\href@noop {} {\bibfield  {journal} {\bibinfo  {journal} {Water
  Resour. Res.}\ }\textbf {\bibinfo {volume} {32}},\ \bibinfo {pages} {2697}
  (\bibinfo {year} {1996})}\BibitemShut {NoStop}%
\bibitem [{\citenamefont {Tuller}\ \emph {et~al.}(1999)\citenamefont {Tuller},
  \citenamefont {Or},\ and\ \citenamefont {Dudley}}]{Tuller1999}%
  \BibitemOpen
  \bibfield  {author} {\bibinfo {author} {\bibfnamefont {M.}~\bibnamefont
  {Tuller}}, \bibinfo {author} {\bibfnamefont {D.}~\bibnamefont {Or}},\ and\
  \bibinfo {author} {\bibfnamefont {L.~M.}\ \bibnamefont {Dudley}},\
  }\href@noop {} {\bibfield  {journal} {\bibinfo  {journal} {Water Resour.
  Res.}\ }\textbf {\bibinfo {volume} {35}},\ \bibinfo {pages} {1949} (\bibinfo
  {year} {1999})}\BibitemShut {NoStop}%
\bibitem [{\citenamefont {Or}\ and\ \citenamefont
  {Tuller}(2000)}]{OrTuller2000}%
  \BibitemOpen
  \bibfield  {author} {\bibinfo {author} {\bibfnamefont {D.}~\bibnamefont
  {Or}}\ and\ \bibinfo {author} {\bibfnamefont {M.}~\bibnamefont {Tuller}},\
  }\href@noop {} {\bibfield  {journal} {\bibinfo  {journal} {Water Resour.
  Res.}\ }\textbf {\bibinfo {volume} {36}},\ \bibinfo {pages} {1165} (\bibinfo
  {year} {2000})}\BibitemShut {NoStop}%
\bibitem [{\citenamefont {Gerke}\ and\ \citenamefont {van
  Genuchten}(1993)}]{GerkeVanGenuchten1993}%
  \BibitemOpen
  \bibfield  {author} {\bibinfo {author} {\bibfnamefont {H.~H.}\ \bibnamefont
  {Gerke}}\ and\ \bibinfo {author} {\bibfnamefont {M.~T.}\ \bibnamefont {van
  Genuchten}},\ }\href@noop {} {\bibfield  {journal} {\bibinfo  {journal}
  {Water Resour. Res.}\ }\textbf {\bibinfo {volume} {29}},\ \bibinfo {pages}
  {305} (\bibinfo {year} {1993})}\BibitemShut {NoStop}%
\bibitem [{\citenamefont {Weiler}\ and\ \citenamefont
  {Naef}(2003)}]{Weiler2003}%
  \BibitemOpen
  \bibfield  {author} {\bibinfo {author} {\bibfnamefont {M.}~\bibnamefont
  {Weiler}}\ and\ \bibinfo {author} {\bibfnamefont {F.}~\bibnamefont {Naef}},\
  }\href@noop {} {\bibfield  {journal} {\bibinfo  {journal} {J. Hydrol.}\
  }\textbf {\bibinfo {volume} {273}},\ \bibinfo {pages} {139} (\bibinfo {year}
  {2003})}\BibitemShut {NoStop}%
\bibitem [{\citenamefont {{\v{S}}im{\r{u}}nek}\ \emph
  {et~al.}(2003)\citenamefont {{\v{S}}im{\r{u}}nek}, \citenamefont {Jarvis},
  \citenamefont {van Genuchten},\ and\ \citenamefont
  {G{\"a}rden{\"a}s}}]{Simunek2003}%
  \BibitemOpen
  \bibfield  {author} {\bibinfo {author} {\bibfnamefont {J.}~\bibnamefont
  {{\v{S}}im{\r{u}}nek}}, \bibinfo {author} {\bibfnamefont {N.~J.}\
  \bibnamefont {Jarvis}}, \bibinfo {author} {\bibfnamefont {M.~T.}\
  \bibnamefont {van Genuchten}},\ and\ \bibinfo {author} {\bibfnamefont
  {A.}~\bibnamefont {G{\"a}rden{\"a}s}},\ }\href@noop {} {\bibfield  {journal}
  {\bibinfo  {journal} {J. Hydrol.}\ }\textbf {\bibinfo {volume} {272}},\
  \bibinfo {pages} {14} (\bibinfo {year} {2003})}\BibitemShut {NoStop}%
\bibitem [{\citenamefont {Roberts}\ and\ \citenamefont
  {Strunin}(2004)}]{RobertsStrunin2004}%
  \BibitemOpen
  \bibfield  {author} {\bibinfo {author} {\bibfnamefont {A.~J.}\ \bibnamefont
  {Roberts}}\ and\ \bibinfo {author} {\bibfnamefont {D.~V.}\ \bibnamefont
  {Strunin}},\ }\href {https://doi.org/10.1093/qjmam/57.3.363} {\bibfield
  {journal} {\bibinfo  {journal} {Quarterly Journal of Mechanics and Applied
  Mathematics}\ }\textbf {\bibinfo {volume} {57}},\ \bibinfo {pages} {363}
  (\bibinfo {year} {2004})}\BibitemShut {NoStop}%
\bibitem [{\citenamefont {van Genuchten}\ and\ \citenamefont
  {Wierenga}(1976)}]{vanGenuchtenWierenga1976}%
  \BibitemOpen
  \bibfield  {author} {\bibinfo {author} {\bibfnamefont {M.~T.}\ \bibnamefont
  {van Genuchten}}\ and\ \bibinfo {author} {\bibfnamefont {P.~J.}\ \bibnamefont
  {Wierenga}},\ }\href@noop {} {\bibfield  {journal} {\bibinfo  {journal} {Soil
  Sci. Soc. Am. J.}\ }\textbf {\bibinfo {volume} {40}},\ \bibinfo {pages} {473}
  (\bibinfo {year} {1976})}\BibitemShut {NoStop}%
\bibitem [{\citenamefont {Burnett}(1936)}]{Burnett1936}%
  \BibitemOpen
  \bibfield  {author} {\bibinfo {author} {\bibfnamefont {D.}~\bibnamefont
  {Burnett}},\ }\href@noop {} {\bibfield  {journal} {\bibinfo  {journal} {Proc.
  London Math. Soc. (2)}\ }\textbf {\bibinfo {volume} {40}},\ \bibinfo {pages}
  {382} (\bibinfo {year} {1936})}\BibitemShut {NoStop}%
\bibitem [{\citenamefont {Fenichel}(1979)}]{Fenichel1979}%
  \BibitemOpen
  \bibfield  {author} {\bibinfo {author} {\bibfnamefont {N.}~\bibnamefont
  {Fenichel}},\ }\href@noop {} {\bibfield  {journal} {\bibinfo  {journal} {J.
  Differ. Equ.}\ }\textbf {\bibinfo {volume} {31}},\ \bibinfo {pages} {53}
  (\bibinfo {year} {1979})}\BibitemShut {NoStop}%
\bibitem [{\citenamefont {Jones}(1995)}]{Jones1995}%
  \BibitemOpen
  \bibfield  {author} {\bibinfo {author} {\bibfnamefont {C.~K. R.~T.}\
  \bibnamefont {Jones}},\ }in\ \href@noop {} {\emph {\bibinfo {booktitle}
  {Dynamical Systems (Montecatini Terme, 1994)}}},\ \bibinfo {series} {Lecture
  Notes in Math.}, Vol.\ \bibinfo {volume} {1609}\ (\bibinfo  {publisher}
  {Springer},\ \bibinfo {address} {Berlin},\ \bibinfo {year} {1995})\
  p.~\bibinfo {pages} {44}\BibitemShut {NoStop}%
\bibitem [{\citenamefont {Gorban}\ and\ \citenamefont
  {Karlin}(2005)}]{Gorban2005}%
  \BibitemOpen
  \bibfield  {author} {\bibinfo {author} {\bibfnamefont {A.~N.}\ \bibnamefont
  {Gorban}}\ and\ \bibinfo {author} {\bibfnamefont {I.~V.}\ \bibnamefont
  {Karlin}},\ }\href@noop {} {\emph {\bibinfo {title} {Invariant Manifolds for
  Physical and Chemical Kinetics}}},\ \bibinfo {series} {Lecture Notes in
  Physics}, Vol.\ \bibinfo {volume} {660}\ (\bibinfo  {publisher} {Springer},\
  \bibinfo {address} {Berlin},\ \bibinfo {year} {2005})\BibitemShut {NoStop}%
\bibitem [{\citenamefont {Hochs}\ and\ \citenamefont
  {Roberts}(2019)}]{HochsRoberts2019}%
  \BibitemOpen
  \bibfield  {author} {\bibinfo {author} {\bibfnamefont {P.}~\bibnamefont
  {Hochs}}\ and\ \bibinfo {author} {\bibfnamefont {A.~J.}\ \bibnamefont
  {Roberts}},\ }\href {https://doi.org/10.1016/j.jde.2019.07.021} {\bibfield
  {journal} {\bibinfo  {journal} {Journal of Differential Equations}\ }\textbf
  {\bibinfo {volume} {267}},\ \bibinfo {pages} {7263} (\bibinfo {year}
  {2019})}\BibitemShut {NoStop}%
\bibitem [{\citenamefont {Alexander}\ and\ \citenamefont
  {Orbach}(1982)}]{AlexanderOrbach1982}%
  \BibitemOpen
  \bibfield  {author} {\bibinfo {author} {\bibfnamefont {S.}~\bibnamefont
  {Alexander}}\ and\ \bibinfo {author} {\bibfnamefont {R.}~\bibnamefont
  {Orbach}},\ }\href {https://doi.org/10.1051/jphyslet:019820043017062500}
  {\bibfield  {journal} {\bibinfo  {journal} {Journal de Physique Lettres}\
  }\textbf {\bibinfo {volume} {43}},\ \bibinfo {pages} {L625} (\bibinfo {year}
  {1982})}\BibitemShut {NoStop}%
\end{thebibliography}%

\begin{acknowledgments}
The Author thanks Chiara Baldo and Lorenzo Duchi who assisted the creation of the MOOC on soil equilibrium theory from which he ultimately found the key to develop this new one. He is especially grateful to A.~J. (Tony) Roberts, who read the first arXiv version of this manuscript and generously wrote to point out the connection with his invariant-manifold framework for multiscale modelling: his correspondence directly motivated the finite-Damk\"ohler validity statements, the nested-family formulation of the multi-permeability hierarchy, and the band--mode equivalence now discussed in Secs.~\ref{sec:multiperm} and the geometric interpretation, and corrected an earlier, imprecise statement about the role of multiple spectral gaps.  This study was carried out within the Space It Up project funded by the Italian Space Agency, ASI, and the Ministry of University and Research, MUR, under contract n. 2024-5-E.0 - CUP n. I53D24000060005.
\end{acknowledgments}

\appendix

These appendices carry out in full the algebra that the main text
states.  Where a step is a standard result rather than a consequence of
the present construction we cite it at the point of use: the Fredholm
alternative~\cite{Fredholm1903}, the spectral theory of weighted graph
Laplacians and the Fiedler bound~\cite{Fiedler1973}, the
Chapman--Enskog procedure and the moment identification of transport
coefficients~\cite{ChapmanCowling1970,Cercignani1988}, the Green--Kubo
form of a transport coefficient~\cite{Kubo1957}, the Schur-complement
identity~\cite{Gorban2005}, the invariant-manifold theorems that certify
the reduction at finite scale
ratios~\cite{Roberts2015ima,BunderRoberts2021,HochsRoberts2019,Roberts2025tma},
and, for the classical closures recovered in
App.~\ref{app:chi}, Burdine~\cite{Burdine1953} and
Mualem~\cite{Mualem1976}.  Statements not so marked are derived here.

\section{The gradient expansion and the entangled flux}
\label{app:gradient}

This appendix supplies the spatial-limit algebra that the 
paper defers, deriving Eqs.~\eqref{eq:Cd}--\eqref{eq:Fentangled} and the
first-order source~\eqref{eq:source}.

\subsection{Expansion of the boundary operators}

We carry out the expansion one step at a time.  Start from the boundary
gain and loss~\eqref{eq:Gbnd}--\eqref{eq:Lbnd}, in which the only
difference from the internal operators is that the neighbour
distribution is evaluated at $\mathbf{x}_+=\mathbf{x}+\ell\hat n$.

\emph{Step 1 (Taylor expansion).}  Expand the neighbour distribution in
the REV spacing $\ell$,
\begin{equation}
g(r',\mathbf{x}_+)=g(r',\mathbf{x})
  +\ell\,\hat n\!\cdot\!\nabla g(r',\mathbf{x})+O(\ell^2).
\label{eq:app_taylor}
\end{equation}

\emph{Step 2 (insert and split).}  Substitute~\eqref{eq:app_taylor} into
$\mathcal{G}_\partial$ and $\Lg_\partial$.  Each splits into a piece with
$g(r',\mathbf{x})$ (the local value, order $\ell^0$) and a piece with
$\ell\,\hat n\!\cdot\!\nabla g(r',\mathbf{x})$ (order $\ell^1$):
\begin{align}
\mathcal{G}_\partial&=\kappa[1{-}g(r)]\!\int\!C_\partial\rho_w\mathrm{g}\,\Phi(r,r')
  \big[g(r')+\ell\hat n\!\cdot\!\nabla g(r')\big]f'dr',\notag\\
\Lg_\partial&=\kappa\,g(r)\!\int\!C_\partial\rho_w\mathrm{g}\,\Phi(r',r)
  \big[1{-}g(r')-\ell\hat n\!\cdot\!\nabla g(r')\big]f'dr'.\notag
\end{align}

\emph{Step 3 (zeroth order cancels).}  The $\ell^0$ part of
$\mathcal{G}_\partial-\Lg_\partial$ is, term by term, the internal
redistribution $\mathcal{G}_{\rm int}-\Lg_{\rm int}$ of
Eqs.~\eqref{eq:G}--\eqref{eq:L} with $C_\partial$ in place of $C$.  It
carries no spatial derivative and so contributes nothing to transport
between REVs; it is already accounted for by $\C$.

\emph{Step 4 (collect first order).}  Write
\begin{equation}
I(r):=\int_0^\infty C_\partial(r,r')\,\rho_w\mathrm{g}\,\Phi(r,r')\,\hat n\!\cdot\!
\nabla g(r')\,f'\,dr'
\end{equation}
 for the common gradient integral.  The gain
contributes $+\ell\,\kappa(r)[1-g(r)]\,I(r)$.  The loss contributes
$-\ell\,\kappa(r)\,g(r)\,I(r)$ to $\Cd$: \emph{two} sign flips occur
inside $\Lg_\partial$,the antisymmetry $\Phi(r',r)=-\Phi(r,r')$, and
the vacancy factor $[1-g(r',\mathbf{x}_+)]$, whose gradient part is
$-\ell\,\hat n\!\cdot\!\nabla g(r')$,and they \emph{cancel}, so the
loss retains the relative minus it carries in
$\Cd=\mathcal{G}_\partial-\Lg_\partial$.  Hence the surviving $\ell^1$
part is
\begin{align}
\Cd&=\ell\,\kappa(r)\Big\{[1-g(r)]
  \!\int\! C_\partial\rho_w\mathrm{g}\,\Phi(r,r')\,\hat n\!\cdot\!\nabla g(r')\,f'\,dr'\notag\\
&\qquad-\,g(r)\!\int\! C_\partial\rho_w\mathrm{g}\,\Phi(r,r')\,\hat n\!\cdot\!\nabla g(r')\,
  f'\,dr'\Big\}+O(\ell^2).
\end{align}

\emph{Step 5 (combine occupancy factors).}  The receiver gate
$[1-g(r)]$ comes from $\mathcal{G}_\partial$ (an empty class $r$
receiving water) and the donor gate $g(r)$ from $\Lg_\partial$ (a full
class $r$ giving water up).  Because the two flips inside the loss have
cancelled, the gates enter $\Cd$ with \emph{opposite} signs, and the two
identical integrals combine with the net occupancy weight
\begin{equation}
[1-g(r)]-g(r)=[1-2g(r)],
\label{eq:app_gate}
\end{equation}
which is the factor quoted in~\eqref{eq:Cd}.  It is the signature of the
exclusion rule: a half-filled class ($g=\tfrac12$) is neutral, a nearly
empty one ($g\to0$) is a net receiver ($+1$), a nearly full one
($g\to1$) a net donor ($-1$).

\emph{Step 6 (read off the flux).}  Writing the result as a divergence,
$\Cd=-\nabla\!\cdot\!\mathbf{F}+O(\ell^2)$, identifies the entangled
flux~\eqref{eq:Fentangled} with kernel
\begin{equation}
\Gamma(r,r';g)=\kappa(r)[1-2g(r)]C_\partial(r,r')\rho_w\mathrm{g}\,\Phi(r,r')f(r')
\end{equation}
 which
is Eq.~\eqref{eq:Gamma}.

\subsection{Localization on the frontier}

On the equilibrium step $g^{(0)}(r';\theta)=H(r^*(\theta)-r')$ the
$r'$-dependence is a sharp front, so its spatial gradient is a delta at
the frontier.  Explicitly, differentiating the step with respect to
$\mathbf{x}$ through its only $\mathbf{x}$-dependence, the frontier
$r^*(\theta(\mathbf{x}))$, and using
$\partial_a H(a)=\delta(a)$,
\begin{equation}
\nabla g^{(0)}(r')=\delta\!\big(r^*-r'\big)\,\nabla r^*
  =\delta(r'-r^*)\,\frac{dr^*}{d\theta}\,\nabla\theta.
\end{equation}
The frontier derivative follows from the defining relation
$\theta=\phi\int_0^{r^*}f\,dr$, which on differentiation gives
$d\theta=\phi f(r^*)\,dr^*$, hence $dr^*/d\theta=1/[\phi f(r^*)]$.
Therefore
\begin{equation}
\nabla g^{(0)}(r')=\frac{\delta(r'-r^*)}{\phi f(r^*)}\,\nabla\theta,
\label{eq:app_stepgrad}
\end{equation}
which collapses the $r'$-integral in~\eqref{eq:Fentangled} onto the
frontier and gives the pore-resolved flux on the equilibrium step,
$\mathbf{F}^{(0)}(r):=\mathbf{F}[g^{(0)}](r)$,
\begin{equation}
\mathbf{F}^{(0)}(r)=-\frac{\kappa(r)}{\phi}\,
  \big[1-2g^{(0)}(r)\big]\,C_\partial(r,r^*)\,
  \Phi(r,r^*)\,\nabla\theta .
\label{eq:app_j}
\end{equation}
The occupancy weight~\eqref{eq:app_gate} must be kept, and it is what
makes the flux well-signed on both sides of the frontier.  On the step
$[1-2g^{(0)}(r)]=-1$ for $r<r^*$ (full: a donor) and $+1$ for $r>r^*$
(empty: a receiver); the driving head flips with it, since
$\Phi(r,r^*)<0$ for $r<r^*$ and $>0$ for $r>r^*$.  The two sign changes
cancel, so the product is positive throughout,
\begin{equation}
\big[1-2g^{(0)}(r)\big]\,\Phi(r,r^*)=\big|\Phi(r,r^*)\big|,
\label{eq:app_absPhi}
\end{equation}
exactly the cancellation already met in the relaxation
rate~\eqref{eq:A}, where $\Phi(r,r')[2g_\eq(r')-1]=|\Phi(r,r')|$.  Hence
\begin{equation}
\mathbf{F}^{(0)}(r)=-\frac{\kappa(r)}{\phi}\,C_\partial(r,r^*)\,
  \big|\Phi(r,r^*)\big|\,\nabla\theta,
\label{eq:app_j2}
\end{equation}
and, through $\Cd=-\nabla\!\cdot\!\mathbf{F}^{(0)}$, the
source~\eqref{eq:source}.  The macroscopic flux is
$\mathbf{q}_{\rm bnd}=\phi\int\mathbf{F}^{(0)}f\,dr$.

\subsection{What the $\kappa$-moment gives, and what it does not}
\label{app:fluxid}
It is worth being explicit about which object is the macroscopic Darcy
flux, because two candidates appear in the derivation and they are
\emph{not} the same.

The macroscopic flux is the $\kappa$-weighted moment of the
\emph{deviation},
\begin{equation}
\mathbf{q}=\phi\!\int_0^\infty \kappa(r)\,g^{(1)}(r)\,f(r)\,dr
   =-K(\theta)\,\nabla\psi ,
\label{eq:app_qtr}
\end{equation}
the second equality being the substitution
$g^{(1)}=-\bm\chi(r)\!\cdot\!\nabla\psi$ of Sec.~\ref{sec:K}: since
$\nabla\psi$ does not depend on $r$ it comes out of the integral,
\begin{equation}
\mathbf{q}=\phi\!\int_0^\infty\!\kappa\,\big[{-}\bm\chi\!\cdot\!\nabla\psi\big]f\,dr
=-\Big[\underbrace{\phi\!\int_0^\infty\!\kappa\,\bm\chi\,f\,dr}_{=\;K(\theta)}\Big]
   \!\cdot\!\nabla\psi ,
\end{equation}
which is precisely the definition~\eqref{eq:K_exact} of $K$.
the throughflow each class carries once the response $g^{(1)}$ is known
(Sec.~\ref{sec:K}).  This is the standard kinetic-theory identification:
transport coefficients are moments of the \emph{first-order correction},
never of the streaming of the local equilibrium, in the gas, the viscous
stress and heat flux are moments of $h_1$, whereas the same moments taken
on the \emph{Maxwellian} (the equilibrium velocity distribution of a gas,
the counterpart of our equilibrium step $g_\eq$) carry no dissipative
content at all~\cite[\S\,V.3]{Cercignani1988}.

The other object, $\mathbf{F}^{(0)}=\mathbf{F}[g^{(0)}]$ of
Eq.~\eqref{eq:app_j2}, is \emph{not} a flux to be read off: it is the
streaming of the local equilibrium, and its divergence is the
\emph{source} of the first-order equation, this is exactly how
Eq.~\eqref{eq:source} was obtained, with
$\nabla\!\cdot\!\mathbf{F}[g^{(0)}]$ standing on the \emph{right}-hand
side of~\eqref{eq:first_order}.  It is the soil-water counterpart of
$\bm\xi\!\cdot\!\nabla f_0$, which drives the CE problem but
is not itself a transport flux.  Contracting it against $f$ and calling
the result the Darcy flux returns precisely the naive conductivity
$K_{\rm naive}$ refuted in App.~\ref{app:chi}: it bypasses $\J^{-1}$, and
with it the entire relaxation spectrum.

\paragraph*{A sharp diagnostic.}
Equation~\eqref{eq:app_qtr} carries a consequence that the naive form
cannot reproduce and that is worth recording.  The CE
constraint~\eqref{eq:constraint} makes $g^{(1)}$ mass-neutral,
$\int g^{(1)}f\,dr=0$.  Hence, were the single-pore conductance
class-independent, $\kappa(r)\equiv\kappa_0$, we would have
\begin{equation}
\mathbf{q}=\phi\,\kappa_0\!\int g^{(1)}f\,dr=0 :
\label{eq:app_kconst}
\end{equation}
a medium whose pores all conduct alike transports nothing at first order,
because the excess the deviation places in some classes is exactly
compensated by the deficit it leaves in others.  Conduction requires the
conductance to \emph{correlate} with the non-equilibrium occupancy, the
Green--Kubo content of~\eqref{eq:K_spectral}, and, in the projector
language of App.~\ref{app:chi}.
 It is the
pore-size dependence of $\kappa$ (the $r^2$ of Hagen--Poiseuille) that
makes the medium conduct at all.  The naive contraction of
$\mathbf{F}^{(0)}$ does not vanish in this limit, which is an independent
way of seeing that it is not the flux.

\section{Spectral properties of $\J$}
\label{app:spectral}

This appendix gives the proofs of the three properties of $\J$ used in
Sec.~\ref{sec:operator}, together with the slaving estimate used in
Sec.~\ref{sec:step2}.  Throughout, $\J$ is the graph-Laplacian
\eqref{eq:graphlap} with conductances~\eqref{eq:Krr}, and the weighted
inner product is
$\langle u,v\rangle_f=\int u\,v\,f\,dr$, the mass measure.
The treatment transcribes Cercignani's
analysis of the linearized collision
operator~\cite[\S IV.1,\S IV.6]{Cercignani1988} to the single-invariant
setting.

Water mass is conserved by the redistribution operator for \emph{any}
distribution: $\int\C[g]\,f\,dr=0$ identically in $g$.  Differentiating
this identity with respect to $g$ in the direction $h$, the order of
the (linear) integration and the Fr\'echet derivative may be exchanged,
giving

\begin{equation}
\frac{d}{d\epsilon}\Big|_{0}
  \int\C[g_\eq+\epsilon h]\,f\,dr=\int\J[h]\,f\,dr=0
\end{equation}
for every $h$.  Thus the constant function $1$ lies in the (left) null
space of $\J$ in the measure $f\,dr$: $\mathrm{coker}\,\J=\ker\J^\dagger
=\mathrm{span}\{1\}$.  On a connected network the graph-Laplacian structure
makes the constants the \emph{only} functions annihilated by the Laplacian
core $\Lop$, so the cokernel is one-dimensional; the kernel of $\J$ itself
sits one stiffness factor away, $\ker\J=\mathrm{span}\{m\}
=\mathrm{span}\{\partial_\theta g_\eq\}$ (Prop.~\ref{prop:selfadj}).  This single
invariant direction is the single redistribution (collision in Cercignani's treatment) invariant, water mass, and is
the origin of the single macroscopic equation.

Recall the conductances in full,
$K_{rr'}=\Mr(r,r';g_\eq)f(r')
 =\kappa_s(r, r')C(r, r')\sqrt{m(r)\,m(r')}\,f(r')$
[Eqs.~\eqref{eq:Mexplicit},~\eqref{eq:Krr}].
With the weight $w$ above, the conductances obey the detailed-balance
symmetry $K_{rr'}f_r=K_{r'r}f_{r'}$, which is exactly the condition for
a weighted graph-Laplacian to be self-adjoint in
$\langle\cdot,\cdot\rangle_f$.  Its quadratic form is the standard
edge-sum
\begin{equation}
\langle u,\Lop u\rangle_f
  =-\tfrac12\sum_{r,r'}K_{rr'}\,f_r\,(u_r-u_{r'})^2\le0,
\end{equation}
nonpositive because every conductance $K_{rr'}\ge0$, and zero only when
$u$ is constant across each connected component.  This is the soil-water
counterpart of Cercignani's
$(h,Lh)=-\tfrac1{4m}\iint f_0 f_{0*}|\Delta h|^2 B\,d\xi\le0$
\cite[Eq.~(IV.1.11)]{Cercignani1988}, with the squared difference
$|\Delta h|^2$ across a collision playing the role of $(u_r-u_{r'})^2$
across a network edge.

Self-adjointness gives a real spectrum and an orthonormal eigenbasis;
the sign result orders it as $\lambda_0=0>\lambda_1\ge\lambda_2\ge\cdots$,
with $\lambda_0=0$ the mass mode.  Restricted to the orthogonal
complement of the kernel in $\langle\cdot,\cdot\rangle_{f/m}$ (note
$\langle h,m\rangle_{f/m}=\int hf\,dr$), $\{h:\int hf\,dr=0\}
=(\ker\J)^{\perp}$, every
eigenvalue is strictly negative, so $\J$ is invertible there with
\begin{equation}
\|\J^{-1}\|=\frac{1}{|\lambda_1|}=\tau_{\rm relax},
\end{equation}
the slowest relaxation time (Prop.~\ref{prop:gap}).  The bounded inverse
on $(\ker\J)^\perp$ is what makes the first-order
equation~\eqref{eq:first_order} solvable and the CE expansion
asymptotic.

Finally we justify dropping $\partial_t g^{(1)}$ from the first-order
equation.  Because $g^{(1)}$ depends on space and time only through
$\theta$ and $\nabla\theta$ [the ansatz~\eqref{eq:ansatz}], its time
derivative is set by the slow macroscopic evolution,
$\partial_t g^{(1)}\sim(\partial g^{(1)}/\partial\theta)\,
\partial_t\theta=O(\Da)$ in the non-dimensional units of
Sec.~\ref{sec:Da}.  Hence the term carrying it in the singularly scaled
equation~\eqref{eq:singular} is $\Da\,\partial_t g^{(1)}=O(\Da^2)$, one
order beyond the $O(\Da^0)$ balance that defines $g^{(1)}$, and is
consistently neglected.  This is the precise statement that, after a few
relaxation times $\tau_{\rm relax}$, $g^{(1)}$ is \emph{slaved} to the
instantaneous macroscopic state.  The estimate degrades as
$\lambda_1\to0$ ($\tau_{\rm relax}\to\infty$, $\Da\to\infty$) at the
percolation threshold, where slaving, and with it the whole
reduction, fails.

\section{Mean-field and serial-path approximations}
\label{app:chi}
\textbf{Mean-field.}  Setting $\mathcal{B}\approx0$ in~\eqref{eq:chi},
$\bm\chi_{\rm mf}=S_0 C_w/\mathcal{A}$, and
$k_{\rm mf}(r)=\kappa(r)\bm\chi_{\rm mf}(r)=
r^2\rho_w\mathrm{g}\,a_w(r)\bar C(r)H(r^*\!-r)/(8\mu\bar\tau^2)$, giving
$K_{\rm mf}$~\eqref{eq:K_mf}: each filled, connected pore contributes its
Hagen--Poiseuille throughflow independently, the independent-pore
(Burdine) estimate.

\textbf{Serial-path.}  The off-diagonal $\mathcal{B}$ couples classes.
Under random pairing $C=C_0$ the Neumann series
\begin{equation}
\bm\chi=-\sum_{k\ge0}(\mathcal{A}^{-1}\mathcal{B})^k\mathcal{A}^{-1}
[S_0C_w]
\end{equation}
resums to the harmonic mean of pore conductances along a serial
flow path of $n=\delta x/\bar L$ pores, replacing $\langle r^2\rangle$
by $\langle r^{-2}\rangle^{-1}$ and giving~\eqref{eq:K_sp}.  For a
log-normal PSD the ratio is
$\langle r^{-2}\rangle^{-1}/\langle r^2\rangle=\exp(-4\sigma^2)$,
Mualem's heterogeneity penalty.

\textbf{Why no moment can substitute for the inversion.}
We give here a second demonstration that the failure of the naive
identification is not a matter of accuracy but
of \emph{support}, and it is sharpest when stated with the null
projector.  Because $\int f\,dr=1$, let
\begin{equation}
\mathsf{P}\,u:=\Big(\int_0^\infty u\,f\,dr\Big)\cdot 1,
\qquad \mathsf{Q}:=\mathcal{I}-\mathsf{P},
\label{eq:app_proj}
\end{equation}
where $\mathcal{I}$ is the identity, $\mathsf{P}$ maps a function $u(r)$ to
the constant equal to its $f$-average, and $\mathsf{Q}$ retains the
fluctuating remainder.  Both are dimensionless (they act on $u$ without
changing its units), and both are idempotent,
$\mathsf{P}^2=\mathsf{P}$, $\mathsf{Q}^2=\mathsf{Q}$, $\mathsf{PQ}=0$.
Thus $\mathsf{P}$ is the projector onto $\ker\J=\mathrm{span}\{1\}$:
\emph{taking the $f$-moment and projecting onto the null space of $\J$
are the same operation}.  Three facts follow.

\emph{(i) The null projection annihilates the operator.}  $\J$ is
self-adjoint with $\ker\J=\mathrm{span}\{1\}$
(Prop.~\ref{prop:selfadj}), so
\begin{equation}
\mathsf{P}\J=\J\,\mathsf{P}=0,
\qquad \J=\mathsf{Q}\,\J\,\mathsf{Q},
\label{eq:app_PJ}
\end{equation}
i.e.\ $\J$ is supported entirely on 
\begin{equation}
W:=\mathrm{ran}\,\mathsf{Q}
=(\ker\J)^\perp 
\end{equation}
the orthogonal complement of the kernel in the mass measure, the space of
zero-mean fluctuations, $W=\{u:\int u f\,dr=0\}$,so $\J$ is invisible to
$\mathsf{P}$.

\emph{(ii) Both the source and the response live in $W$.}  The
solvability condition of Sec.~\ref{sec:solvability} is precisely
$\mathsf{P}S=0$, and the CE constraint~\eqref{eq:constraint} is
precisely $\mathsf{P}g^{(1)}=0$.

\emph{(iii) Hence the $f$-moment of the first-order equation is empty.}
(Here $S$ denotes the first-order source of Eq.~\eqref{eq:source}, the
divergence of the equilibrium flux.)
Applying $\mathsf{P}$ to $\J[g^{(1)}]=S$ returns, by (i)--(ii), the
identity $0=0$: it reproduces the solvability condition, mass
conservation, and \emph{nothing else}.  The macroscopic budget is
exactly the component of the first-order balance that the null
projection retains, and it carries no transport information whatever.
All of $K$ resides in the complementary block, on which
\begin{equation}
\J^{-1}\big|_W=\sum_{n\ge1}\frac{1}{\lambda_n}\,
  \varphi_n\otimes\varphi_n\;\neq\;\mathcal{I},
\label{eq:app_Jinv}
\end{equation}
in words: on $W$ the inverse operator acts mode by mode, projecting onto
each eigenvector $\varphi_n$ and dividing by its rate $\lambda_n$, so that
each mode is weighted by its own relaxation time $\tau_n=1/\lambda_n$.
The $n=0$ term is absent because $\lambda_0=0$ is excluded from $W$: that
is the conserved mode, which never relaxes.  Crucially $\J^{-1}|_W\neq
\mathcal{I}$, so replacing the inversion by the identity,the naive
shortcut, discards exactly the relaxation times that carry the
transport.
Equivalently, in the operator language of App.~\ref{app:schur}, $K$ is
the Schur complement $\mathsf{P}\mathsf{T}\mathsf{Q}\,
(\J|_\mathsf{Q})^{-1}\,\mathsf{Q}\mathsf{T}\mathsf{P}$: the gradient must
be carried \emph{out} of the kernel by $\mathsf{T}$, relaxed by
$(\J|_\mathsf{Q})^{-1}$ in $W$, and carried back, a round trip through
$W$ that no moment taken in $\ker\J$ can shortcut.

Consequently only the kernel-orthogonal parts of the conductance and of
the source contribute,
\begin{equation}
K=\phi\,\big\langle\,\mathsf{Q}\kappa\;\big|\;\J^{-1}\;\big|\;
  \mathsf{Q}(S_0C_w)\,\big\rangle_f ,
\label{eq:app_KQ}
\end{equation}
the constant (mass) component of $\kappa$ pairing with a $W$-vector and
dropping out.  The naive identification
$K_{\rm naive}=\phi\langle\kappa\,|\,S_0C_w\rangle_f$ amounts to
replacing $\J^{-1}|_W$ by $\mathcal{I}$,assigning every relaxation
mode the same unit rate, and so discards the spectrum
$\{\lambda_n\}$ altogether.  It is not an approximation of $K$ but a
different functional: it is not even dimensionally a conductivity, since
$\J^{-1}$ carries the relaxation time that the substitution deletes; and
it stays finite as $\lambda_1\to0$, whereas the true $K$
diverges~\eqref{eq:K_spectral}.  The naive form is therefore blind to
the breakdown the reduction exists to detect.

\section{Band projection and the exchange coefficient}
\label{app:bands}

Let $\mathsf{B}=\{r>r_\times\}$, $\mathsf{T}=\{r<r_\times\}$ split the
spectrum at the valley $r_\times$.  Decompose
$\C=\C_{bb}+\C_{tt}+\C_{bt}+\C_{tb}$, where $\C_{ab}$ moves water from
band $b$ to band $a$; mass conservation gives
$\int_\mathsf{B}\C_{bt}f\,dr=-\int_\mathsf{T}\C_{tb}f\,dr$.  The
intra-band operators inherit the gap $|\lambda_1^{b}|,|\lambda_1^{t}|$
and define band Damk\"ohler numbers $\Da_b,\Da_t\ll1$; the cross-band
operator carries the small rate $1/\tau_t$ and is $O(\Lambda)$ relative
to the body band.  Applying the CE reduction within each band slaves
$g|_\mathsf{B}\to g_\eq^b(\theta_b)$, $g|_\mathsf{T}\to
g_\eq^t(\theta_t)$, each with its frontier $r^*_b, r^*_t$, retention
$\psi_b,\psi_t$, and conductivity $K_b,K_t$ from the band-restricted
\eqref{eq:K_mf}.  The $f$-weighted moment over each band gives the
coupled system~\eqref{eq:dpb}--\eqref{eq:dpt}, with the exchange
\begin{align}
\Gamma_w&=\phi\int_\mathsf{B}\C_{bt}[g_\eq^b,g_\eq^t]\,f\,dr\notag\\
&=\frac{\rho_w\mathrm{g}}{\bar L^2}
  \iint_{\mathsf{B}\times\mathsf{T}}\!\!\kappa_s C\,f f'
  \big[\mu_w^t-\mu_w^b\big]\,dr\,dr'/(\rho_w\mathrm{g}),
\end{align}
where the square brackets are ordinary grouping: $[\mu_w^t-\mu_w^b]$ is
the thermodynamic force of Eq.~\eqref{eq:jsplit}, the potential difference
between a throat class $r'\in\mathsf{T}$ and a body class
$r\in\mathsf{B}$, carried by the pair mobility $\kappa_s C f f'$.  The
expression is linear in that difference,
$\mu_w^t-\mu_w^b=\rho_w\mathrm{g}(\psi_t-\psi_b)$ to leading order,
giving the first-order law $\Gamma_w=\alpha_w(\psi_t-\psi_b)$ with
$\alpha_w$ as in~\eqref{eq:alpha}.  The band potentials are not free
labels: each is set by its own band frontier, $\psi_a=-2\gamma\cos\theta_c
/(\rho_w\mathrm{g}\,r^*_a)$ with $r^*_a$ the frontier radius of band
$a\in\{\mathsf{B},\mathsf{T}\}$ fixed by that band's water content through
its own retention curve~\eqref{eq:psi}.  Linearization is controlled by the intra-band
$\Da$; at $\Da\sim1$ the exchange becomes nonlinear in
$(\psi_t-\psi_b)$, the kinetic generalization of the Gerke--van~Genuchten
transfer.  The $N$-band case decomposes $\C$ into $N$ diagonal blocks
plus an off-diagonal exchange matrix $\{\C_{ij}\}$, giving $N$ Richards
equations coupled by $\alpha_{ij}$.

\section{The linearized operator the classical way:
a step-by-step Fr\'echet construction}
\label{app:cercignani}

The main text builds $\J$ from the conservative form of $\C$
(Sec.~\ref{sec:operator}).  Readers coming from kinetic theory will
expect  anoother, older construction, the one Cercignani uses for the
Boltzmann collision operator, in which the gain and the loss are kept
separate.  We first exhibit the explicit redistribution current they
assemble into, the origin of the occupancy gates that the main-text
mobility carries, and then differentiate it term by term to recover the
operator.  Done consistently, with the entropy of mixing carried inside the
potential, as below, the route reaches the same $\J$; it is more concrete:
every term keeps a transparent meaning on the pore network.  Because it
is the one a hydrologist or soil scientist is least likely to have met,
we spell out every step and assume no prior kinetic-theory background.

\subsection*{The exclusion current: where the gates of $\Mr$ come from}
The main text carries the redistribution current $j=\Mr\,[\mu_w(r')-
\mu_w(r)]$~\eqref{eq:jsplit} and notes only that the occupancy gates live
inside the mobility $\Mr$.  Here is where they come from.  Assemble the
current from the gain and loss without yet linearizing.  The
redistribution is $\C=\mathcal{G}_{\rm int}-\Lg_{\rm int}$; both
pieces~[Eqs.~\eqref{eq:G}--\eqref{eq:L}] integrate over $r'$ against the
common measure $C(r,r')\,f(r')\,dr'$, so sliding the $r$-only gates
$[1-g(r)]$ and $g(r)$ inside and subtracting collects everything under one
integral,
\begin{widetext}
\begin{equation}
\C[g](r)=\int_0^\infty \kappa(r,g)\,C(r,r')\,
   \big[\,\Phi(r,r')\,g(r')(1{-}g(r))
        -\Phi(r',r)\,g(r)(1{-}g(r'))\,\big]\,f'\,dr'.
\label{eq:app_GminusL}
\end{equation}
\end{widetext}
Comparison with $\C=\int j\,f'\,dr'$~\eqref{eq:jcurrent} reads the
integrand off as $j$.  Two cosmetic moves bring it to standard form.
First, write the signed drivings as \emph{directional weights}
$\omega_{r'\to r}:=\Phi(r,r')$ and $\omega_{r\to r'}:=\Phi(r',r)=
-\,\omega_{r'\to r}$, so that each hop carries the head with the sign of
its own direction.
Second, replace the one-sided rate $\kappa(r,g)$,which would spoil the
antisymmetry, by the symmetric harmonic-mean pair conductance
$\kappa_s(r,r')$ that the throat joining the two pores actually has.  The
result is the \emph{exclusion current}
\begin{widetext}
\begin{equation}
j[g](r,r')=\kappa_s(r,r')\,C(r,r')\,
  \big[\,\omega_{r'\to r}\,g(r')[1-g(r)]
        -\omega_{r\to r'}\,g(r)[1-g(r')]\,\big].
\label{eq:jexpl}
\end{equation}
\end{widetext}
It is antisymmetric by construction: under $r$ exchanges with $r'$ and viceversa, the
symmetric $\kappa_s$ and $C$ are unchanged while the two $\omega$ and the
two occupancy products exchange, flipping the bracket; this is the
algebraic form of water conservation.  Matching to the
Onsager form $j=\Mr[\mu_w(r')-\mu_w(r)]$ identifies the mobility: the
\emph{asymmetry} of $\omega_{r'\to r}$ against $\omega_{r\to r'}$ is the
force $\propto\Phi$, and their symmetric envelope of rates and gates is
$\Mr$.  The single donor/receiver gate $g(r')[1-g(r)]$ becomes the symmetric
geometric-mean gate $\sqrt{m\,m'}$ of Eq.~\eqref{eq:Mexplicit} once the
entropy of mixing is carried inside $\mu_w=\mu_0+\psi_T\ln[g/(1-g)]$ rather
than inside the rates $\omega$, \emph{and} the pair rates are taken of
Arrhenius (heat-bath) form; the correspondence between rate law and gate
content is made precise below.

The redistribution operator $\C[g]$ is nonlinear in the filling field
$g(r)$: it contains products like $g(r')[1-g(r)]$ (a pore of radius $r'$
can only \emph{send} water to a pore of radius $r$ if the first is full
and the second has room).  We want its behaviour for states that sit
\emph{close} to the equilibrium step $g_\eq$.  Write the state as the
step plus a small disturbance,
\begin{equation}
g(r)=g_\eq(r)+\epsilon\,h(r),
\label{eq:app_perturb}
\end{equation}
with $\epsilon$ a bookkeeping smallness parameter and $h$ the shape of the
disturbance.  Substituting into $\C$ and keeping only the part that grows
linearly with $\epsilon$ defines the \emph{linearized operator}
\begin{equation}
\J[h]:=\left.\frac{d}{d\epsilon}\,\C[g_\eq+\epsilon h]\right|_{\epsilon=0}
      \equiv\C'\big|_{g_\eq}[h].
\label{eq:app_frechet}
\end{equation}
This is exactly the operation a hydrologist already performs when
linearizing a nonlinear storage or conductivity term around an operating
point, here the ``operating point'' is the whole equilibrium profile
$g_\eq$, and the derivative is taken with respect to a function rather
than a number (a Fr\'echet derivative).  The zeroth-order part vanishes
because $g_\eq$ is an equilibrium ($\C[g_\eq]=0$); the first-order part is
$\J[h]$; higher orders are dropped.

{\subsection*{Carry the entropy in the potential, then perturb}
As written, the driving in~\eqref{eq:jexpl} is the capillary head alone, and
the step is \emph{not} an equilibrium of that current: for a pair straddling
the frontier ($r<r^*<r'$) the second gate product $g(r)[1-g(r')]$ is open and
$j\neq0$.  Differentiating~\eqref{eq:jexpl} term by term with $\Phi$ held
fixed therefore does \emph{not} produce the operator of the main text: it
produces an operator with no $1/m$ stiffness, a sign-indefinite diagonal, and
growing modes.  The missing ingredient is the one the main text builds in
from the start: the entropy of mixing belongs inside the potential.  Replace
the bare capillary head by the full lattice-gas potential,
\begin{equation}
\Phi[g](r,r')=\frac{\mu_w[g](r')-\mu_w[g](r)}{\rho_w\mathrm{g}},\quad
\mu_w[g]=\mu_{\rm cap}+\psi_T\ln\tfrac{g}{1-g},
\label{eq:app_Phistate}
\end{equation}
so that the driving is state-dependent.  On the resolved equilibrium band the
total potential is uniform, $\mu_w[g_\eq]=\lambda$, hence
$\Phi[g_\eq]\equiv0$ \emph{pairwise}: strong detailed balance, now genuinely.}

{\subsection*{Differentiate: the gate terms die, the stiffness survives}
Perturb $g=g_\eq+\epsilon h$ in~\eqref{eq:jexpl}, with the state-dependent
driving~\eqref{eq:app_Phistate}, and apply the product rule.  Every term in
which the \emph{gates} are differentiated carries the undifferentiated
driving $\Phi[g_\eq]=0$ and drops out.  The only surviving term
differentiates the driving itself,
\begin{equation}
\delta\Phi=\frac{\delta\mu_w(r)-\delta\mu_w(r')}{\rho_w\mathrm{g}},
\qquad
\delta\mu_w=\frac{\psi_T}{m}\,h ,
\label{eq:app_dPhi}
\end{equation}
and this is where the factor $1/m$ enters, once for $h(r)$ (the diagonal)
and once for $h(r')$ (the off-diagonal), exactly as in Eq.~\eqref{eq:J}.
The gate products, frozen at equilibrium, remain behind as the symmetric
envelope multiplying $\delta\Phi$.

\subsection*{The envelope depends on the rate law; the structure does not}
The surviving envelope is the equilibrium one-way exchange flux $\omega$ of
Eq.~\eqref{eq:omega}, and its gate content depends on which microscopic rate
law the exclusion current encodes:
\begin{itemize}
\item a rate \emph{linear} in the driving with the single exclusion gates,
as literally written in~\eqref{eq:jexpl}, leaves the arithmetic combination
$g'_\eq(1-g_\eq)+g_\eq(1-g'_\eq)$;
\item \emph{Arrhenius (heat-bath)} hopping rates,
$w_{r'\to r}\propto\kappa_s\,e^{-[\mu_w(r)-\mu_w(r')]/2\psi_T}$, leave the
geometric mean $\sqrt{m\,m'}$, the mobility adopted in
Eq.~\eqref{eq:Mexplicit};
\item the double product $m\,m'$ arises only if the intrinsic pair
conductance is itself occupancy-gated, which counts the exclusion twice.
\end{itemize}
All members of this detailed-balance family share the structure
$\J=\Lop\circ(\psi_T/m)$, the kernel, the cokernel, and the
self-adjointness statements of Prop.~\ref{prop:selfadj}; they differ in the
spectrum, and the band-edge argument of Sec.~\ref{sec:gap} is what selects
the geometric mean.  The classical route and the conservative route are
therefore the same construction once the same rate law is used in both; what
the component form adds is the transparent origin of the gates, of the
stiffness, and of the family.}

\section{The ordering of the time derivative}
\label{app:dtorder}
Section~\ref{sec:K} removes $\partial_t g^{(0)}$ from the first-order
equation ``by the ordering, not by neglect.''  Here is that ordering in
full, and the place the derivative returns.

\paragraph{The field is slow.}
Integrating the continuum kinetic equation~\eqref{eq:contKE2} against
$f\,dr$ and using mass conservation~\eqref{eq:bulk}, $\int\C f\,dr=0$,
gives the \emph{exact} balance
\begin{equation}
\partial_t\theta+\nabla\!\cdot\!\mathbf{q}=0,\qquad
\theta=\phi\!\int g\,f\,dr,\quad \mathbf{q}=\phi\!\int\mathbf{F}\,f\,dr .
\label{eq:G_mass}
\end{equation}
Every mechanism that moves $\theta$ is $O(\Da)$, so $\theta$ evolves on
the slow clock $T=\Da\,t$~\eqref{eq:slowtime}, with
$\partial_t=\Da\,\partial_T$ on the equilibrium manifold (the fast time is
spent in the initial layer).

\paragraph{Expand every term.}
Insert $\partial_t=\Da\,\partial_T$ into~\eqref{eq:contKE2} and
$g=g^{(0)}+\Da g^{(1)}+\Da^2 g^{(2)}+\cdots$.  With $\J:=\C'|_{g^{(0)}}$,
and $\mathbf{F}$ linear in $g$,
\begin{subequations}
\begin{align}
\C[g]&=\C[g^{(0)}]+\Da\,\J[g^{(1)}]\notag\\
   &\quad+\Da^2\big(\J[g^{(2)}]+\tfrac12\C''[g^{(0)}][g^{(1)},g^{(1)}]\big)+\cdots,\\
\nabla\!\cdot\!\mathbf{F}[g]&=\nabla\!\cdot\!\mathbf{F}[g^{(0)}]
   +\Da\,\nabla\!\cdot\!\mathbf{F}[g^{(1)}]+\cdots,\\
\Da\,\partial_T g&=\Da\,\partial_T g^{(0)}+\Da^2\partial_T g^{(1)}+\cdots .
\end{align}
\end{subequations}
The time term begins at $O(\Da)$: that single fact is the discharge.

\paragraph{Collect powers of $\Da$.}
\begin{subequations}
\begin{align}
O(\Da^{-1}):&\quad \C[g^{(0)}]=0\ \Rightarrow\ g^{(0)}=g_\eq(r;\theta),\\
O(\Da^{0}):&\quad \J[g^{(1)}]=\nabla\!\cdot\!\mathbf{F}[g^{(0)}]
   \quad(\text{no }\partial_T g^{(0)}\text{ term}),\label{eq:G_first}\\
O(\Da^{1}):&\quad \J[g^{(2)}]+\tfrac12\C''[g^{(1)},g^{(1)}]
   =\partial_T g^{(0)}+\nabla\!\cdot\!\mathbf{F}[g^{(1)}].\label{eq:G_second}
\end{align}
\end{subequations}
The first-order equation~\eqref{eq:G_first},which is~\eqref{eq:first_order}
with the $O(\Da)$ time term absent and its source~\eqref{eq:source}
proportional to $\nabla\theta$,fixes $g^{(1)}$; the discharged derivative
sits one order higher, in~\eqref{eq:G_second}.  Equivalently, by the chain
rule, $\partial_t g^{(0)}=(\partial_\theta g_\eq)\,\partial_t\theta
=\Da\,(\partial_\theta g_\eq)\,\partial_T\theta=O(\Da)$, one order below the
$O(1)$ source.

\paragraph{Where it returns.}
Expand the exact balance~\eqref{eq:G_mass} on the slow clock,
$\mathbf{q}=\mathbf{q}^{(0)}+\Da\,\mathbf{q}^{(1)}+\cdots$:
\begin{equation}
O(\Da^0):\ \nabla\!\cdot\!\mathbf{q}^{(0)}=0,\qquad
O(\Da^1):\ \partial_T\theta+\nabla\!\cdot\!\mathbf{q}^{(1)}=0 .
\label{eq:G_return}
\end{equation}
The leading balance is the solvability condition itself: $\mathbf{q}^{(0)}$
is the $\kappa$-moment of the \emph{equilibrium} state $g^{(0)}$, and the
Fredholm condition of Sec.~\ref{sec:solvability} requires precisely that
its divergence carry no net water, $\nabla\!\cdot\!\mathbf{q}^{(0)}=0$;
equivalently, at $O(\Da^0)$ there is no slow time yet for $\theta$ to
evolve on.
With the response $g^{(1)}=-\boldsymbol{\chi}\!\cdot\!\nabla\psi$
[Sec.~\ref{sec:K}] and $\mathbf{q}^{(1)}=-K(\theta)\,\nabla\psi$, the
second reads $\partial_T\theta=\nabla\!\cdot\![K\nabla\psi]$,Richards'
equation.  The removed derivative did not vanish: it dropped one order and
became the macroscopic law.

\paragraph{The loop closes.}
Projecting~\eqref{eq:G_first} onto $\ker\J=\mathrm{span}\{1\}$ and using
self-adjointness ($\J[1]=0$) forces the source orthogonal to the mass
mode, $\int\nabla\!\cdot\!\mathbf{F}[g^{(0)}]\,f\,dr
=\nabla\!\cdot\!\mathbf{q}^{(0)}=0$.  The solvability of the equation from
which $\partial_t g^{(0)}$ was discharged is thus the leading balance
of~\eqref{eq:G_return}, and the discharged derivative is its next order:
the two are the same conservation statement read at consecutive orders.
The single-time (Euler-substitution) form of this computation is the one
carried out in App.~\ref{app:cercignani}.


\section{The reduction in one step: the Schur complement}
\label{app:schur}
The order-by-order hierarchy of Sec.~\ref{sec:hierarchy} has a compact
operator counterpart that uses the continuum form of $\J$ directly.
Write the closed continuum kinetic equation~\eqref{eq:contKE2} (sinks
suppressed) as 
\begin{equation}
\partial_t g=\Da^{-1}\C[g]-\mathsf{T}[g]
\end{equation}
with the
transport operator $\mathsf{T}:=\nabla\!\cdot\!\mathbf{F}$.  Linearizing
both currents about $g_\eq(r;\theta(\mathbf{x}))$ gives the two-scale
generator on $\mathcal{H}=L^2_f(r)\otimes L^2(\Omega)$,
\begin{equation}
\partial_t u=\mathbb{L}\,u-\partial_t g_\eq,\qquad
\mathbb{L}:=\tfrac{1}{\Da}\,\J-\mathsf{T}.
\label{eq:Lgen}
\end{equation}
Here $\J$ is block-diagonal in $\mathbf{x}$ (local, self-adjoint,
$\ker\J=\mathrm{span}\{1\}$, gap $\lambda_1$), so $\Da^{-1}\J$ is fast;
$\mathsf{T}$ couples REVs and is slow.  Let $\mathsf{P}(\mathbf{x})$
project onto $\ker\J$ in $\langle\cdot,\cdot\rangle_f$ and
$\mathsf{Q}=\mathbb{I}-\mathsf{P}$.  Slaving the fast part on the slow
time $T=\Da\,t$ (where $\J|_\mathsf{Q}$ is invertible while the gap is
open),
\begin{equation}
\mathsf{Q}u=\Da\,(\J|_\mathsf{Q})^{-1}\,
   \mathsf{Q}\mathsf{T}\mathsf{P}\,(\mathsf{P}u)+O(\Da^2)=\Da\,g^{(1)}+\cdots,
\label{eq:slave}
\end{equation}
where $\mathsf{Q}\mathsf{T}\mathsf{P}$ acting on $\theta$ is precisely the inter-REV source $S$ of Eq.~\eqref{eq:source}, so \eqref{eq:slave} is the response function~$\bm\chi$ by another route; substituting into the $\mathsf{P}$-equation closes the slow dynamics
with the \emph{Schur complement of the fast block},
\begin{equation}
\mathbb{L}_{\rm eff}=-\mathsf{P}\mathsf{T}\mathsf{P}
   -\Da\,\mathsf{P}\mathsf{T}\mathsf{Q}\,(\J|_\mathsf{Q})^{-1}\,
    \mathsf{Q}\mathsf{T}\mathsf{P}.
\label{eq:schur}
\end{equation}
With $\mathsf{P}u\leftrightarrow\theta$ this is Richards' equation: the
first term is advection/gravity, the second is diffusion with
coefficient
\begin{equation}
K=\phi\,\langle\kappa\,|\,\J^{-1}\,|\,S_0 C_w\rangle_f
 =\phi\sum_{n\ge1}\frac{\langle\kappa,\varphi_n\rangle_f\,
   \langle\varphi_n,S_0 C_w\rangle_f}{\lambda_n},
\label{eq:Kschur}
\end{equation}
identical to the response result~\eqref{eq:K_exact} and to the eigenmode
form of Sec.~\ref{sec:eigenmodes}.  The hierarchy expands the
\emph{equation} order by order; the Schur complement projects the
\emph{operator} once; they agree at leading order because both extract
the slow spectral branch of $\mathbb{L}$.  Two features are sharp in this
form.  The projector $\mathsf{P}$ depends on $\mathbf{x}$ through
$\theta$, so $[\mathsf{P},\nabla]\neq0$ is a connection whose curvature
$[\mathsf{T},\mathsf{P}]$ is the operator home of the non-Fickian
(dynamic-capillarity, hysteretic) corrections, vanishing for a spatially
uniform frontier.  And the reduction exists only while
$(\J|_\mathsf{Q})^{-1}$ is bounded: as $\theta\to\theta_c$ the gap closes,
the slow branch merges with the fast one, and~\eqref{eq:schur}
diverges, the breakdown of Prop.~\ref{prop:gap}, read as the loss of a
spectral gap in the global generator $\mathbb{L}$.


\end{document}


\title{Supplemental Material for\\
\emph{Richards' equation as a hydrodynamic limit}:\\
numerical support for the redistribution spectrum and its band structure}
\author{Riccardo Rigon}
\email{riccardo.rigon@unitn.it}
\affiliation{Centro Agricoltura Alimenti Ambiente (C3A), University of Trento,
38098 San Michele all'Adige (TN), Italy}
\date{September 3, 2026}

\begin{abstract}
This Supplement provides the numerical support for the three structural properties of the linearized redistribution operator $\J$ that the reduction of the main text relies on: (S1) $\J$ has a single invariant (water mass), a spectral gap, and radius-graded, localized modes; (S2) the gap closes exactly at the percolation threshold, where the hydraulic conductivity diverges; and (S3) a bimodal pore-size distribution splits the spectrum into two bands separated by a slow, sign-changing cross-band mode, the microscopic origin of the dual-permeability exchange coefficient. Every quantity is fixed by the pore-size distribution, the connectivity, and Hagen--Poiseuille; a single executed notebook reproduces all numbers, tables, and figures.
\end{abstract}

\maketitle

\noindent
This Supplement uses the notation of the main text; Table~\ref{tab:symbols}
collects the symbols and their units, and App.~\ref{app:glossary} defines
the terms of art (\emph{active band}, \emph{Fiedler value}, \emph{Galerkin
projection}, and so on) used below. Figures are drawn from the accompanying
notebook and are placed beside the tables they illustrate; every table is
referred to at the point in the text where its numbers are discussed.

\begin{table*}[t]
\centering
\renewcommand{\arraystretch}{1.25}
\begin{tabular}{@{}llll@{}}
\toprule
\textbf{symbol} & \textbf{meaning} & \textbf{units} & \textbf{where} \\
\midrule
$r,\,r'$ & pore-class radius & $\mu$m (SI: m) & S1--S3 \\
$f(r)$ & pore-size density (Kosugi/log-normal), $\int f\rd=1$ & m$^{-1}$ (per unit $r$) & S1--S3 \\
$r_m,\,\sigma$ & median radius and log-width of $f$ & $\mu$m,\, & S1--S3 \\
$w_C$ & connectivity width in $\ln r$, $C(r,r')=e^{-(\ln r-\ln r')^2/2w_C^2}$ &, (dimensionless) & S1, S3 \\
$g_\eq(r)$ & equilibrium filling fraction of class $r$ &, ($0\le g_\eq\le1$) & S1--S3 \\
$\psi_T$ & configurational ``temperature'' smoothing the packing step &, (in units of $\mu_w$) & S1 \\
$\theta,\,\phi$ & volumetric water content and porosity & m$^3$m$^{-3}$ & S1--S3 \\
$\theta_c$ & percolation threshold water content & m$^3$m$^{-3}$ & S2 \\
$\psi$ & matric head (suction), Young--Laplace of $r$ & m (or kPa) & S1--S3 \\
$\kappa(r)$ & Hagen--Poiseuille pore conductance, $\propto r^2$ & (rel.) & S1, S3 \\
$\kappa_s$ & harmonic-mean pair conductance & (rel.) & S1, S3 \\
$C(r,r')$ & connectivity kernel (dimensionless) &, & S1, S3 \\
$K(r,r')$ & symmetric exchange kernel, Eq.~\eqref{eq:Kdisc} & s$^{-1}$ & S1, S3 \\
$\Mr(r,r')$ & conductance prefactor of $K$ (geometric-mean gates $\times$ $\kappa_s C$) & s$^{-1}$ & S1, S3 \\
$m=g_\eq(1-g_\eq)$ & active-band weight; physical rate $=\A/m$ &, & S1, S3 \\
$\J$ & linearized redistribution operator, $\J=-\A+\mathcal{B}$ & s$^{-1}$ & S1--S3 \\
$\A(r)$ & local loss rate (diagonal of $\J$), row sum of $K$ & s$^{-1}$ & S1, S3 \\
$\mathcal{B}(r,r')$ & off-diagonal redistribution kernel of $\J$ & s$^{-1}$ & S1, S3 \\
$\varphi_n,\,\lambda_n$ & eigenmode and its relaxation rate, $\J\varphi_n=-\lambda_n\varphi_n$ &, ,\ s$^{-1}$ & S1--S3 \\
$\lambda_0=0$ & invariant (water mass) & s$^{-1}$ & S1 \\
$\lambda_1$ & spectral gap / Fiedler value & s$^{-1}$ & S1--S3 \\
$\langle\cdot,\cdot\rangle_f$ & mass inner product, $\int u\,v\,f\rd$ &, & S1, S3 \\
$\mathrm{PR}_n$ & participation ratio of mode $n$, Eq.~\eqref{eq:ipr} &, (no.\ of classes) & S1 \\
$L=D-W$ & combinatorial graph Laplacian &, & S2 \\
$\hat L=I-D^{-1/2}WD^{-1/2}$ & size-normalized graph Laplacian &, & S2 \\
$K_{\rm Stokes}$ & Stokes permeability of the spanning cluster & m$^2$ & S2 \\
$\lambda_{\rm 2band}$ & Galerkin two-domain exchange rate & s$^{-1}$ & S3 \\
$\Gamma_w$ & dual-permeability inter-domain exchange coefficient & s$^{-1}$ & S3 \\
\bottomrule
\end{tabular}
\caption{Table of symbols. Rates are reported in one convention
throughout (the mobility of the main text, Eq.~(M-explicit)), so the
$\lambda_n$ of S1 and S3 are directly comparable; only \emph{ratios} of
rates are independent of this choice. ``(rel.)'' marks quantities used only
in ratios, whose overall scale cancels. Radii are quoted in $\mu$m and
converted to suction $\psi$ by Young--Laplace, Eq.~\eqref{eq:YL}.}
\label{tab:symbols}
\end{table*}

\noindent\textbf{Scope and roadmap.}
The reduction in the main text rests on properties of the linearized
redistribution operator $\J$ that are argued analytically; this Supplement
tests them by direct computation, and exhibits the band structure behind
the dual-permeability result.  Every computation is parameter-free in the
physical sense, each quantity is fixed by the pore-size distribution, the
connectivity, and Hagen--Poiseuille, and the scripts are self-contained.

\begin{table*}[t]
\centering
\renewcommand{\arraystretch}{1.35}
\begin{tabular}{@{}lll@{}}
\toprule
\textbf{Section} & \textbf{main-text claim tested} & \textbf{result} \\
\midrule
S1 & $\J$ has one exact invariant, a gap, graded modes
   & confirmed; radius label to $1\%$ \\
S2 & the gap closes at percolation, where $K$ diverges
   & gap and permeability vanish together at $\theta_c$ \\
S3 & a bimodal PSD splits the spectrum into two bands
   & gap $\sim\!10\times$ scale separation; slow cross-band mode \\
\bottomrule
\end{tabular}
\caption{Scope and roadmap: the three main-text claims tested here and the
outcome of each.}
\label{tab:scope}
\end{table*}

\noindent
As Table~\ref{tab:scope} summarizes, s1 diagonalizes the operator of the paper itself; S2 builds a
three-dimensional pore network and measures the gap and the Stokes
permeability as the medium drains; S3 shows that the bimodal spectrum
underlying the multi-band section of the main text, the origin of the
dual- and multiple-permeability models, is a genuine gap, with the slow
mode the inter-band exchange whose coefficient the main text computes from
connectivity.

\paragraph{Which soils these examples represent.}
Every example below is specified by a pore-size distribution, so it is
worth saying at the outset which real soils those distributions describe.
The translation is Young--Laplace: a pore of radius $r$ drains at the
matric head
\begin{equation}
\psi(r)=-\frac{2\gamma\cos\theta_c}{\rho_w\mathrm{g}\,r}
\simeq-\frac{1.5\times10^{-5}\,\mathrm{m^2}}{r},
\label{eq:YL}
\end{equation}
so a radius is also a suction, and a pore-size distribution is a retention
curve.  Table~\ref{tab:soils} makes the correspondence explicit.

\begin{table*}[t]
\centering
\renewcommand{\arraystretch}{1.25}
\begin{tabular}{@{}llll@{}}
\toprule
 & \textbf{pore-size distribution} & \textbf{drains at} &
   \textbf{soil this represents} \\
\midrule
S1, S2 & unimodal, $r_m=8\,\mu$m, $\sigma=0.7$--$0.75$
  & $\psi(r_m)\simeq-1.9$\,m ($-18$\,kPa)
  & a well-graded \textbf{loam / silt loam}: \\
 & & & one population of pores, no structure \\[3pt]
S3 & bimodal: matrix $4\,\mu$m ($70\%$)
  & $-3.7$\,m ($-36$\,kPa)
  & a \textbf{structured / aggregated soil}: \\
 & \phantom{bimodal:} macropores $40\,\mu$m ($30\%$)
  & $-0.37$\,m ($-3.6$\,kPa)
  & aggregate matrix plus biopores, \\
 & \phantom{bimodal:} split at $r_\times\simeq14\,\mu$m
  & $-1.1$\,m ($-10$\,kPa)
  & cracks or root channels \\
\bottomrule
\end{tabular}
\caption{The three computations in soil terms.  Note that the S3 split
radius falls at $\psi\simeq-1$\,m $\simeq-10$\,kPa, which is the
conventional field boundary between ``macropores'' that drain quickly and
matrix pores that retain water, so the spectral band split the operator
produces coincides with the split soil physicists already draw by hand.}
\label{tab:soils}
\end{table*}

\noindent
S1 and S2 therefore ask: \emph{in an ordinary unstructured loam, does the
redistribution operator behave as the reduction assumes?}  S3 asks:
\emph{in a structured soil, does the operator itself produce the two
domains that dual-permeability models postulate?}

\section{The spectrum of the pore-class operator}
\label{sec:A}

\subsection{Construction}
We discretize the pore axis into $N=160$ logarithmically spaced classes
$r\in[0.5,80]\,\mu$m, with a Kosugi (log-normal) pore-size density $f(r)$
of median $r_m=8\,\mu$m and log-width $\sigma=0.7$, normalized to
$\sum f=1$. The equilibrium filling is the Fermi--Dirac smoothing of the
packing step, $g_\eq(r)=\{1+\exp[(\mu_w(r)-\lambda)/\psi_T]\}^{-1}$ with
$\mu_w=-1/r$ (in $\mu\mathrm{m}^{-1}$) and the level $\lambda$ set by
$\theta=\tfrac12\phi$, which places the frontier at $r^*=8.0\,\mu$m.
The configurational temperature is specified through the
dimensionless frontier band width $w=\psi_T/|\mu_w(r^*)|$, the control
parameter of Prop.~2 of the main text; the headline case is $w=0.10$,
i.e.\ $\psi_T=0.0125$ in these units.  (A wider band, e.g.\ $\psi_T=0.05$,
i.e.\ $w=0.36$, lies beyond the resolved-band regime in which the main
text's statements hold, as the sweep below shows.)  This smoothing is what gives the
operator a resolved spectrum: on the sharp step the conductances vanish
identically.

The operator is built exactly as in the main text. With the
Hagen--Poiseuille conductance $\kappa(r)\propto r^2$, the harmonic-mean
pair conductance $\kappa_s=2\kappa\kappa'/(\kappa+\kappa')$, a
connectivity $C(r,r')=\exp[-(\ln r-\ln r')^2/2w_C^2]$ with $w_C=0.9$, and
the occupancy gates,
\begin{widetext}
\begin{equation}
K(r,r')=\Mr(r,r';g_\eq)\,f(r'),\qquad
\Mr=\kappa_s C\,\sqrt{g_\eq(1{-}g_\eq)\,g_\eq'(1{-}g_\eq')} ,
\label{eq:Kdisc}
\end{equation}
\begin{equation}
\Lop[u](r)=\int K(r,r')\big[u(r')-u(r)\big]\rd' ,
\qquad
\A(r)=\int K(r,r')\rd' ,
\label{eq:Jdisc}
\end{equation}
\end{widetext}
and, following the main text, the linearized redistribution operator is
$\J=\Lop\circ(\psi_T/m)$ with $m=g_\eq(1-g_\eq)$: the Laplacian acts on
the potential perturbation $\delta\mu_w=\psi_T h/m$, so the relaxation
rate of class $r$ is $\A(r)/m(r)$ and not $\A(r)$.
\emph{All spectral quantities reported below are those of $\J$}, the
operator that actually governs the relaxation, in the inner product in
which it is self-adjoint, $\langle u,v\rangle_{f/m}=\sum(f/m)\,uv$.
Symmetrizing with $(f/m)^{1/2}$ gives a matrix whose off-diagonal entries
are $\psi_T\kappa_s C\sqrt{ff'}$ and whose diagonal is $-\psi_T\A/m$:
with the geometric-mean gate the occupancy factors cancel out of the
coupling entirely and survive only on the diagonal, as $\sqrt{m'/m}$.
This is the algebraic reason the slow end of the spectrum is pinned to the
frontier, where $m$ is largest, rather than to wherever the band is
truncated.  The eigenvectors $\varphi_n$ are orthonormal in
$\langle\cdot,\cdot\rangle_{f/m}$; $\ker\J=\mathrm{span}\{m\}
=\mathrm{span}\{\partial_\theta g_\eq\}$ and the constants span the
cokernel.  All rates $\lambda_n$ carry the factor $\psi_T$.
Two consequences of the gate choice in $\Mr$ are worth stating.
First, the physical decay rate of class $r$ is $\A(r)/m(r)$, not the
conductance row sum $\A(r)$.  Second, the geometric-mean gates are what put
the slow end of the spectrum at the frontier: with the double product
$m\,m'$ the slowest modes sit at the truncation
edge of the active band, $\lambda_1$ depends on the truncation tolerance,
and the gap-width exponent below changes sign.
The gates vanish where a class is full or empty, so the operator acts
only on the \emph{active band}, here $114$ of the $160$ classes ($m>10^{-12}$; the count is a tolerance, $\lambda_1$ is not: with $m>10^{-6}$ the band has $100$ classes and $\lambda_1$ changes by $2\times10^{-5}$ in relative terms). The
eigenproblem is solved in the mass measure,
$\langle\varphi_m,\varphi_n\rangle_f=\delta_{mn}$, by symmetrizing with
$f^{1/2}$.

\paragraph{What the experiment is.}
We are, in effect, building the soil's redistribution network on paper and
then asking how it relaxes.  Each pore class is a node; two classes are
joined if water can pass between them; the strength of the joint is the
conductance $K(r,r')$ of Eq.~\eqref{eq:Kdisc}.  We then take a half-wet
loam ($\theta=\phi/2$), disturb its pore filling slightly away from
equilibrium, and ask: \emph{what shapes of disturbance decay without
changing shape, and how fast?}  Those are the eigenvectors $\varphi_n$ and
the rates $\lambda_n$.  The whole reduction in the main text depends on
three answers, that exactly one shape never decays (water mass), that all
the others decay at a rate bounded away from zero (the gap), and that the
modes are graded by pore radius (so the spectrum has the structure the
reduction exploits), so the point of this section is to check those three
answers by direct computation rather than argument.

The Jupyter notebooks that produce every number, table, and figure in this Supplement are available at \url{https://osf.io/bz9u4/} and will be assigned a DOI upon acceptance.

\subsection{Results}
\paragraph{The invariant is exact, and it is $m$, not $1$.}
The computed $|\lambda_0|=3\times10^{-14}$, and the null vector
$\varphi_0$ is parallel to $m=g_\eq(1-g_\eq)$ to $\cos\angle(\varphi_0,m)
=1.000000$ in $\langle\cdot,\cdot\rangle_{f/m}$: the kernel is the tangent
to the equilibrium family, as Prop.~1 states.  The constants are the
\emph{left} null vector, $\|\mathbf 1^{\!\top}\!f\J\|/\|\J\|=1\times10^{-18}$: water
mass is a numerically exact invariant, since $\A(r)$ is the row sum of $K$.
The two null directions differ because $g_\eq$ is a step; in Boltzmann's
case they coincide.

\paragraph{Two different statements about $\A(r)$, and how they relate.}
Two formulas below involve $\A(r)$ and they are easy to confuse, so we
separate them once and for all.

\emph{(i) The sum rule is exact and always carries the sum.}  It says that
the local exchange conductance at radius $r$ is recovered by adding up
\emph{every} mode's contribution at that radius,
\begin{equation}
\psi_T\,\A(r)=f(r)\sum_{n\ge1}\lambda_n\,\big|\varphi_n(r)\big|^2 ,
\label{eq:sumrule}
\end{equation}
with $\{\lambda_n,\varphi_n\}$ the eigenpairs of $\J$ orthonormal in
$\langle\cdot,\cdot\rangle_{f/m}$.  It is the diagonal of the spectral
decomposition $\J=-\sum_n\lambda_n\varphi_n\langle\varphi_n,\cdot\rangle_{f/m}$
evaluated at $r$, where $\J_{rr}=-\psi_T\A(r)/m(r)$ and the weight $f/m$
cancels the $1/m$: the factor $m$ drops out and the identity takes the same
form as for a mass-orthonormal basis.  We verify it below to $10^{-13}$.
The sum never disappears.

\emph{(ii) The radius label is an approximation, and it is where the sum
collapses.} Because the modes here are localized on two or three adjacent
classes ($\mathrm{PR}_n\simeq2$--$3$, Table~\ref{tab:radius}), the density
$|\varphi_n(r)|^2$ is sharply peaked at a resonant radius $r_n$, so at
$r=r_n$ one term of the sum dominates and
\begin{equation}
\lambda_n\simeq\psi_T\,\A(r_n)/m(r_n),
\label{eq:radlabel}
\end{equation}
the \emph{local rate}, not the conductance $\A$ itself.  This is a
statement about one \emph{particular} pair $(\lambda_n,r_n)$, not about all
modes at once, and it is a consequence of localization, not an identity.
Table~\ref{tab:radius} tests how well it holds ($2$--$6\%$ here);
Sec.~\ref{sec:C} shows a mode for which it fails, because that mode is not
localized by radius at all.

\paragraph{The sum rule requires the density.}
Table~\ref{tab:sum} tests Eq.~\eqref{eq:sumrule} numerically, and also
tests what happens if the factor $f(r)$ is dropped.  We compute the
left-hand side from the operator (the row sum of $K$), the right-hand side
from the eigenpairs, and report the median ratio over the $114$
active classes.

\begin{table*}[t]
\centering
\renewcommand{\arraystretch}{1.35}
\begin{tabular}{@{}lcl@{}}
\toprule
\textbf{quantity computed} & \textbf{median ratio} & \textbf{verdict} \\
\midrule
$\dfrac{f(r)\sum_{n\ge1}\lambda_n|\varphi_n(r)|^2}{\psi_T\A(r)}$
  & $\mathbf{1.0000}$ (max.\ deviation $1\times10^{-13}$) & identity confirmed to machine precision \\[10pt]
$\dfrac{\sum_{n\ge1}\lambda_n|\varphi_n(r)|^2}{\psi_T\A(r)}$ \ (no $f$)
  & $128$ & off by two orders: the density is essential \\
\bottomrule
\end{tabular}
\caption{Numerical test of the spectral sum rule~\eqref{eq:sumrule}.  Each
entry is (right-hand side from the spectrum) divided by (left-hand side from
the operator), median over the active band.  The second row shows what the
identity looks like if one forgets the weight in which the eigenvectors are
normalized: the mismatch is a factor $\sim10^2$ (the density $f$ spans
a factor $224$ across the band), i.e.\ $f(r)$ carries real
information about how many pores of each size there are.}
\label{tab:sum}
\end{table*}

\paragraph{Two pictures of what the modes could look like.}
Before reporting what the modes \emph{are}, it helps to name the two
possibilities we are choosing between, because they are the two things a
relaxing system usually does.

\emph{Picture 1: harmonics (like a guitar string).}  If the active band
behaved as a uniform stretched medium, the modes would be standing waves
across it: the slowest would be a single arch spanning the whole band, the
next would have one node in the middle, the third two nodes, and so on.
For that geometry the rates grow with the square of the mode number,
$\lambda_n\propto n^2$, exactly as the overtones of a vibrating string rise
in pitch.  Each mode would then be spread over \emph{all} the pore classes
at once.

\emph{Picture 2: localized modes (like separate tuning forks).}  If instead
the local rate $\psi_T\A(r)/m(r)$ varies strongly from one radius to
the next, then each radius relaxes at its own rate, largely independently of
the others.  The modes are then localized: mode $n$ lives at one particular
radius $r_n$, and its rate is just the local rate there.

These are distinguishable by measurement, and we measure them.

\paragraph{How localization is measured: the participation ratio.}
To ask ``how many pore classes does this mode actually occupy?'' we use a
standard index.  Normalize the mode into a probability over classes,
$p_n(r)=|\varphi_n(r)|^2f(r)/m(r)$ with $\sum_r p_n=1$, and define
\begin{equation}
\mathrm{PR}_n=\Big[\sum_r p_n(r)^2\Big]^{-1}.
\label{eq:ipr}
\end{equation}
If the mode sits entirely on one class, $\mathrm{PR}=1$; if it is spread
evenly over all $N_{\rm band}$ active classes, $\mathrm{PR}=N_{\rm band}$.
It is the effective number of classes participating in the mode, the
exponential of the R\'enyi-2 entropy of $p_n$.  Picture~1 predicts
$\mathrm{PR}\approx N_{\rm band}=114$; Picture~2 predicts
$\mathrm{PR}\approx1$.

\paragraph{The modes carry a radius, and they climb out from the frontier.}
Picture~2 wins.  The local rate $\psi_T\A/m$ spans $4.2$ decades
across the active band, with its \emph{minimum} at $r=7.5\,\mu$m, next
to the frontier $r^*=8.0\,\mu$m, and its maxima in the two wings where
$m\to0$.  So $\J$ is strongly graded, but graded \emph{outward from $r^*$},
not from small to large radius.  Table~\ref{tab:radius} shows the
consequence: every low mode occupies two or three adjacent classes
($\mathrm{PR}=1.8$--$3.2$ out of $114$), the resonant radii alternate on
either side of the frontier and move away from it as $n$ grows
($7.5,\,7.3,\,7.1,\,8.0,\,6.9,\,8.3,\dots\,\mu$m), and $\lambda_n$ equals
the local rate at $r_n$ to within $2$--$6\%$.  Ordering the modes by rate
is therefore the same thing as ordering the pore classes by distance from
the frontier: slow modes at $r^*$, fast modes in the wings.  A picture with
``slow modes on small pores'' would be an artefact of the
double-product gate, which suppressed the wings by a full factor $m$ instead
of $\sqrt m$ and let the $\kappa\propto r^2$ trend win.

The radius label is a statement about the \emph{diagonal} of $\J$ alone;
the off-diagonal kernel couples distinct classes and is what makes the
$\mathrm{PR}$ slightly larger than one.

This is worth stating in hydrological terms.  The soil does not relax as a
single body with one drainage time; it relaxes radius by radius, and the
\emph{slowest} classes are those \emph{at the drainage front}, where the
driving $|\Phi(r,r^*)|$ vanishes: it is the frontier, not the finest pores,
that sets the rate at which the REV re-sorts its water.  This is why the
conductivity of the main text is \emph{regularized by the gap} $\lambda_1$,
and why the gap is tied to the frontier band width below.
Figure~\ref{fig:A} shows the modes and the collapse onto the local rate.

\paragraph{The modes are \emph{not} standing waves.}
The last two columns of Table~\ref{tab:radius} test Picture~1 directly.
The observed ratios $\lambda_n/\lambda_1=1,\,1.009,\,1.029,\,1.048,\dots$
are nowhere near $n^2=1,4,9,16,\dots$; the spectrum above $\lambda_1$ is a
quasi-continuum ($\lambda_2/\lambda_1=1.009$), and the low modes have two
or three sign changes, not $n$ nodes.  The graded diagonal, not the band
geometry, sets the spectrum.

\begin{table*}[t]
\centering
{\begin{tabular}{@{}cccccccc@{}}
\toprule
$n$ & $\lambda_n$ & $r_n\,[\mu\mathrm{m}]$ & $\psi_T\A(r_n)/m(r_n)$ &
ratio & $\mathrm{PR}_n$ & $\lambda_n/\lambda_1$ & $n^2$ \\
\midrule
1 & $0.2558$ & 7.54 & $0.2421$ & 1.056 & 1.8 & 1.000 & 1 \\
2 & $0.2582$ & 7.30 & $0.2450$ & 1.054 & 2.4 & 1.009 & 4 \\
3 & $0.2632$ & 7.07 & $0.2547$ & 1.033 & 3.2 & 1.029 & 9 \\
4 & $0.2682$ & 8.04 & $0.2568$ & 1.044 & 2.6 & 1.048 & 16 \\
5 & $0.2782$ & 6.85 & $0.2717$ & 1.024 & 3.0 & 1.087 & 25 \\
6 & $0.2851$ & 8.30 & $0.2745$ & 1.039 & 2.5 & 1.115 & 36 \\
7 & $0.3011$ & 6.63 & $0.2963$ & 1.016 & 2.7 & 1.177 & 49 \\
8 & $0.3093$ & 8.56 & $0.2995$ & 1.033 & 2.5 & 1.209 & 64 \\
\bottomrule
\end{tabular}}
\caption{Radius labeling of the eigenmodes ($w=0.10$).  The resonant
radii sit on either side of $r^*=8.0\,\mu$m and move outward with $n$;
$\lambda_n$ equals the local rate to $2$--$6\%$; $\lambda_n\propto
n^2$ fails.}
\label{tab:radius}
\end{table*}
\begin{figure*}[t]
\centering
\includegraphics[width=\textwidth]{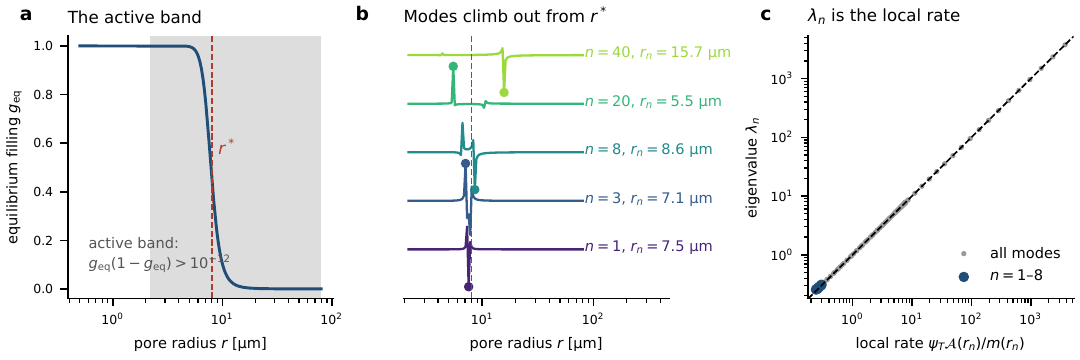}
\caption{The spectrum of $\J$ for the loam of Table~\ref{tab:soils}
at band width $w=0.10$.
\textbf{(a)} The smoothed equilibrium $g_\eq(r)$ and the active band
(shaded): the partially filled classes are the only ones that exchange
water, so they are the only ones $\J$ acts on.
\textbf{(b)} Five representative eigenmodes plotted \emph{offset
vertically}.  Each is a spike localized at its own resonant radius $r_n$
(dot); the low modes sit at the frontier and higher modes move outward
into the two wings, where near-degenerate pairs straddle $r^*$.
\textbf{(c)} $\lambda_n$ against the local rate $\psi_T\A(r_n)/m(r_n)$
at the mode's own radius, dashed line equality, for all modes (grey) and
$n=1$--$8$ (blue).}
\label{fig:A}
\end{figure*}

\subsection*{The gap and the frontier band width}
The conductances~\eqref{eq:Kdisc} are supported on the frontier band, and
the band width $w=\psi_T/|\mu_w(r^*)|$ is the control parameter of Prop.~2.
Sweeping it over a factor $50$ on a grid fine enough to resolve the
narrowest band ($N=640$) gives Table~\ref{tab:gapwidth}.  Two regimes
appear.  For $w\lesssim0.2$ the slowest mode sits at the frontier
($r_1/r^*=0.8$--$1.0$), is localized ($\mathrm{PR}_1\simeq2$), and the gap
follows
\begin{equation}
\lambda_1\propto w^{\,\gamma},\qquad
\gamma=1.93\ (w\le0.1),\quad \gamma=1.79\ (w\le0.2),
\label{eq:gapwidth_sm}
\end{equation}
so the sharp step $w\to0$ has \emph{no} gap and the reduction is
asymptotic only for a resolved band.  For $w\gtrsim0.3$ the frontier no
longer controls: the slowest mode migrates to small radii
($r_1/r^*=0.6\to0.2$), delocalizes ($\mathrm{PR}_1\simeq35$--$45$), and
$\lambda_1$ saturates and turns over, because the $\kappa\propto r^2$
suppression of the fine pores overtakes the vanishing of the gates.  The
main text's statements about the spectrum are statements about the first
regime, and $w\simeq0.2$--$0.3$ is, for this loam, the band width beyond
which they fail; the earlier headline $\psi_T=0.05$ ($w=0.36$) sat just past
it.  This is the numerical content of the band-width condition of Prop.~2:
the width is a physical property of the medium (configurational
fluctuations, gravitational smearing over the REV, within-class geometric
dispersion), it is measurable, and it must be small.  Along a saturation
sweep at fixed $w=0.10$ the gap moves smoothly with $\theta$
(Table~\ref{tab:gapwidth}, lower block) and $\lambda_2/\lambda_1$ stays at
$1.01$: a unimodal medium has one invariant and one gap, and no second
scale separation on which a two-mode model could rest.  That second
separation requires a bimodal density, Sec.~\ref{sec:C}.

\begin{table*}[t]
\centering
\small
{\begin{tabular}{@{}lccccccccccc@{}}
\toprule
$w$ & 0.02 & 0.03 & 0.05 & 0.07 & 0.10 & 0.15 & 0.20 & 0.30 & 0.36 & 0.50 & 1.00 \\
\midrule
$\lambda_1$ & 0.0115 & 0.0258 & 0.0709 & 0.135 & 0.255 & 0.475 & 0.671 & 0.892 & 0.929 & 0.88 & 0.402 \\
$r_1/r^*$ & 1.00 & 0.99 & 0.99 & 0.98 & 0.95 & 0.89 & 0.81 & 0.63 & 0.55 & 0.40 & 0.17 \\
$\mathrm{PR}_1$ & 2 & 2 & 2 & 2 & 2 & 2 & 2 & 35 & 45 & 45 & 3 \\
\bottomrule
\end{tabular}}\\[6pt]
{\begin{tabular}{@{}lccccccccc@{}}
\toprule
$\theta/\phi$ ($w=0.10$) & 0.1 & 0.2 & 0.3 & 0.4 & 0.5 & 0.6 & 0.7 & 0.8 & 0.9 \\
\midrule
$\lambda_1$ & 0.056 & 0.113 & 0.165 & 0.214 & 0.256 & 0.290 & 0.310 & 0.307 & 0.260 \\
$\lambda_2/\lambda_1$ & 1.012 & 1.011 & 1.009 & 1.010 & 1.009 & 1.011 & 1.009 & 1.009 & 1.009 \\
\bottomrule
\end{tabular}}
\caption{Upper block: gap, position and localization of the slowest
mode against the frontier band width $w$ at $\theta=\phi/2$ ($N=640$
classes).  The frontier-controlled regime with $\lambda_1\propto w^{1.8\text{--}1.9}$
ends between $w=0.2$ and $0.3$.  Lower block: saturation sweep at $w=0.10$
($N=160$).}
\label{tab:gapwidth}
\end{table*}

\subsection*{The gap ratio $\lambda_2/\lambda_1$}
The invariant-manifold theorems that certify the reduction bound the
provable order of \emph{nonlinearity} of an $M$-mode model by the gap
ratio $\lambda_M/\lambda_{M-1}$~\cite[\S3.4]{Roberts2025tma}.  For the one-mode (Richards) reduction the relevant ratio is
formally infinite (the retained eigenvalue is $\lambda_0=0$); for the
two-mode/two-band models of S3 the controlling number is
$\lambda_2(\theta)/\lambda_1(\theta)$.  For the unimodal loam of S1 this ratio is $1.01$ at every water
content of the saturation sweep (Table~\ref{tab:gapwidth}): the spectrum
above $\lambda_1$ is a quasi-continuum, so a two-mode model of a unimodal
medium has no theoretical support at any $\theta$, before the linear
gap-closure criterion of S2 is ever reached.  For the bimodal medium of S3
the same ratio is $71$ at the headline connectivity.

\section{Gap closure at the percolation threshold}
\noindent\emph{Note.}  Nothing in this section involves the mobility gate: the
network is drained sharply and the two quantities computed live on the
spanning cluster's own graph Laplacian and Stokes problem.  Its results are
therefore independent of the gate choice discussed in S1.
\label{sec:B}

\paragraph*{Operator gap versus grid gap.}
Any finite diagonalization has $\lambda_1>0$ by construction, so the
statement ``the gap closes at $\theta_c$'' is a statement about the
sequence of discretizations, not about any single run.  Two checks
separate the operator's gap from the grid's: (i)~away from the
threshold, $\lambda_1(\theta;N)$ must be stable under refinement of the
pore-class grid and under growth of the network; (ii)~near the
threshold, $\lambda_1$ collapses
with system size, consistent with a spectral density accumulating at
zero in the thermodynamic limit rather than with an isolated small
eigenvalue.  Check (i) licenses trusting the gap; check (ii) licenses
trusting its closure.  For check (i) the convergence is not fast: once
the active band is resolved (of order $10^2$ classes across it),
successive doublings of the grid move $\lambda_1$ by a few percent,
which is small against the decades over which $\lambda_1$ varies along
the saturation sweep but is not negligible, and it should be reported
as a convergence figure rather than asserted.  We therefore quote
$\lambda_1$ throughout with the grid stated, and treat differences
below $\sim5\%$ between neighbouring saturations as unresolved.

\subsection{Construction}
The radius labeling of Sec.~\ref{sec:A} is a statement about a graded
diagonal. The main text claims it fails at percolation, where the slowest
mode is captured by \emph{topology}: the gap eigenvector becomes the
Fiedler vector of the network, and $\lambda_1\to0$ as the water phase
fragments. This requires a real network.

\paragraph{The tool, and what it can and cannot do.}
The network calculations of this section are carried out with
\textsc{OpenPNM}~\cite{Gostick2016}, an open-source pore-network
modelling framework: it builds the lattice and its throats, performs
invasion-percolation drainage, identifies connected clusters, and solves
Stokes flow on the water-filled subnetwork.  It is the natural tool
here precisely because it computes the two objects we need, cluster
connectivity and Stokes permeability, by independent routes, neither of
which involves the kinetic equation of the main text.

That is also the limitation, and it should be stated plainly.  What
\textsc{OpenPNM} supplies is a \emph{static} pore-network realization of
the medium: a quasi-static drainage sequence and, at each step, a
steady-state flow problem.  It does not integrate the continuum kinetic
equation~(2) of the main text; no solver for the CKE exists yet.  So the
verifications reported here test the \emph{structural} claims, that the
operator built on a real pore network has a single invariant, a gap,
radius-graded modes, that the gap closes with the permeability, and that
a bimodal distribution produces a slow cross-band mode, and they do so
on a network whose geometry is independently generated.  They do not
test the reduction dynamically: a transient comparison between the CKE
and its Richards limit, which would measure the residual and the
emergence rate directly, must wait for a CKE solver.  Everything below
should be read in that light.

We build a cubic lattice of $24^3=13{,}824$ pores with log-normal radii
(median $8\,\mu$m, $\sigma=0.75$), bonds between nearest neighbours, and
throat radius equal to the smaller of the two pores. Drainage at suction
$\psi$ empties every pore with $r>r^*(\psi)$. On the water-filled bonds we
place Hagen--Poiseuille conductances $g\propto r_{\rm throat}^4$, extract
the cluster spanning the sample from top to bottom, and compute two things
on that cluster and nothing else:
\begin{itemize}\itemsep2pt
\item the \emph{spectral gap} $\lambda_1$ of its graph Laplacian, the
algebraic connectivity, i.e.\ the Fiedler value;
\item the \emph{Stokes permeability} $K_{\rm Stokes}$, from a unit pressure
drop across the sample.
\end{itemize}
These are independent objects, and that independence is the whole point:
one is a spectral property of how the network is \emph{connected}, computed
without ever solving a flow problem; the other is a flow problem, solved
without ever looking at a spectrum.  Nothing in the calculation forces them
to agree.

\paragraph{What the experiment is.}
This is a numerical drainage experiment on a synthetic loam.  We assemble a
cube of $24^3$ pores of random sizes, connect neighbours through throats
whose radius is the narrower of the two pores they join, and then drain it
step by step, exactly as a pressure plate would, emptying the widest pores
first.  At each step we look only at the water that is still connected
right through the sample (the \emph{spanning cluster}: water in isolated
pockets can hold water but cannot conduct it), and we ask two questions of
it.  First, a structural question: how well connected is what remains?
Second, a hydraulic question: if we impose a pressure drop across the
sample, how much water flows?  The claim under test is that these two
questions have the same answer, and in particular that they fail at the
same moment.

\subsection{Results}
Both vanish at the same water content. Below
$\theta_c\simeq3.5\times10^{-3}$ no spanning cluster exists: the water
phase is disconnected, and $\lambda_1=0$ and $K_{\rm Stokes}=0$
\emph{identically}, not small, but zero. Above $\theta_c$ both switch on
together and rise steeply; over the critical window $\lambda_1$ moves
through $1.6$ decades while $K_{\rm Stokes}$ moves through $3.7$
(Fig.~\ref{fig:B}).

This is the numerical content of the claim in the main text: the closure of
the spectral gap and the loss of macroscopic conductivity are not two
phenomena but one, occurring at a single threshold. The Chapman--Enskog
reduction requires $\|\J^{-1}\|=1/\lambda_1<\infty$; at $\theta_c$ that
norm diverges, and the expansion has no domain left. Therefore, the conductivity
formula $K=\phi\langle\kappa|\J^{-1}|S\rangle_f$ diverges there, while the
physical permeability vanishes, the divergence is precisely the signal
that the reduction, not the medium, has failed.

\paragraph{Why $\theta_c$ is so small.}
The threshold looks far below the textbook percolation value, but the two numbers measure different things. The simple-cubic site-percolation threshold $p_c\simeq0.31$ is a \emph{number fraction} of occupied sites; $\theta_c$ is a \emph{volumetric water content}. Because water is retained in the \emph{smallest} pores first and pore volume scales as $r^3$, the $\sim\!31\%$ of pores (by count) needed to span the sample hold only $\sim\!0.3\%$ of the pore volume. Measuring the number fraction at threshold in our network recovers exactly the percolation value ($0.316$), and $\theta_c$ is not a finite-size artefact: sweeping $16^3$, $24^3$ and $32^3$ leaves $\theta_c$ stable at $\sim\!3\text{--}4\times10^{-3}$ while the number fraction stays at $0.316$ (Fig.~\ref{fig:thetac}). The small $\theta_c$ is thus a genuine consequence of the broad, fine-tailed pore-size distribution, and it is why residual water contents in real soils are small volumetric numbers.

\begin{figure*}[t]
\centering
\includegraphics[width=\textwidth]{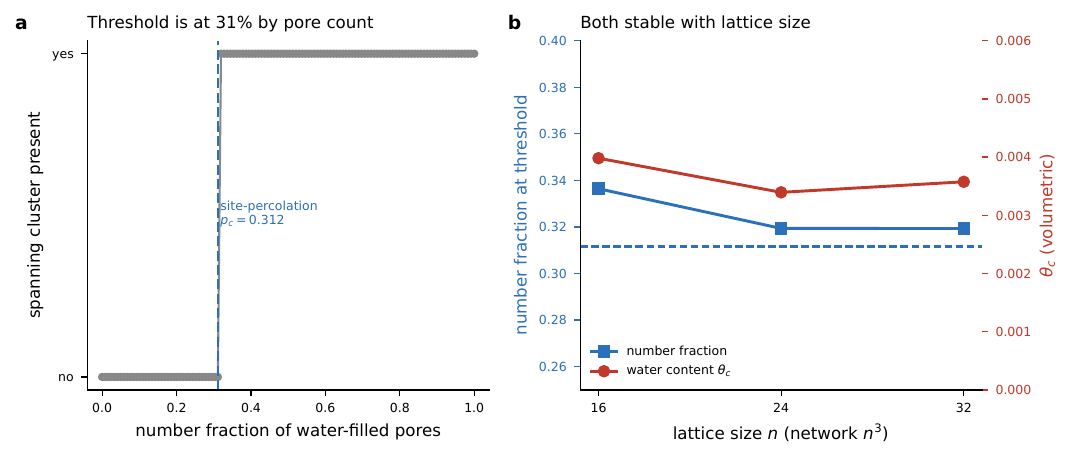}
\caption{Why the percolation threshold $\theta_c$ is a small volumetric number. \textbf{(a)} A spanning water cluster appears when the \emph{number fraction} of water-filled pores reaches $\simeq0.31$, the simple-cubic site-percolation value $p_c=0.312$ (dashed). \textbf{(b)} Across lattice sizes $16^3$, $24^3$, $32^3$ the number fraction at threshold ($\simeq0.32$) and the volumetric $\theta_c$ ($\simeq3\text{--}4\times10^{-3}$) are both stable, so the small $\theta_c$ is set by volume weighting ($V\propto r^3$, smallest pores fill first), not by finite sample size.}
\label{fig:thetac}
\end{figure*}

\paragraph{An independent cross-check.}
The result of this section does not depend on our own network code.  The
same drainage experiment, re-run in \textsc{OpenPNM} on an independently
generated cubic network, reproduces both halves of the claim
(Fig.~\ref{fig:openpnm}): the Stokes permeability of the spanning cluster
switches on at $\theta_c\simeq4\times10^{-3}$, and the fraction of pores
belonging to the spanning cluster rises from zero at the same water
content.  The two curves are computed by different code paths, one a flow
solve, the other a connectivity census, and they turn on together.

\begin{figure*}[t]
\centering
\includegraphics[width=\textwidth]{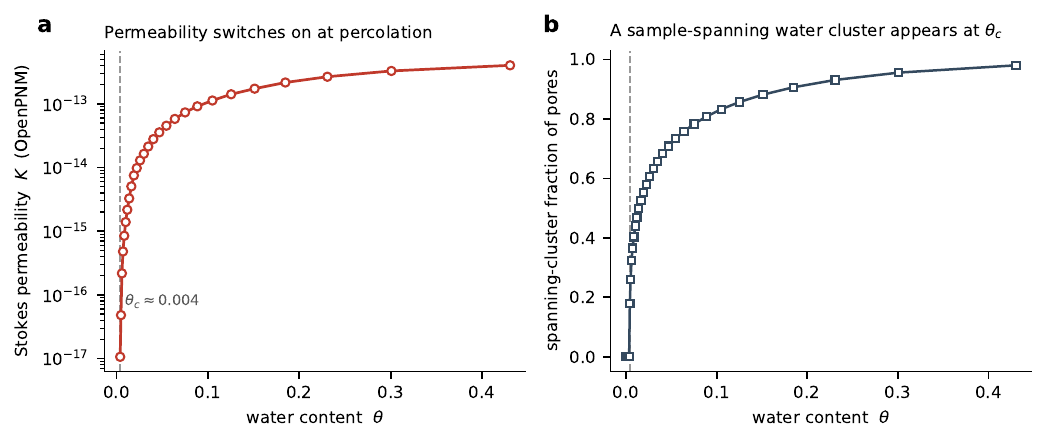}
\caption{Independent \textsc{OpenPNM} cross-check of the percolation
result of this section.  \textbf{(a)} Stokes permeability $K$ of the
spanning cluster against water content: it switches on at
$\theta_c\simeq4\times10^{-3}$ and rises by four decades over the
following few percent of $\theta$.  \textbf{(b)} Fraction of pores
belonging to the sample-spanning cluster, from the same runs: it leaves
zero at the same $\theta_c$.  Connectivity and conductivity are computed
independently, a census and a flow solve, and appear together, which is
the content of the claim tested in S2.}
\label{fig:openpnm}
\end{figure*}

\paragraph{What this means physically.}
The soil does not stop holding water at $\theta_c$, it stops
\emph{transmitting} it, because the water phase has broken into
disconnected pockets.  What the experiment shows is that the operator knows
this before any flow calculation is done: the same connectivity loss that
kills the permeability also removes the last decaying mode from the
spectrum, so the medium loses simultaneously its ability to conduct and its
ability to re-equilibrate.  This is the precise sense in which the theory
predicts its own limit of validity.  A macroscopic conductivity is a
meaningful object only while the pore network can re-sort water faster than
the forcing disturbs it; at $\theta_c$ that ceases to be true, and the
formula for $K$ signals the failure by diverging rather than by quietly
giving a wrong number.  For practice the message is that residual water
content is not a fitting artefact: it marks a genuine change of regime,
below which Richards' equation is not merely inaccurate but ill-founded.

\begin{figure*}[t]
\centering
\includegraphics[width=\textwidth]{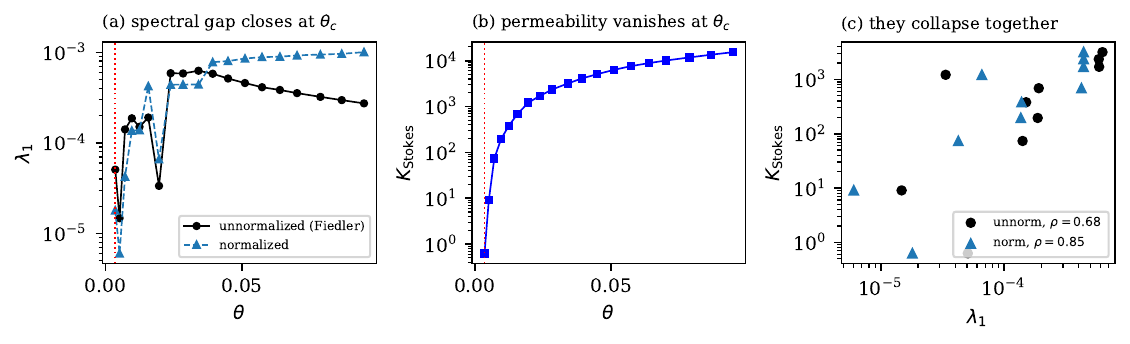}
\caption{Pore network, $24^3$ pores. (a) The spectral gap $\lambda_1$ of
the spanning water cluster against water content; the dashed line marks
$\theta_c$, below which no spanning cluster exists and $\lambda_1=0$.
(b) The Stokes permeability of the same cluster, vanishing at the same
$\theta_c$. (c) $K_{\rm Stokes}$ against $\lambda_1$.}
\label{fig:B}
\end{figure*}

\paragraph{What the Fiedler value is.}
The quantity $\lambda_1$ used above deserves a name and an intuition.  For
any network, the graph Laplacian has a smallest eigenvalue $\lambda_0=0$
(the constant mode: nothing happens if everything is at the same level).
The \emph{next} one, $\lambda_1$, is called the \emph{algebraic
connectivity} or \emph{Fiedler value}, after Fiedler~\cite{Fiedler1973},
and its eigenvector is the \emph{Fiedler vector}.  Two facts make it the
right diagnostic here.

First, $\lambda_1>0$ if and only if the network is connected in one piece;
it drops to exactly zero the moment the network splits.  It is therefore a
connectivity detector, not a size or a density.

Second, its eigenvector changes sign across the network's \emph{weakest
cut}: the positive part and the negative part are the two halves that are
most nearly separate, and $\lambda_1$ measures how thin the bridge between
them is.  In soil terms, the Fiedler vector points at the bottleneck, the
throat, or set of throats, whose loss would break the water phase in
two, and $\lambda_1$ measures how close that break is.  This is why it and
the permeability vanish together: they are both reporting on the same
bridge.  (See the glossary, App.~\ref{app:glossary}.)

\paragraph{An honest limitation.}
Away from the threshold the Fiedler value of a growing cluster is
increasingly dominated by its \emph{size} rather than by its connectivity
(for a well-connected cluster of $n$ nodes, $\lambda_1$ decreases with
$n$), so $\lambda_1$ turns over and declines at large $\theta$. The
correlation between $\log\lambda_1$ and $\log K_{\rm Stokes}$ is therefore
meaningful only on the critical branch, where it is $0.68$ over ten
sampled contents. The size bias is a property of the \emph{combinatorial}
Laplacian $L=D-W$; replacing it with the size-normalized Laplacian
$\hat L=I-D^{-1/2}WD^{-1/2}$, whose spectrum is insensitive to cluster
growth, removes the turnover and raises the correlation to $\rho\simeq0.9$
across the full drainage curve (mean over five random lattices; see the
accompanying notebook). Either way, what is unambiguous, and is all the
argument requires, is the joint vanishing at a single $\theta_c$: this is
exact, not statistical, since both quantities are identically zero in the
absence of a spanning cluster.

\section{Band structure of a bimodal medium}
\label{sec:C}
The dual- and multiple-permeability models of the main text are
obtained by projecting the kinetic equation onto the bands of a
gap-separated spectrum.  That construction presumes a bimodal pore-size
distribution actually produces such a gap.  Here we confirm it.

In soil terms this is a \emph{structured, aggregated soil}: a fine aggregate matrix (the $4\,\mu$m mode, draining near $-3.7$\,m) threaded by macropores, biopores, root channels, or shrinkage cracks (the $40\,\mu$m mode, draining near $-0.4$\,m), as set out in Table~\ref{tab:soils}. It is the class of soil for which dual-permeability models were invented.

\subsection{Construction}
We take a bimodal Kosugi density, a matrix mode at $r_m^{(1)}=4\,\mu$m
(weight $0.7$, $\sigma_1=0.35$) and a macropore mode at
$r_m^{(2)}=40\,\mu$m (weight $0.3$, $\sigma_2=0.30$), and build $\J$
exactly as in S1 (band width $w=0.10$, $\theta=\phi/2$), with a connectivity $C(r,r')$ local in $\ln r$
($w_C=0.30$ for the headline case), so that the two size populations are
only weakly bridged.
By Table~\ref{tab:soils} this is a structured soil: an aggregate matrix
draining around $-3.7$\,m, plus biopores or cracks draining around
$-0.4$\,m.

\paragraph{A note on units.}
The absolute rates $\lambda_n$ and $\A$ are reported in the
convention of Eq.~\eqref{eq:Kdisc}, in the same
convention as S1 (so that S1 and S3 rates are directly comparable). Only
\emph{ratios} of rates, the scale separation $\lambda_2/\lambda_1$ and the
Galerkin ratio $\lambda_{\rm 2band}/\lambda_1$ of Table~\ref{tab:galerkin},
which are the quantities the two-band reduction is tested against, are
independent of this choice.

\paragraph{What the experiment is.}
Dual-permeability models are written down by \emph{assuming} that a
structured soil behaves as two overlapping domains, a fast macropore
domain and a slow matrix domain, exchanging water at a rate $\Gamma_w$
that is then fitted to data.  Here we do not assume it.  We hand the
operator a single bimodal pore-size distribution, with no notion of
``domain'' anywhere in its construction, and ask whether the two domains
appear \emph{by themselves} in its spectrum, and whether the exchange
coefficient they need is something we can compute rather than fit.  Two
things must be true for the main text's claim to hold: (i) the spectrum
must split into two groups separated by a gap, with one slow mode that
moves water \emph{between} the groups; and (ii) the rate of that slow mode
must be the one the two-domain model would predict from the connectivity
between them.

\subsection{Results}
Diagonalizing on the active band ($130$ classes at $w=0.10$, split at
the spectral valley $r_\times\simeq14.3\,\mu$m into $75$ matrix and $55$
macropore classes; the frontier sits at $r^*=4.9\,\mu$m, inside the matrix
mode) shows the expected structure (Fig.~\ref{fig:C}).  The relaxation
spectrum separates into two clusters of intra-band modes, median local rates
$\psi_T\A/m\simeq0.26$ (matrix) and $4.5$ (macropore), above a single
slow mode
\begin{equation}
\lambda_1=0.00154,
\end{equation}
lying a factor $71$ below the slowest intra-band rate
$\lambda_2=0.11$.  The eigenvector $\varphi_1$ of this slow mode is
\emph{sign-changing across the split radius}, with exactly one node,
positive on the matrix and negative on the macropores, so it is precisely
the mode that exchanges water between the two bands while leaving each
internally equilibrated.  This is the spectral object the band projection
slaves: the gap sets the time-scale separation that makes the two-band
reduction valid, and the slow cross-band mode is the physical carrier of the
inter-band exchange $\Gamma_w$.  Requirement~(i) is therefore met.

\paragraph{The exchange coefficient is computed, not fitted.}
Requirement~(ii) is the sharper test, and it can be made without guessing a
closed form for $\Gamma_w$.  The two-domain model is nothing other than $\J$
restricted to a two-dimensional trial space: potential perturbations
that are constant on each band, $\delta\mu_w\in\mathrm{span}\{\mathbf
1_\mathsf{B},\mathbf 1_\mathsf{T}\}$, equivalently occupancy perturbations
$h\in\mathrm{span}\{m\mathbf 1_\mathsf{B},m\mathbf 1_\mathsf{T}\}$, each band
displaced along its own equilibrium family.  We form the \emph{Galerkin
projection} of $\J$ onto that space in the $\langle\cdot,\cdot\rangle_{f/m}$
inner product.  (In words: we allow the soil only two degrees of freedom,
``how wet is the matrix'' and ``how wet are the macropores'', and read off
the rate at which the operator moves water between them.)  Its single
nonzero eigenvalue $\lambda_{\rm 2band}$ \emph{is} the exchange rate the
reduced model predicts from connectivity alone; being a Rayleigh
quotient over a restricted space it is an upper bound on the true slow rate,
$\lambda_{\rm 2band}\ge\lambda_1$.

If the two-domain picture is right, $\lambda_{\rm 2band}$ must track the
slow cross-band eigenvalue $\lambda_1$ of the \emph{full} operator.  Varying
the connectivity width $w_C$ to tune the separation gives
Table~\ref{tab:galerkin} and Fig.~\ref{fig:C2}.

\begin{table*}[t]
\centering
\renewcommand{\arraystretch}{1.3}
{\begin{tabular}{@{}ccccc@{}}
\toprule
$w_C$ & \textbf{scale separation} $\lambda_2/\lambda_1$ & $\bm{\lambda_1}$ \textbf{(full operator)} &
$\bm{\lambda_{\rm 2band}}$ \textbf{(two-domain)} & \textbf{ratio} \\
\midrule
0.70 & 1.0 & $0.15$ & $0.435$ & 2.91 \\
0.60 & 1.0 & $0.147$ & $0.189$ & 1.28 \\
0.50 & 2.7 & $0.0529$ & $0.0614$ & 1.16 \\
0.42 & 8.8 & $0.016$ & $0.019$ & 1.19 \\
0.36 & 24.9 & $0.0054$ & $0.00658$ & 1.22 \\
0.30 & 71.3 & $0.00154$ & $0.00192$ & 1.25 \\
0.26 & 101.0 & $0.000606$ & $0.000775$ & 1.28 \\
\bottomrule
\end{tabular}}
\caption{The two-domain exchange rate against the true slow rate as the
bands are made progressively more separate (narrower connectivity kernel
$w_C$, larger scale separation).  Once the bands are resolved
($\lambda_2/\lambda_1\gtrsim3$) the Galerkin rate tracks $\lambda_1$ over
three decades but overshoots it by a stable factor $1.15$--$1.3$; for
$w_C\ge0.6$ there is no gap and the two-domain picture does not apply.}
\label{tab:galerkin}
\end{table*}

Two things are true at once.  The Galerkin rate follows $\lambda_1$ over
three decades as the cross-connectivity is weakened, so the exchange
coefficient is indeed a functional of the inter-band connectivity, computed
and not fitted.  But the ratio does \emph{not} fall to unity as the gap
widens: it settles at $1.15$--$1.3$.  The reason is visible in
Fig.~\ref{fig:C}(c).  The true slow mode is not the balanced band indicator;
it has internal structure near the valley, because the classes adjacent to
$r_\times$ are coupled to the interior of their own band by the same local
kernel $C(r,r')$ that bridges the bands, so when $w_C$ shrinks the
within-band equilibration next to the bottleneck weakens at the same rate as
the bottleneck itself.  The band-indicator ansatz forces those classes to
the band's mean potential and therefore overestimates the exchange
conductance, by exactly the serial-path (Mualem-type) correction that the
main text derives for the single-band conductivity.  The overshoot is
therefore not error but the leading profile correction to the two-domain
truncation, and it is $O(1)$ rather than vanishing: a dual-permeability
exchange coefficient computed from connectivity alone is right to within a
factor $\sim1.2$ for a structured soil of this kind, and the correction is
itself computable from the same operator.

\paragraph{What this means for practice.}
Gerke--van~Genuchten and its relatives are usually justified by their fit.
What this computation shows is that they have a derivation: give the
operator a bimodal pore-size distribution and the two domains, the gap
between them, and the transfer term all appear without being put in by
hand.  The exchange coefficient, in particular, ceases to be a free
parameter, it is a functional of the cross-band connectivity, and the
computation above is the check that the functional is the right one to
within the $O(1)$ profile factor just described.

\begin{figure*}[t]
\centering
\includegraphics[width=\textwidth]{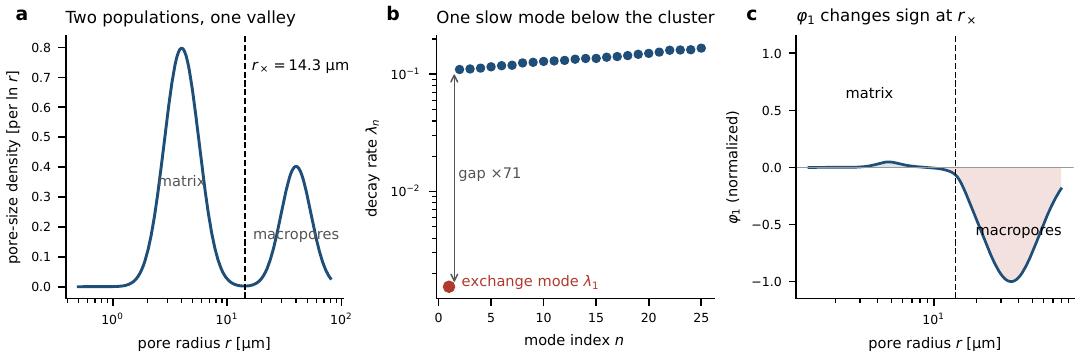}
\caption{The structured soil of Table~\ref{tab:soils}.
\textbf{(a)} The bimodal pore-size distribution, with the split radius
$r_\times\simeq14.3\,\mu$m (dashed), the boundary the operator will
reproduce on its own.  \textbf{(b)} The low-lying spectrum: two clusters of
fast intra-band modes sit well above a single slow mode $\lambda_1$, the
gap that makes a two-domain description meaningful.  \textbf{(c)} The slow
eigenvector $\varphi_1$: it changes sign across $r_\times$, positive on one
band and negative on the other.  A sign-changing mode is one that moves
water \emph{from} one side \emph{to} the other while leaving each side
internally equilibrated, which is precisely what an inter-domain exchange
term does.}
\label{fig:C}
\end{figure*}

\begin{figure*}[t]
\centering
\includegraphics[width=0.72\textwidth]{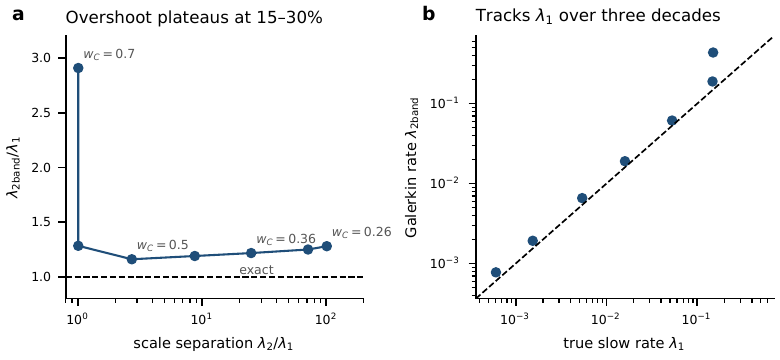}
\caption{The two-domain reduction tracks the full operator, with an $O(1)$ profile correction.
\textbf{(a)} The ratio $\lambda_{\rm 2band}/\lambda_1$ against the
intra/inter-band scale separation: once a gap exists the reduced
model's exchange rate overshoots the true slow rate by a stable
$15$--$30\%$ (dashed line, ratio $=1$).  \textbf{(b)} The same data as a direct
comparison of the two rates across connectivity widths; the dashed line is
equality.  Together these say that $\Gamma_w$ is the Galerkin image of the
microscopic connectivity, up to a computable serial-path correction
that does not shrink with the gap.}
\label{fig:C2}
\end{figure*}

\subsection{Transitivity check: the cross-band sum rule}
\label{sec:sumrule}
The two-band reduction is \emph{transitive}: adiabatic elimination of
the exchange mode must return the single Richards equation with the
full one-band conductivity.  In the eigenbasis this is an exact
identity.  Write
\begin{equation}
K=\phi\sum_{n\ge1}
\frac{\langle\kappa,\varphi_n\rangle_f\langle\varphi_n,S_0C_w\rangle_f}
{\lambda_n}
\end{equation}
for the one-band Green--Kubo sum, and
\begin{align}
K^{(1)}&:=\phi\sum_{n\ge2}
\frac{\langle\kappa,\varphi_n\rangle_f\langle\varphi_n,S_0C_w\rangle_f}
{\lambda_n}\nonumber\\
&=K-\phi\,\frac{\langle\kappa,\varphi_1\rangle_f
\langle\varphi_1,S_0C_w\rangle_f}{\lambda_1}
\end{align}
for the same sum with the $n=1$ term removed.  The superscript counts
the modes carried \emph{explicitly}: $K^{(1)}$ is the conductivity left
to the slaved modes once the slowest one, $\varphi_1$, has been promoted
to an independent variable of the two-band model, so it is smaller than
$K$ by exactly that mode's contribution.  With this notation the
identity reads
\begin{equation}
K^{(1)}+\phi\,\frac{\langle\kappa,\varphi_1\rangle_f\,
\langle\varphi_1,S_0C_w\rangle_f}{\lambda_1}\;\equiv\;K :
\end{equation}
eliminating the retained mode restores its spectral term.  Every
factor is already produced by the S1/S3 diagonalization, so the check
is pure post-processing: (i)~assemble $K$ and $K^{(1)}$ from the
eigenpairs; (ii)~verify the identity to machine precision (it is
exact, so any discrepancy measures numerical error, not physics);
(iii)~as a physical corollary, report the $n=1$ share of $K$ across
the saturation sweep, the fraction of the conductivity that the
two-band model promotes from slaved to explicit.  For the bimodal
medium of this section the same check, performed in the band basis
via the Galerkin $2\times2$ projection of Table~\ref{tab:galerkin},
quantifies how closely the band chart approximates the modal chart
(they agree exactly only in the limit in which $\varphi_1$ is the
balanced band indicator).  The notebook implements all three steps in
Part~C.

\section*{Reproducibility}
All three computations are self-contained: Sec.~\ref{sec:A} uses only the
pore-class operator~\eqref{eq:Kdisc}--\eqref{eq:Jdisc} and a dense
symmetric eigensolver; Sec.~\ref{sec:B} uses a cubic lattice, a
connected-components search and a sparse Laplacian solve;
Sec.~\ref{sec:C} repeats Sec.~\ref{sec:A} on a bimodal density and adds a
$2\times2$ generalized eigenproblem for the Galerkin projection;
the gap-ratio sweep of S1 and the sum-rule check of S3 are
post-processing of the same eigenpairs and add no new computation.  Random
seeds are fixed.  No parameter is fitted to any of the reported results.
A single executed notebook,
\texttt{SM\_PRE2\_numerical\_verification.ipynb}, reproduces every number,
table, and figure in this Supplement (Sections S1--S3 correspond to its
Parts A--C); it also carries the independent \textsc{OpenPNM} cross-check of
the S2 percolation result shown in Fig.~\ref{fig:openpnm}. All computations use the same operator
convention throughout, so the rates reported in S1 and S3 are directly
comparable; S1 and S3 were regenerated with the geometric-mean gate of
Eq.~\eqref{eq:Kdisc} and the band width $w=0.10$, and the OSF record will
carry the corresponding commit; Figs.~\ref{fig:A}, \ref{fig:C} and \ref{fig:C2}, and Fig.~1 of the main text, are the outputs of the notebook \texttt{SM\_RichardsLimit\_S1\_S3\_regeneration}, Figs.~\ref{fig:thetac} and \ref{fig:B} of the percolation notebook, and Fig.~\ref{fig:openpnm} of the companion paper's pore-network notebook.

\appendix
\section{Glossary}
\label{app:glossary}
Terms used here in a specific technical sense.

\begin{description}
\item[Active band] The pore classes that are \emph{partially} filled,
$g_\eq(1-g_\eq)>0$.  Full pores have nothing to give and empty pores
nothing to give away, so only these exchange water and only these carry
the operator's modes.  In S1 it is $114$ of $160$ classes.

\item[Algebraic connectivity] See \emph{Fiedler value}.

\item[Eigenmode $\varphi_n$, rate $\lambda_n$] A shape of disturbance that
decays without changing shape, $\propto e^{-\lambda_n t}$, and its decay
rate.  Any disturbance is a superposition of these, so they are the
soil's natural relaxation patterns.

\item[Fiedler value $\lambda_1$] The smallest nonzero eigenvalue of a graph
Laplacian, also called the \emph{algebraic connectivity}.  It is positive
exactly when the network is in one connected piece and drops to zero when
it splits, so it detects connectivity rather than size.  Its eigenvector,
the \emph{Fiedler vector}, changes sign across the network's weakest cut
and so points at the bottleneck.

\item[Galerkin projection] Restricting an operator to a small subspace by
sandwiching it between the subspace's basis vectors.  Used in S3 to build
the two-domain model from the full operator without assuming its form: the
subspace is spanned by the two band indicators.

\item[Gates] The occupancy factors $g_\eq(1-g_\eq)$ in the conductance.
They close when a class is full or empty and peak at half filling; they
are the exclusion rule (donor full, receiver empty) written smoothly.

\item[Graph Laplacian] For a network with edge weights $W$ and degrees
$D$, the matrix $L=D-W$.  Symmetric, negative semidefinite, and with the
constants as its kernel when the network is connected, which is the
algebraic statement that nothing leaks out.

\item[Participation ratio $\mathrm{PR}$] The effective number of pore
classes a mode occupies, Eq.~\eqref{eq:ipr}.  $\mathrm{PR}=1$ means the
mode lives on a single class (localized); $\mathrm{PR}=N_{\rm band}$ means
it is spread over the whole band (extended).

\item[Spanning cluster] The connected body of water that reaches from one
face of the sample to the opposite one.  Water in isolated pockets is
still water content, but it cannot conduct, so only the spanning cluster
enters a permeability.

\item[Spectral gap] The distance $\lambda_1$ between the one non-decaying
mode (water mass) and the slowest decaying one.  It sets the relaxation
time $1/\lambda_1$, bounds the inverse operator, and closes at
percolation.

\item[Stokes permeability] The permeability obtained by actually solving
viscous flow on the pore network under an imposed pressure drop, as
opposed to inferring it from a formula.  Used in S2 as the independent
check on the spectral prediction.
\end{description}